\documentclass[jmp,graphicx,oneside,a4paper,superscriptaddress,showpacs,10pt]{revtex4-1}

\usepackage{graphicx}
\usepackage{color}
\usepackage[T1]{fontenc}
\usepackage[utf8]{inputenc}
\usepackage{amssymb}
\usepackage{amsthm}
\usepackage{bbm}
\usepackage{mathtools}
\usepackage{float}
\usepackage[hidelinks]{hyperref}

\usepackage[a4paper,margin=2cm,footskip=2cm]{geometry}
\newtheorem{thm}{Theorem}

\newtheorem{lem}[thm]{Lemma}
\newtheorem{prop}[thm]{Proposition}

\def\tr{\operatorname{tr}}
\def\idty{{\leavevmode\rm 1\mkern -5.4mu I}} %  unit operator

\def\Rl{{\mathbb R}}
\def\Cx{{\mathbb C}}

\def\Nl{{\mathbb N}}

\def\norm #1{\Vert #1\Vert}

\def\brAket#1#2{\langle #1\vert#2\rangle}
\def\brAAket#1#2#3{\langle#1\vert#2\vert#3\rangle}
\def\bra#1{{\langle#1\vert}}
\def\ket #1{\vert#1\rangle}
\def\ketbra #1#2{{\vert#1\rangle\langle#2\vert}}
\def\kettbra#1{\ketbra{#1}{#1}}

\def\tr{\mathop{\rm tr}\nolimits}
\def\abs#1{\vert#1\vert}

\def\Cases#1{\left\lbrace\begin{array}{cl}#1\end{array}\right.}
\def\Order{{\mathcal O}}

\let\veps\varepsilon

\def\inv{^{-1}}
\def\conv{\mathop{\mathrm{conv}}}
\def\re{\Re e}
\def\im{\Im m}

\def\CC{{\mathcal C}}
\def\HH{{\mathcal H}}

\def\wigner{{\mathcal W}}

\def\var#1#2{\Delta^2_{#1}(#2)} %#1=state, #2= POVM
\def\varpow#1#2#3{\Delta^#3_{#1}(#2)} %#1=state, #2= POVM #3=power

\def\ve{{\mathbf e}}
\def\vL{{\mathbf L}}
\def\eL{{\ve{\cdot}\vL}}
\def\vlambda{{\boldsymbol\lambda}}
\def\veta{{\boldsymbol\eta}}
\def\vx{{\mathbf x}}
\def\vk{{\mathbf k}}
\def\vkap{{\mathbf\kappa}}
\def\north{{\boldsymbol\nu}}
\def\Dmax{D_{\rm max}}
\def\Delmax{\Delta_{\rm max}}
\def\minix{_{\rm min}} % min-index

\def\essup{\mathop{\rm ess\,sup}}
\def\Eessup{E\!{\rm -}\!\essup}
\def\SU#1{{\rm SU(#1)}}

\begin{document}
	% Title
	\title{Uncertainty Relations for Angular Momentum}
	\author{Lars Dammeier}
	\email{lars.dammeier@itp.uni-hannover.de}
	\author{Ren\'e Schwonnek}
	\email{rene.schwonnek@itp.uni-hannover.de}
	\author{Reinhard F. Werner}
	\email{reinhard.werner@itp.uni-hannover.de}
	\affiliation{Institut f\"ur Theoretische Physik, Leibniz Universit\"at, Hannover, Germany}
	\date{\today}
\begin{abstract}
In this work we study various notions of uncertainty for angular momentum in the spin-$s$ representation of $SU(2)$.
		We characterize the ``uncertainty regions'' given by all vectors, whose components are specified by the variances of the
		three angular momentum components. A basic feature of this set is a lower bound for the sum of the three
		variances. We give a  method for obtaining optimal lower bounds for uncertainty regions for general operator
		triples, and evaluate these for small $s$.  Further lower bounds are derived by generalizing the technique by
		which Robertson obtained his state-dependent lower bound. These are optimal for large $s$, since they are
		saturated by states taken from the Holstein-Primakoff approximation. We show that, for all $s$, all variances are
		consistent with the so-called vector model, i.e., they can also be realized by a classical probability measure
		on a sphere of radius $\sqrt{s(s+1)}$.
		Entropic uncertainty relations can be discussed similarly, but are minimized by different states than those minimizing the variances for small $s$. For large $s$ the Maassen-Uffink bound becomes sharp and we explicitly describe the extremalizing states.

		Measurement uncertainty, as recently discussed by Busch, Lahti and Werner for position and momentum, is
		introduced and a generalized observable (POVM) which minimizes the worst case measurement uncertainty of all
		angular momentum components is explicitly determined, along with the minimal uncertainty. The output vectors for the optimal measurement all have the same length $r(s)$, where $r(s)/s\to1$ as $s\to\infty$.
	\end{abstract}

	\maketitle

% % % Introduction % % %

\section{Introduction}
The textbook literature on quantum mechanics seems to agree that the uncertainty relations for angular momentum, and indeed for any pair of quantum observables $A,B$ should be given by Robertson's \cite{Robertson} inequality
\begin{equation}\label{Robson}
	\var\rho A \var\rho B\geq \frac14(\tr\rho\,i[A,B])^2,
\end{equation}
valid for any density operator $\rho$, with $\var\rho A$ denoting the variance of the outcomes of a measurement of $A$ on the state $\rho$. Perhaps the main reason for the ubiquity of this relation in textbooks is that it is such a convenient intermediate step to the proof of uncertainty relations for position and momentum. In that case the right hand side is $\hbar^2/4$, independently of the state $\rho$. For any pair $A,B$ other than a canonical pair, however, the relation \eqref{Robson} makes a much weaker statement, requiring some prior information about the state. This begs the question: When and with what bounds it is true that
\begin{quote}
	{\it [Preparation Uncertainty:]\\ One cannot choose a state $\rho$ so that $\var\rho A$ and $\var\rho B$ simultaneously become arbitrarily small.\/}
\end{quote}
Robertson's relation supports no such conclusion, but on the other hand such a statement does hold in many situations. In fact, in a finite dimensional context it is true whenever $A$ and $B$ do not have a common eigenvector. In this paper we will provide optimal bounds for angular momentum components, establishing the methods for deriving optimal bounds in the general case along the way.

The second reason that \eqref{Robson} is unsatisfactory is that it addresses only the preparation side of uncertainty, in the sense loosely described in the italicized sentence above. However, there is always also a measurement aspect to uncertainty, for which Heisenberg's $\gamma$-ray microscope \cite{Heisenberg1927} is a paradigm. The error disturbance tradeoff would be stated as
\begin{quote}
	{\it [Error-Disturbance Uncertainty:]\\  An approximate measurement of $A$ of accuracy $\Delta A$ disturbs the system in such a way that from the post-measurement state and the measurement result for $A$ the distribution for observable $B$ can only be inferred with accuracy $\Delta B$, where $\Delta A$ and $\Delta B$ cannot be simultaneously arbitrarily small.\/}
\end{quote}
It is often easier to think of the whole experiment as a {\it joint measurement} of $A$ and $B$, and state relations of the kind:
\begin{quote}
	{\it [Measurement Uncertainty:]\\  For any measurement device with both an $A$- type and a $B$-type output, the marginals will have worst case error $\Delta A$, $\Delta B$ with respect to ideal measurements of $A$ and $B$, satisfying a tradeoff relation.}
\end{quote}
Again, generic observables and angular momenta satisfy non-trivial relations of this kind. Errors $\Delta A=\Delta B=0$ can occur only if $A$ and $B$ commute, i.e., under an even more stringent condition than for preparation uncertainty. In this paper we will provide some sharp measurement uncertainty relations for angular momentum, establishing along the way some methods which may be of interest in more general cases.

There is a third reason that one should not be satisfied with \eqref{Robson} with $A=L_1,\ B=L_2$: It involves only two of the three components of angular momentum. But
there is no reason tradeoff-relations as described above should not be stated for more than two observables. For angular momentum this seems especially natural. Moreover, it seems natural to state relations for all components simultaneously, i.e., not only for the three components along the axes of an arbitrarily chosen Cartesian reference frame, but for the angular momenta along arbitrary rotation axes, restoring the rotational symmetry of the problem.

Indeed the idea that uncertainty should involve just pairs of observables can be traced to Bohr's habit of expressing complementarity as a relation between ``opposite'' aspects, like `in vitro' and `in vivo' biology. This dualistic preference had more to do with his philosophy than with the actual structure of quantum mechanics.  Other founding fathers of quantum mechanics did not share this preference. As Wigner said in an interview \cite{Wigner63} in 1963:
\begin{quote}
	I always felt also that this duality is not solved and in this I may have been under Johnny’s [John von Neumann] influence, who said, ``Well, there are many things which do not commute and you can easily find three operators which do not commute.'' I also was very much under the influence of the spin where you have three variables which are entirely, so to speak, symmetric within themselves and clearly show that it isn’t two that are complementary; and I still don’t feel that this duality is a terribly significant and striking property of the concepts.
\end{quote}
In this spirit, an uncertainty relation for triples of canonical operators was recently proposed and proved \cite{Weigert}, and further generalizations are clearly possible. However, we will stick to angular momentum in this paper, and particularly seek to establish relations which do not break rotation invariance.

Our paper responds to an increasing interest in quantitative uncertainty relations. This interest is connected to an increasing number of experiments reaching the uncertainty dominated regime, so that that rather than qualitative or order-of-magnitude statements one is more interested in the precise location of the quantum limits. Measurement uncertainty was made rigorous in \cite{BLW1,BLW2,SchoRe}. There is also a controversial \cite{BLW2013a,BLW2014} state-dependent version \cite{Ozawa2004b}. That adequate uncertainty relations are sometimes better stated in terms of the sum of variances rather than their product has been noted repeatedly \cite{hof,huang,Maccone}. There has been some renewed interest also in the uncertainty between angular momentum and angular position \cite{Franke2004}, angular momentum of certain states \cite{rivas} and other non-standard complementary pairs \cite{ET}.

% % % Setting % % %
\subsection{Setting and notation}
In physics angular momentum appears as orbital or as spin angular momentum. Our theory applies to both, but it must be noted that the bounds obtained do depend on the quantum number for $\vL^2$. For example, there are states with vanishing orbital angular momentum uncertainties (precisely the rotation invariant ones, i.e., $s=0$) but none for a $s=\frac12$ degree of freedom. Therefore, one first has to decompose the given space into irreducible angular momentum components (integer or half integer), and then use the results for the appropriate $s$.
Hence we will consider throughout a system of spin $s$, with $s=1/2,1,3/2,...$ in its $(2s+1)$-dimensional Hilbert space $\HH=\Cx^{2s+1}$. The three angular momentum components will be denoted by $L_k$, $k=1,2,3$, and the component along a unit vector $\ve\in\Rl^3$ by $\eL$. We denote by $\ket m$ the eigenvectors of $L_3$, so that $L_3\ket m=m\ket m$ where $-s \leq m \leq s$ and, with $L_\pm=L_1\pm iL_2$,
\begin{equation}\label{Lpm}
L_\pm\ket m=\sqrt{s(s+1)-m(m\pm1)}\ket{m\pm1}.
\end{equation}
Rotation matrices, whether they are considered as elements of SO(3) or of SU(2), will typically be denoted by $R$, the corresponding matrix in the spin $s$ representation by $U_R$, and normalized Haar measure on SO(3) or SU(2) by $dR$. We will always set $\hbar=1$.
Observables are in general always allowed to be normalized positive operator valued measures (POVMs), with a typical letter $F$. For a self-adjoint operator $A$ the spectral measure is an observable in this sense, denoted by $E_A$. For a component of angular momentum, i.e., $A=\eL$ we write $E_\ve$ for short. For the unit vectors $\ve_k$ along the axes we further abbreviate this to $E_k$.

For the variance of the probability distribution obtained by measuring $F$ on a state (density operator) $\rho$ we write
\begin{equation}\label{variance}
\var\rho F=\min_\xi\int (x-\xi)^2\tr\rho F(dx).
\end{equation}
This minimum is taken over a quadratic expression in $\xi$, and it is attained when $\xi=\int x\tr\rho F(dx)$ is the mean value of the distribution. The most familiar case is that of the spectral measure for an operator $A$, in which case we abbreviate the variance by $\var\rho{A}$. Then the second moment $\int x^2\tr\rho F(dx)=A^2$ can also be expressed by $A$ and we get
\begin{equation}\label{vSpect}
\var\rho{A}=\var\rho{E_A}=\tr(\rho A^2)-\tr(\rho A)^2.
\end{equation}

We say that a unit vector $\ket\phi\in\HH$ is a {\it maximal weight vector}, if for some direction $\ve \in \Rl^3$ it satisfies $\eL \ket\phi=s\ket\phi$. This is the same as saying that, for some rotation $R\in$SU(2), $\ket\phi=U_R\ket{s}$ up to a phase. For such a vector we call $\rho=\kettbra\phi$ a {\it spin coherent state}. These states are candidates for states of minimal uncertainty.

% % % Summary % % %
\subsection{Summary of Main Results}
We now describe the structure of our paper and the main results.

\noindent{\it Sect.~\ref{sec:prep}: Preparation uncertainty}. ---\
The basic object of study is the variance $\var\rho{\eL}$ of the angular momentum in direction $\ve$ as a function of the unit vector $\ve$, especially properties which hold for an arbitrary state $\rho$.
After clarifying some general features and explicitly solving the two cases $s\leq1$ (Sect.~\ref{sec:lows}), we look at the traditional setting of just two components $L_1,L_2$. The set of uncertainty pairs $\bigl(\var\rho{L_1},\var\rho{L_2}\bigr)$ is studied, and  the fact that not both variances can be small is found to be well expressed by a lower bound not on the product but on the sum of the variances. We compute numerically (and exactly up to $s=3/2$) the best constants in
\begin{equation}\label{2compI}
\var\rho{L_1}+\var\rho{L_3}\geq c_2(s)
\end{equation}
and find that they asymptotically behave like $c_2(s)\sim s^{2/3}$ (Sect.~\ref{sec:twocomp}). For three components the uncertainty region is also studied in some detail. A prominent feature is again given by a linear bound \cite{hof}
\begin{equation}\label{3compI}
\var\rho{L_1}+\var\rho{L_2}+\var\rho{L_3}\geq s,
\end{equation}
which is very easy to prove (see Sect.~\ref{sec:Pbasics}, \eqref{triSum}).

Turning to features of the whole function $\ve\mapsto\var\rho\eL$, we show in Sect.~\ref{sec:powerMean} that for any $\rho$ there is at least one direction $\ve$ such that $\var\rho\eL\geq s/2$, i.e.,
\begin{equation}\label{icompI}
\max_\ve \var\rho\eL \geq \frac s2.
\end{equation}
This bound is optimal, since it is saturated by spin coherent states. We generalize from the maximum (seen as the $L^\infty$-norm) to all $L^p$-norms (Prop.~\ref{icomp}).

For large $s$ Eq.~\eqref{3compI} suggests the scaling $\sim 1/s$ (Sect.~\ref{sec:asymptotic}). Indeed, the triples $(\var\rho{L_1}/s,\var\rho{L_2}/s,\var\rho{L_3}/s)$ converge as $s\to\infty$ (Theorem~\ref{thm:largeS}). The lower bound on the limit set is obtained by a generalization of Robertson's method for proving \eqref{Robson} Sect.~\ref{sec:Robson}, which for finite $s$ is \eqref{RobsonNorm}
\begin{equation}\label{RobsonVars}
\var\rho{L_1}\Bigl(\var\rho{L_2}+\var\rho{L_3}\Bigr)\geq \frac14\Bigl(s(s+1)-\var\rho{L_1}-\var\rho{L_2}-\var\rho{L_3}\Bigr),
\end{equation}
where the components are ordered so that $\var\rho{L_1}\geq\var\rho{L_2}\geq\var\rho{L_3}$.
The upper bound in Theorem~\ref{thm:largeS} is provided by a family of states suggested by the Holstein-Primakoff approximation.

\noindent{\it Sect.~\ref{sec:vecMod}: Vector model and moment problems}. ---\
We revisit the so-called vector model of angular momentum, a classical model which is still found in some textbooks. We show that it can correctly portray the moments up to second order (i.e., means and variances) of the angular momentum observables, but fails on higher moments and, of course, on correlations.

\noindent{\it Sect.~\ref{sec:entropic}: Entropic uncertainty relations}. ---\
We discuss entropic uncertainty relations only very briefly. We point out that the criteria ``variance'' and ``entropy'' may disagree on which of two distributions is ``more sharply concentrated''. This effect is illustrated by the uncertainty diagrams for $s=1$. We show also that the general Maassen-Uffink bound \cite{muff} while suboptimal for $s=1$,  becomes sharp for $s\to\infty$, and determine a family of states saturating it.

\noindent{\it Sect.~\ref{sec:measure}: Measurement uncertainty}. ---\
We consider two measures for the deviation of an approximate observable from an ideal reference, called metric error and calibration error. We then discuss uncertainty relations for the joint measurement of {\it all} angular momentum components. The output of such an observable is an angular momentum  3-vector $\veta$, from which one can obtain a measurement of the $\ve$-component (for any unit vector $\ve$) simply by taking $\ve\cdot\veta$ as the output. Such a marginal observable can in turn be compared with the quantum observable $\eL$. The uncertainty relation in this case gives a lower bound on the  error in the worst case with respect to $\ve$. Our main result (Theorem~\ref{thm:mur}) is a determination of the optimal bound, and an observable saturating it. It turns out that the optimal observable is covariant with respect to rotations, and this implies that it simultaneously minimizes the maximal metric error and the maximal calibration error. All the output vectors have the same length $r\minix(s)$, which depends in a non-trivial way on $s$ but is close to $s$ for large $s$.

% % % Preparation Uncertainty % % %

\section{Preparation Uncertainty}\label{sec:prep}
In this section we consider the preparation uncertainty, i.e., a property of a given state $\rho$. For every unit vector $\ve\in\Rl^3$ we can form the variance of the angular momentum component $\eL$, and hence study the function
\begin{equation}\label{vE}
	v(\ve)=\var\rho{\eL}=\tr\Bigr(\rho(\eL)^2\Bigl)-\Bigr(\tr\rho\,\eL\Bigl)^2
\end{equation}
on the unit sphere. For the purposes of this section, this function summarizes all the uncertainty properties of the state $\rho$, and all results in this section are statements about properties of this function, which are valid for all $\rho$. To visualize the function $v$, we can use a three-dimensional radial plot, i.e., the surface containing all vectors $v(\ve)\ve$, as $\ve$ runs over all unit vectors.  A typical radial plot is shown in FIG.~\ref{fig:radialPU}. Often we are also interested in the components with respect to some Cartesian reference frame. In this case the best visualization is an {\it uncertainty diagram}, which represents the possible pairs/triples etc.\ of variances in the same state. In our case this will be the set of pairs $(v(\ve_1),v(\ve_2))$, or triples $(v(\ve_1),v(\ve_2),v(\ve_3))$. The diagrams for $s=1$ are shown in FIG.~\ref{fig:spin1mon}. In this diagram it can be seen that the uncertainty region is not convex in general. Since we are only interested in lower bounds, we therefore always take the {\it monotone closure} of the uncertainty region, i.e., we also include with every point the whole quadrant/octant of points in which one or more of the coordinates increase. This is described in more detail in Sect.~\ref{sec:method}.
\begin{figure}[ht]
	\includegraphics[width=0.5\textwidth]{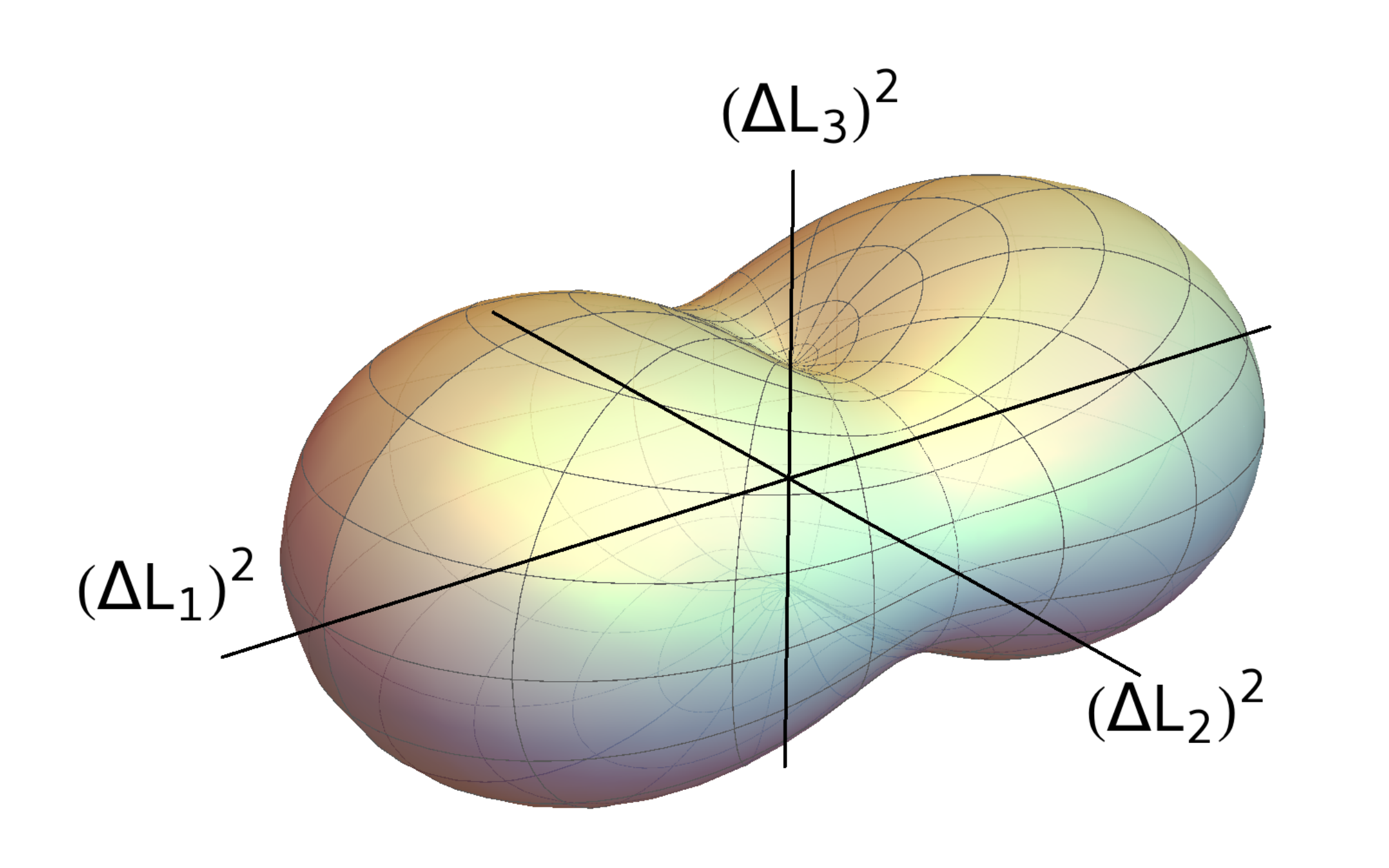}
	\caption{The function $v(\mathbf{e})$ from equation \eqref{vLambda}, where $\Lambda$ is diagonal. }
	\label{fig:radialPU}
\end{figure}

It turns out that after a rotation to suitable principal axes (which has already been carried out in FIG.~\ref{fig:radialPU}), the function $v$ depends only on three real parameters $\mu_1,\mu_2,\mu_3$. To see this we introduce the $3\times3$-matrix $\Lambda=\Lambda(\rho)$ by
\begin{eqnarray}\label{vLambda}
	v(\ve)&=&\sum_{jk}e_je_k\ \Lambda_{jk}(\rho) \qquad \mbox{with}\\
	\Lambda_{jk}(\rho)&=& \re\tr(\rho L_jL_k)-\lambda_j\lambda_k \label{Lambda}\\
	\lambda_{j}(\rho)&=&\tr(\rho L_j).      \label{lambda}
\end{eqnarray}
Since the $L_k$ transform as a vector operator (i.e., with respect to the spin-1 representation of SU(2)) we see that by choice of an appropriate coordinate basis in $\Rl^3$ we can diagonalize $\Lambda$, i.e., we can choose
$\Lambda_{jk}(\rho)=\mu_j\delta_{jk}$ with $\mu_1\geq\mu_2\geq\mu_3\geq0$.

The eigenvalues $(\mu_1,\mu_2,\mu_3)$ of any matrix $\Lambda(\rho)$ are also a possible triple of variances, namely for a suitably rotated state. In fact, we can find the uncertainty triples for {\it all} rotated versions of $\rho$ quite easily: When $R$ is the rotation matrix taking the eigenbasis of $\Lambda$ to the basis $\ve_1,\ve_2,\ve_3$ under consideration, then
\begin{equation}\label{Birkhoff}
	v(\ve_j)=\sum_kR_{jk}^2\mu_k.
\end{equation}
Now the squared rotation matrix is doubly stochastic, so by Birkhoff's theorem it is a convex combination of permutation matrices. We therefore find the variance triple in basis $\ve$ in the convex hull of the six points, arising from the triple of $\mu_k$ by permutation. These six points lie in a plane orthogonal to the vector $(1,1,1)$, so they form a hexagon (see FIG.~\ref{fig:hexagon}), which degenerates into a triangle if two of the $\mu_k$ are equal. One can easily check that the full hexagon is attained by squared rotation matrices.
\begin{figure}[ht]
	\includegraphics[width=0.5\textwidth]{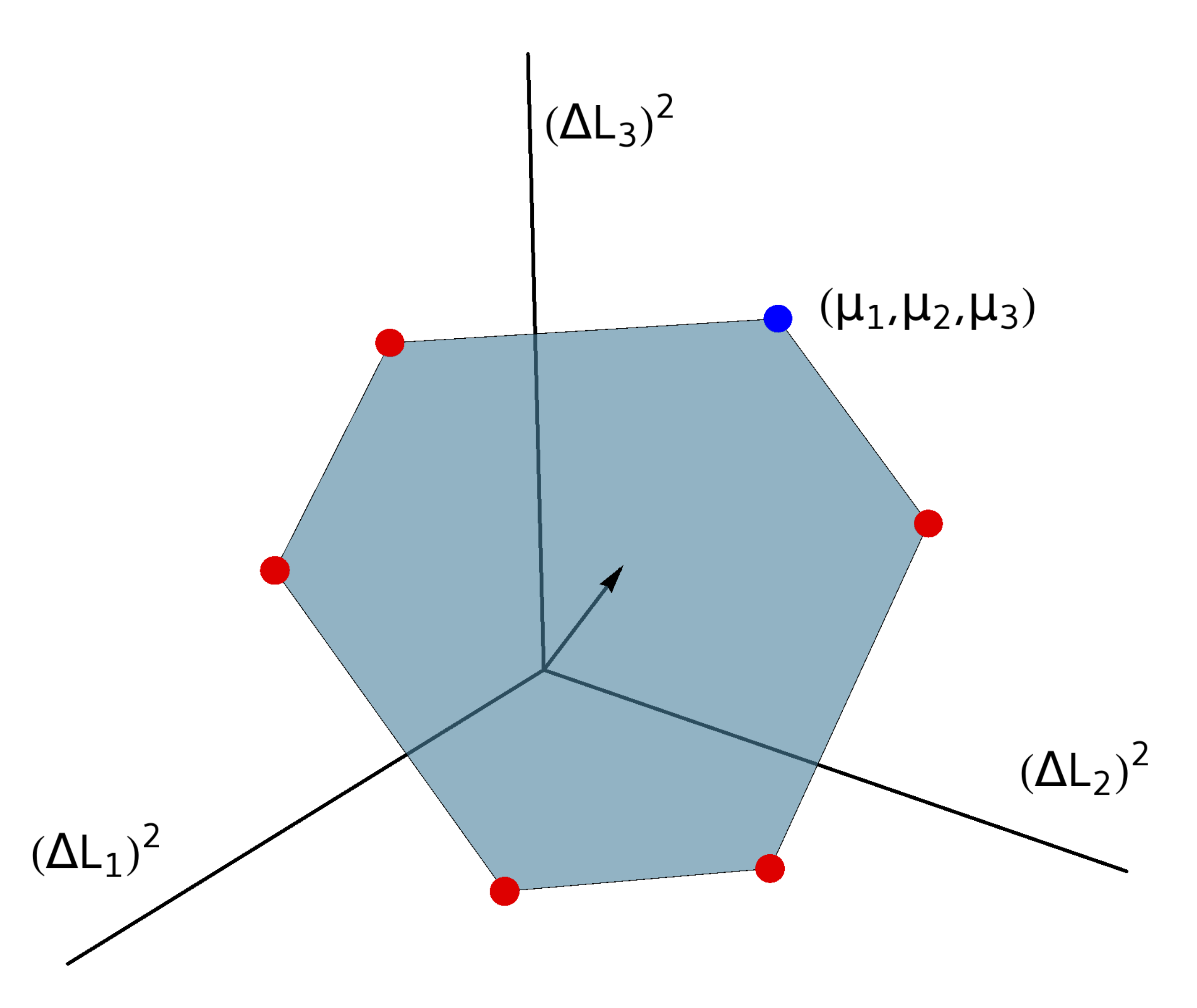}
	\caption{The orbit of a point under permutations of the coordinates, and its convex hull.  }
	\label{fig:hexagon}
\end{figure}

% % % Basic Bounds % % %
\subsection{Basic bounds}\label{sec:Pbasics}
For an $L_3$-eigenstate $\rho=\kettbra m$ we find $\lambda=(0,0,m)$ and
\begin{equation}\label{LambdaEigen}
\Lambda(\kettbra m)=\frac12(s(s+1)-m^2)\begin{pmatrix}1& 0 & 0 \\ 0 & 1 & 0 \\ 0 & 0&0\end{pmatrix}
\geq\frac s2\,\begin{pmatrix}1& 0 & 0 \\ 0 & 1 & 0 \\ 0 & 0&0\end{pmatrix},
\end{equation}
which re-inforces that, for eigenstates, the variances in all directions are smallest for the maximal weight $m=s$. Maximal variance is attained for an equal weight mixture $\rho=1/2 \bigl(\kettbra s+\kettbra{-s}\bigr)$,
which has $L_3$-variance $s^2$. Hence, for all $\ve$:
\begin{equation}\label{basicc}
	0\leq v(\ve)\leq s^2.
\end{equation}

The average $\overline{v(\ve)}$ over $\ve$ with respect to the surface measure on the sphere, or equivalently the average of $v(R\ve)$ over Haar-random rotations $R$, is readily computed from \eqref{vLambda},
since the average of $e_je_k$ over the unit sphere is just $\delta_{jk}/3$. Therefore, from \eqref{vLambda},
\begin{equation}\label{vaverage}
	\overline{v(\ve)}=\frac13\tr\Lambda(\rho)= \frac13\bigl(s(s+1)-\abs\vlambda^2\bigr)\geq \frac13\bigl(s(s+1)-s^2\bigr)\geq \frac s3.
\end{equation}
In the same simple way we can get an inequality for the variances along the three coordinate directions  of a Cartesian coordinate system:
\begin{eqnarray}\label{triSum}
	\sum_{k=1}^3v(\ve_k)=\tr\Lambda(\rho)\geq s.
\end{eqnarray}
In both cases equality holds precisely for $\abs\vlambda=s$, i.e., if $\rho$ is an eigenstate of one of the operators $\eL$ for the maximum eigenvalue $m=s$.

% % % Special features  % % %
\subsection{Special features for $s=1/2$ and $s=1$}\label{sec:lows}
For $s=1/2$, it happens that $L_j$ and $L_k$ (i.e., up to a factor the Pauli matrices) anticommute for different $j,k$, so that
\begin{equation}\label{LambdaHalf}
	s=\frac12:\qquad \Lambda_{jk}(\rho)=\frac14\delta_{jk}-\lambda_j\lambda_k.
\end{equation}
The eigenvalues are $(\mu_1,\mu_2,\mu_3)=1/4 (1,1,1-4 \abs{\lambda}^2)$. Of course, pure states are characterized by $\abs\vlambda=\frac 1 2$. The uncertainties are $1/4-\lambda_j^2$, and so the uncertainty region is described by a triangle.

\begin{figure}[t]
	\includegraphics[width=0.4\textwidth]{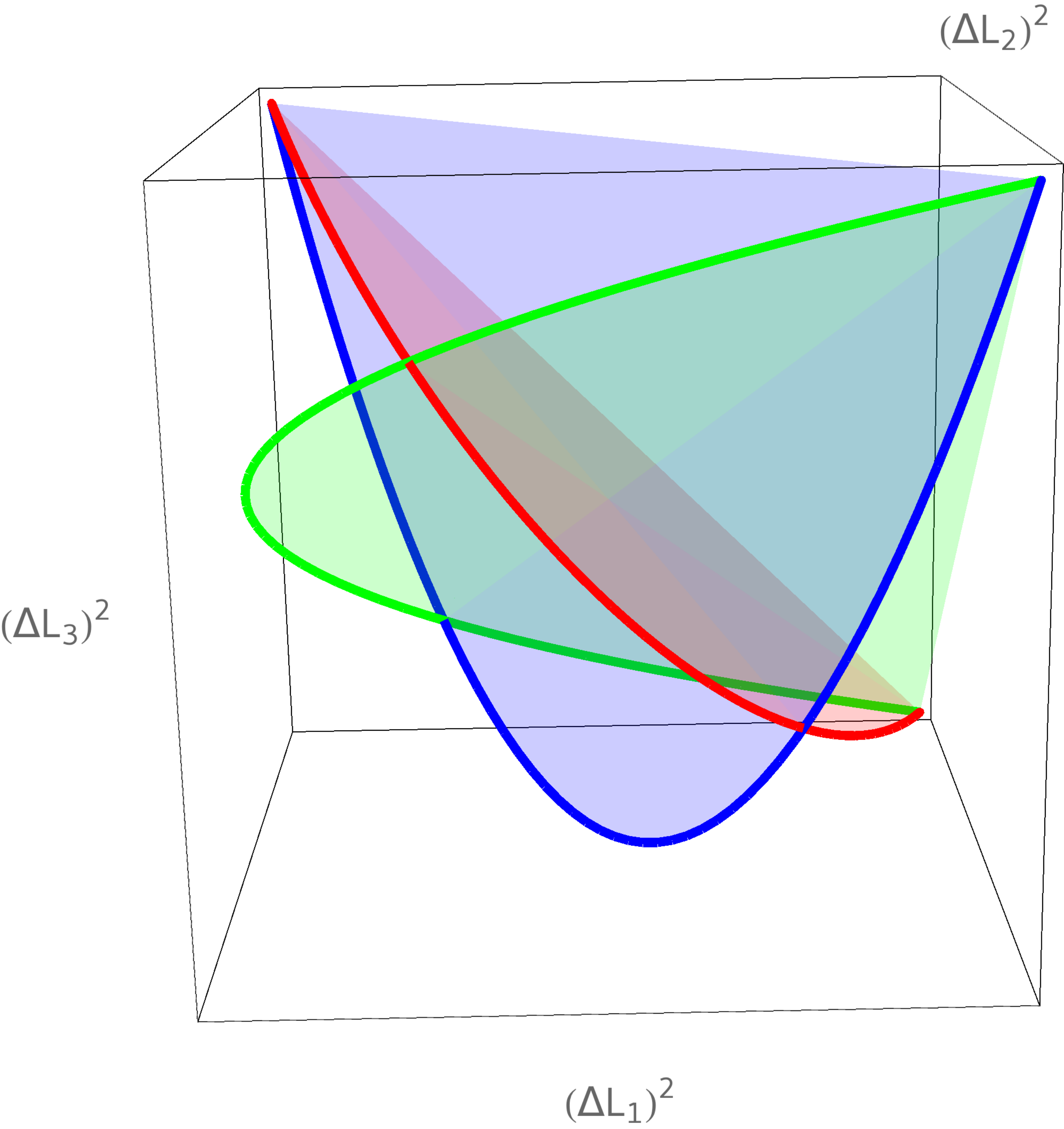}
	\includegraphics[width=0.4\textwidth]{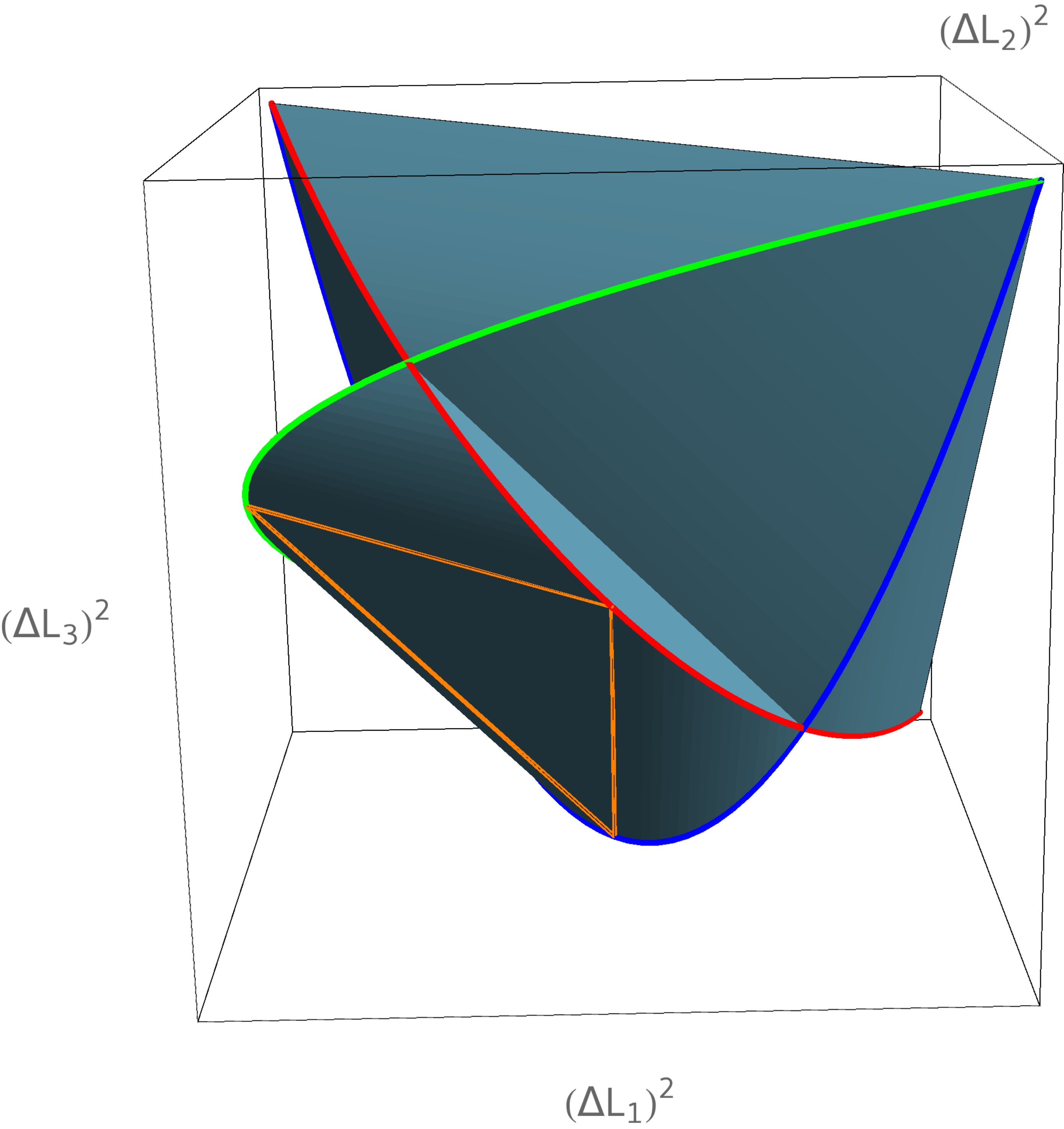}
	\caption{Pure state uncertainties for $s=1$: Left panel: Parabola after \eqref{spin1psi} with its permutations. Adding to every point the hexagon generated by its permutation orbit generates the solid shown in the right panel.
		This is precisely the set of all variance triples for pure states. Its monotone closure is shown in FIG.~\ref{fig:spin1mon}.  }
	\label{fig:parabolics}
\end{figure}

The case $s=1$ is still special because the $3+6$ operators $L_k$ and $(L_jL_k+L_kL_j)/2$ form a basis of the operators on $\Cx^3$. Therefore, $\rho$ can be reconstructed from $(\vlambda,\Lambda)$ and, in particular, the set of pure $\rho$ can be characterized in terms of conditions on the eigenvalues $\mu_k$ and $\lambda_k$. In order to analyze these conditions, let us  take the representation of the group SU(2) by real orthogonal matrices. Now consider a vector $\psi\in\Cx^3$, which we can split into $\psi=\re\psi+i\,\im\psi$ with $\re\psi,\im\psi\in\Rl^3$. Note that the real continuous function $t\mapsto\left\langle\re\bigl(e^{it}\psi\bigr),\im\bigl(e^{it}\psi\bigr)\right\rangle$ changes sign between $t=0$ and $t=\pi/2$, so we can choose a complex phase for $\psi$ to make $\re\psi$ and $\im\psi$ orthogonal. Moreover, we can apply a rotation, so that $\re\psi$ and $\im\psi$ are along the first two coordinate axes. Hence, up to a rotation, we have
\begin{equation}\label{spin1psi}
\psi=\begin{pmatrix}\cos t \\ i\sin t \\ 0\end{pmatrix} \qquad
\vlambda=\begin{pmatrix}0 \\ 0 \\ 2\sin t\cos t\end{pmatrix}\qquad
\Lambda=\begin{pmatrix}\tau & 0 & 0 \\ 0 & 1-\tau & 0  \\ 0 & 0 & 1-4\tau(1-\tau)\end{pmatrix},
\end{equation}
where $t\in[-\pi,\pi]$ or $\tau=(\sin t)^2\in[0,1]$. In the three-component diagram the curve parameterized by $\tau$ is a parabola lying in a diagonal plane. This parabola, and the two copies arising by coordinate permutation are shown in FIG.~\ref{fig:parabolics}, as well as the body of uncertainty triples of all pure states, which arises by adding to each point on the parabola the hexagon formed by its permutation orbit. A paper cutout model of this solid is provided as a supplement \cite{papermodel}.
\begin{figure}[t]
	\includegraphics[width=0.4\textwidth]{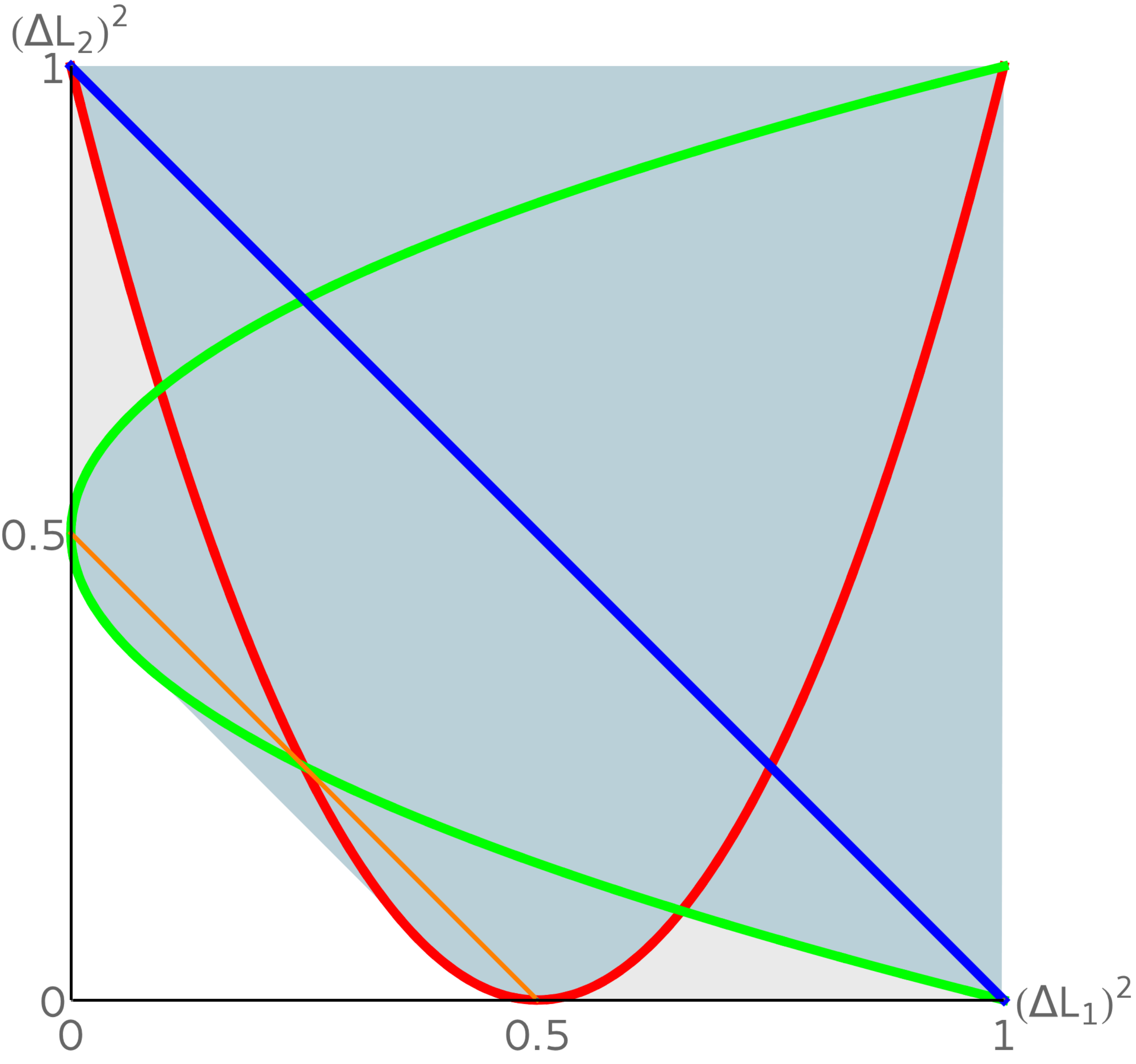}
	\includegraphics[width=0.4\textwidth]{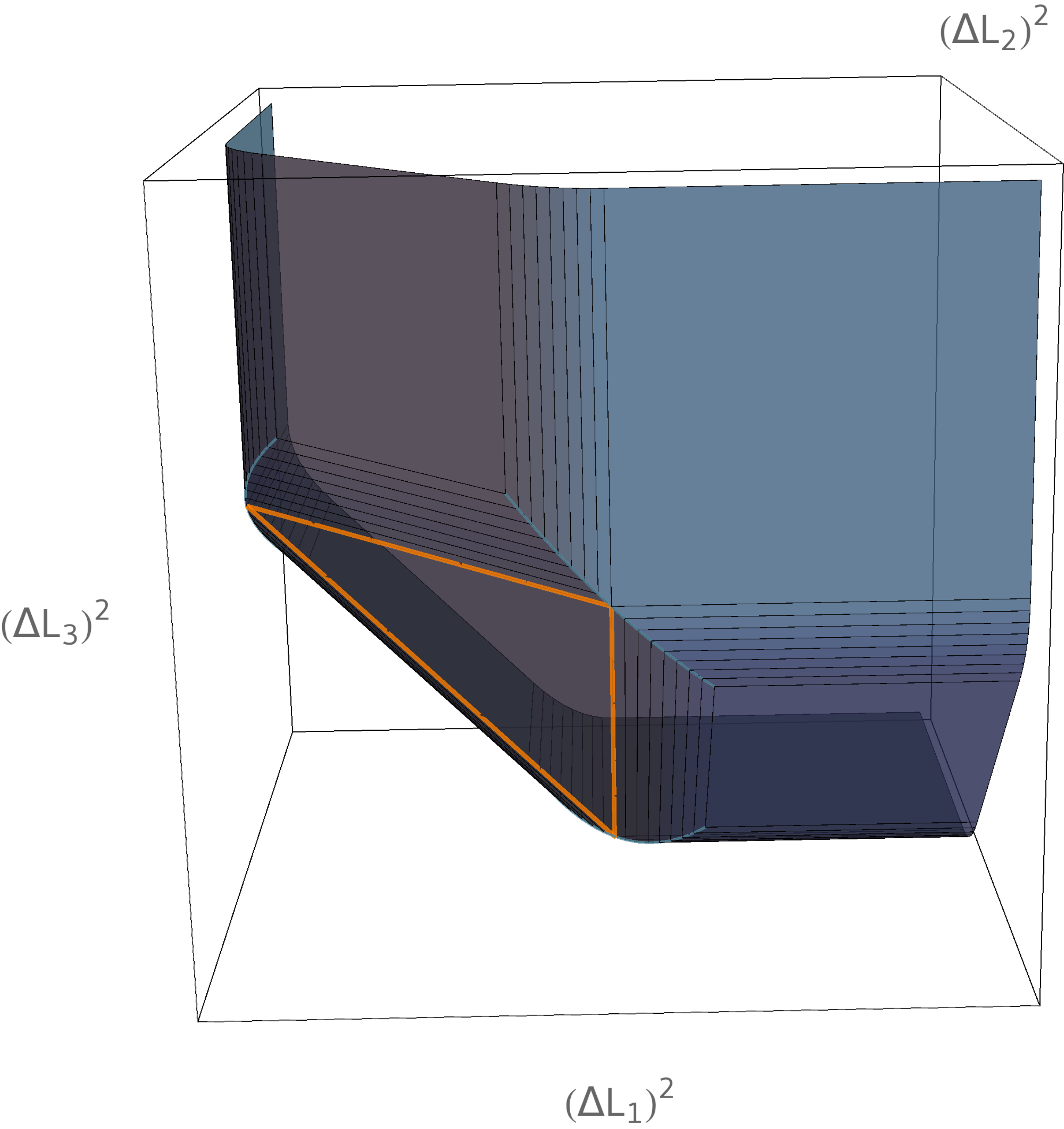}
	\caption{The monotone closure of the uncertainty region of spin 1. Since this turns out to be convex, it is equal also to the lower convex hull (see text).
Left panel: for two orthogonal spin components. The light gray area belongs to the monotone closure, but these points cannot be realized as uncertainty pairs.
The parabolas outline the shape (compare also FIG.~\ref{fig:parabolics}) the orange lines correspond to coherent states. Right panel: the analogue for three spin components. Projecting this body onto one coordinate plane gives the shape shown in the left panel.  }
	\label{fig:spin1mon}
\end{figure}

% % % Algorithm % % %
\subsection{General minimization method}\label{sec:method}
Consider now, a little more generally, any collection of hermitian operators $A_1,\ldots,A_n$. We can then form, for any state $\rho$, the variance $n$-tuple $(\var\rho{A_1},\ldots,\var\rho{A_n})$, and ask which region $\Omega$ in $\Rl^n$ is filled, when $\rho$ runs over the whole state space. We call $\Omega$ the {\it uncertainty region} of the operator tuple. Typically, this is not a convex set, because $\var\rho A$ contains a term quadratic in $\rho$, which consequently does not respect convex mixtures. $\Omega$ will be simply connected (as a continuous image of the state space), but beyond that there are few general facts. It can happen that starting from a point in the uncertainty region we can leave the region by increasing one of the coordinates, i.e., the region encodes upper bounds on variances as well as lower bounds. This is clearly not relevant to the theme of uncertainty relations, where we ask for universal lower bounds only. We can therefore consider the {\it monotone closure} of the uncertainty region, by including all points with larger uncertainties, i.e.,
\begin{equation}\label{monotoneClos}
	\Omega^+=\bigl\{(x_1,\ldots,x_n)\bigm|\exists \rho \ \forall \alpha\ x_\alpha\geq \var\rho {A_\alpha}\bigr\}.
\end{equation}
This is still not necessarily a convex set. We will denote the convex hull of $\Omega^+$ by $\Omega^\vee$ and call it the {\it lower convex hull} of $\Omega$ (see FIG.~\ref{fig:spin1mon}).
It is this set which has an efficient characterization. Indeed, as a closed convex set it is the intersection of all half spaces containing it, and the monotonicity condition restricts these half spaces to those whose normal vector $w$ has all components non-negative. In other words,
\begin{eqnarray}\label{OmVee}
	\Omega^\vee&=&\bigl\{(x_1,\ldots,x_n)\bigm|\forall w_\alpha\geq0:\  \sum_\alpha w_\alpha x_\alpha\geq m(w_1,\ldots,w_n)\bigr\},\\
	m(w_1,\ldots,w_n)&=&\inf_\rho\ \sum_\alpha w_\alpha\var\rho{A_\alpha} \nonumber\\
	&=&\inf_\rho\ \inf_{a_1,\ldots,a_n}\ \sum_\alpha w_\alpha \tr\rho\bigl(A_\alpha-a_\alpha\idty\bigr)^2
\end{eqnarray}
In this double infimum we can exchange the order, leading to two kinds of operations: With fixed $\rho$ the global minimum over the $a_\alpha$ is obviously at $a_\alpha=\tr\rho A_\alpha$. On the other hand, with fixed $a_\alpha$ the global minimum over $\rho$ is computed by finding the ground state of the positive semidefinite operator $H(a)=\sum_\alpha w_\alpha\bigl(A_\alpha-a_\alpha\idty\bigr)^2$. An efficient algorithm is therefore obtained by alternating between these two steps. The upper estimates on $m$ obtained in this way are non-increasing and in practice converge quite well, and independently of the starting value. However, we do not have a theorem to this effect. An analytic consequence of this algorithm (independent of convergence) is that we can restrict the infimum to pure states, since this is sufficient to get the ground state energies.

The algorithm is then run for a suitable set of tuples $(w_1,\ldots,w_n)$, so that for each run, one obtains a tangent plane to $\Omega^\vee$ but also the state $\rho$ and with it, the tuple of variances in $\Omega$.
We illustrate the results in FIG.~\ref{fig:spin1mon} for the case of spin 1, and the operator tuples $(L_1,L_2)$ and $(L_1,L_2,L_3)$, respectively. For low spin these diagrams can be determined analytically (see the next subsection). The most prominent feature of two-component diagram is the symmetric linear bound, which depends on $s$ and is determined in subsection~\ref{sec:twocomp}.

% % % Two component % % %
\subsection{The linear two-component bound}\label{sec:twocomp}
For every $s$, let $c_2(s)$ be the best constant in the inequality
\begin{equation}\label{2comp}
	\var\rho{L_1}+\var\rho{L_3}\geq c_2(s).
\end{equation}
For $s=1/2,1$ it is readily computed from the eigenvalues of $\Lambda$ given in Sect.~\ref{sec:lows}. For arbitrary $s$ we can use a slightly simplified version of the variational principle \eqref{OmVee}.
We have $w_1=w_3=1$, and can assume that $a_1=0$ by rotation invariance around the 2-axis. Thus
\begin{equation}\label{c2var}
	c_2(s)=\inf_\phi\inf_a \bigr\langle\phi|\ s(s+1)\idty-L_2^2-2a L_3+a^2\idty\bigm|\phi\bigr\rangle,
\end{equation}
where the first infimum runs over all pure states (for fixed $a$ a ground state problem) and $a$ over the reals (for fixed $\phi$ the expectation value of $L_3$). One notes that in this operator only matrix elements with even $m-m'$ are non-zero, so the problem can be further reduced. For up to $s=3/2$ it effectively leads to two-dimensional ground state problems. In this way (resp.\ by  using the results of Sect.~\ref{sec:lows}) we get
\begin{eqnarray}
	c_2(\textstyle\frac12)&=& \textstyle\frac14\\
	c_2(1)&=& \textstyle\frac{7}{16} \\
	c_2(3/2)&=&\frac94+\gamma^2-\sqrt{4 \gamma ^2+2 \gamma +1}\approx0.600933\\
	&&\mbox{where}\ \gamma=\cos(\pi/9)\nonumber.
\end{eqnarray}
Note that the bound $c_2(1)$ was already obtained in \cite{hof}.
It is readily seen numerically that $c_2(s)$ increases with $s$, but sub-linearly. This means that if we scale the diagram of $\Omega^\vee$ (see FIG.~\ref{fig:spin1mon}, right) by a factor $1/s$ so that the bottom triangle described by \eqref{triSum} stays fixed, the two-component inequality excludes an asymptotically small prism around the axes. FIG.~\ref{fig:twocomps} shows the asymptotic behavior of $c_2$ in a log-log plot, which suggests that
\begin{equation}\label{twocomploglog}
	c_2(s) \approx 0.569524 \ s^{2/3} \; \text{for large s}.
\end{equation}
\begin{figure}[h]
	\includegraphics[width=0.48\textwidth]{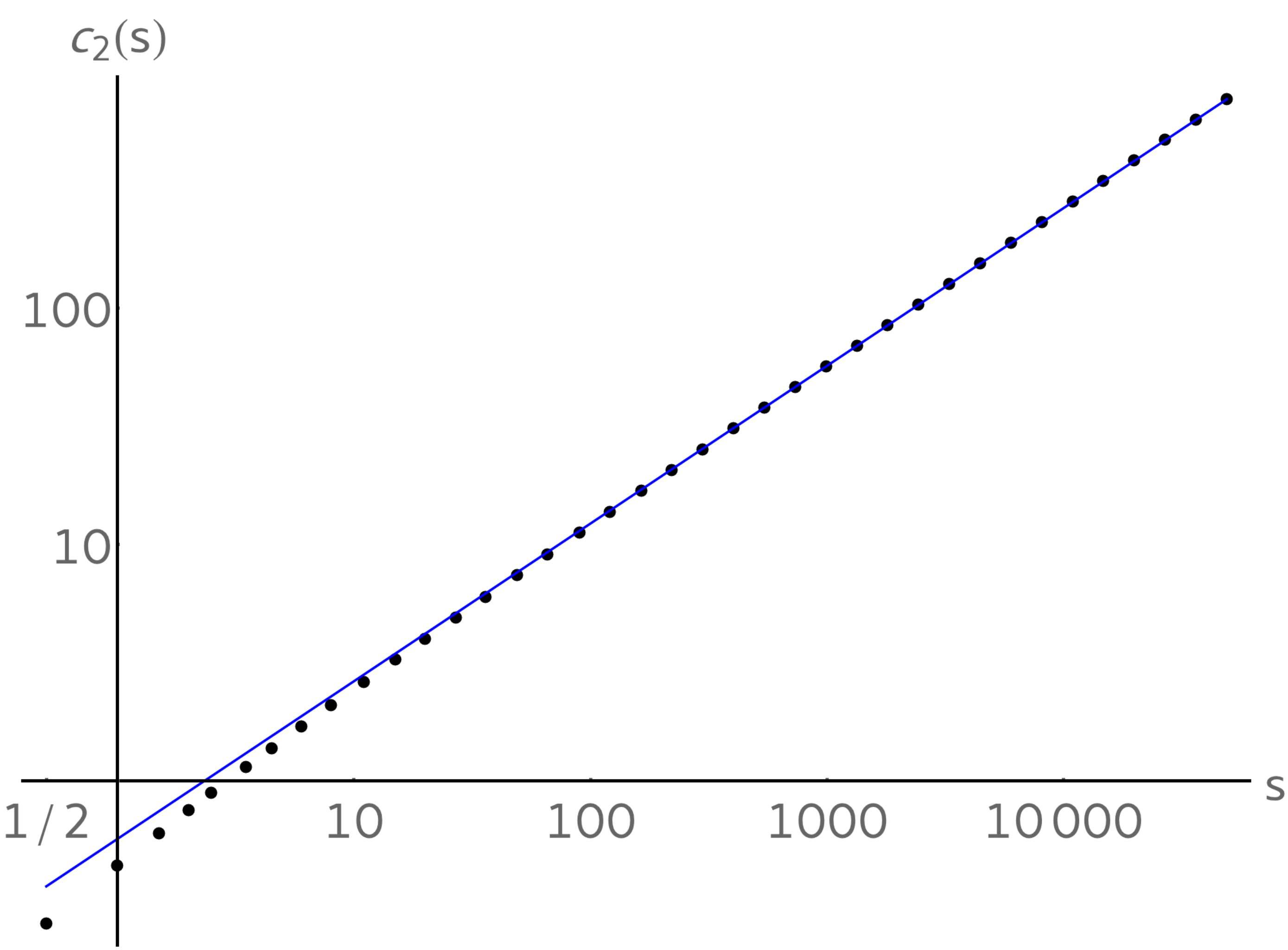}
	\includegraphics[width=0.48\textwidth]{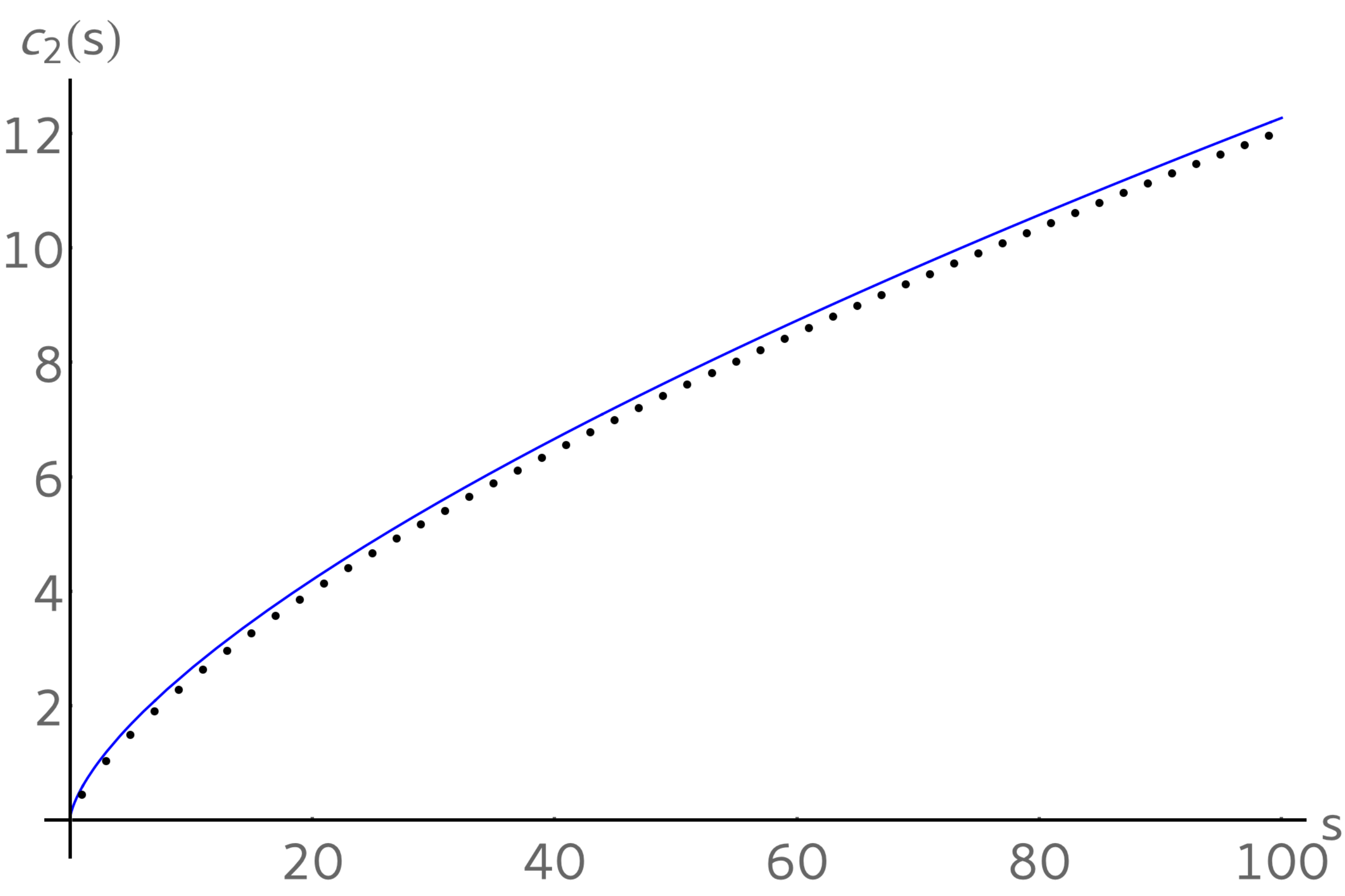}
	\caption{Left: Log-log plot of the numerical calculations of the two component bound $c_2(s)$ in black and the arising fit \eqref{twocomploglog} in blue. Right: Numerics and the fit for small s.}
	\label{fig:twocomps}
\end{figure}

% % % Power mean % % %
\subsection{Power mean and maximal uncertainty}\label{sec:powerMean}
A natural way to characterize states with small variance is to look for the maximum of the variance function $v(\ve)$ defined in \eqref{vLambda}. An uncertainty relation would then put a lower bound $c(s)$ on this maximum. In other words we would like to prove the following statement: {\it For every state $\rho$ there is some direction $e$ such that $\var\rho{E_e}$ is larger than $c$.}\/ By considering coherent states we can immediately see that $c(s)\leq s/2$. The following proposition shows that coherent states in fact have minimal variance in terms of this criterion, and we even have equality.

Such a result can be seen as one end of a one-parameter family of criteria, of which \eqref{vaverage} is the other end: We can judge the ``size'' of the function $v$ by its ${\cal L}^p$ norm, of which the maximum is the special case $p=\infty$, and the mean the case $p=1$. We therefore formulate a proposition to cover all these cases.

\begin{prop}\label{icomp}
	For every $s\in\Nl/2$ and every $p\in[1,\infty]$ there is a constant $c(p,s)$ such that,
	for every density operator $\rho$ in the spin $s$ representation,
	\begin{equation}\label{cprep}
		\left\Vert v\right\Vert_p=\left(\!\int\frac{de}{4\pi}\ \varpow\rho{\eL}p\right)^{\textstyle\frac1p}\geq c(p,s)
	\end{equation}
	with equality whenever $\rho$ is a spin coherent state. For $p<\infty$ these are the only states with equality. For $p=\infty$ equality holds also for mixtures $\rho=p_+\kettbra{+s}+p_-\kettbra{-s}$, and rotations thereof, provided that $p_+p_-\geq1/(8s)$.\\
	The constant is
	\begin{equation}\label{constpnorm}
		c(p,s)=\frac s2\ \left(\frac{\sqrt{\pi}\ \Gamma(p+1)}{2\,\Gamma \left(p+\frac{3}{2}\right)}\right)^{\textstyle\frac1p}\ ,
	\end{equation}
	with special values $c(1,s)=s/3$, $c(2,s)=s \sqrt{8/15}\approx0.73s$, $c(\infty,s)=s/2$.
\end{prop}

\begin{proof}
	Let $\lambda_j=\tr\rho L_j$ be the vector of expectation values, and consider the set of density operators $\rho^\beta$ arising from $\rho$ by rotation $R_\beta$ around the vector $\lambda$ by the angle $\beta$. For each $\rho^\beta$, we call the variance function $v_\beta(\mathbf{e})=v(R_\beta \mathbf{e})$.
	By averaging over $\beta$ we find a state $\overline\rho$, with variance function
	\begin{equation}\label{avv}
		\overline{v}(\mathbf{e})=\frac1{2\pi}\int d\beta\ v_\beta(\mathbf{e}),
	\end{equation}
	where we used, crucially, that all $\rho^\beta$ and $\overline\rho$  have the same expectations $\lambda_j$. By the triangle inequality for the $p$-norm, we have $\norm{\overline{v}}_p\leq\norm v_p$. Hence we can restrict the search for the $\rho$ with minimal $\norm v_p$ to those which are rotation invariant around some axis, say the 3-axis.

	Such a state can be jointly diagonalized with $L_3$, and is hence of the form $\rho=\sum_m p_m\kettbra m$. Then $\lambda=(0,0,\sum_mp_m m)$, and $\Lambda(\rho)$ is diagonal with
	\begin{eqnarray}\label{Lambdadiagonal1}
		\Lambda_{11}&=& \Lambda_{22}=\frac12\bigl(s(s+1)-\sum_mp_mm^2\bigr)\geq \frac12(s^2+s-s^2)=\frac s2 \\
		\Lambda_{33}&=& \sum_mp_mm^2-(\sum_mp_mm)^2\geq0 \label{Lambdadiagonal2}\\
		v(\mathbf{e})&=&(e_1^2+e_2^2)\Lambda_{11}+e_3^2\Lambda_{33}. \label{Lambdadiagonal3}
	\end{eqnarray}
	The last equation shows that the function $v$ becomes pointwise smaller (and hence smaller in $p$-norm) if we decrease some $\Lambda_{ii}$. That is, we have to go to the minimum on both $\Lambda_{11}$ and $\Lambda_{33}$. The minimum in \eqref{Lambdadiagonal1} is attained precisely when $p_m\neq0$ only for $m=\pm s$.
	Then minimality in \eqref{Lambdadiagonal1} forces $\rho$ to be a spin coherent state. For $p=\infty$ the norm only sees the maximum, so the pointwise minimum need not be chosen, and we may allow $0\leq\Lambda_{33}\leq\Lambda_{11}$ without changing the maximum. The latter inequality translates to the one given in Prop.~\ref{icomp}.

	The concrete constants follow easily by integrating the $p^{\rm th}$ power of \eqref{Lambdadiagonal3} with $\Lambda_{11}=s/2$ and $\Lambda_{33}=0$ with respect to the normalized surface measure on the sphere, i.e.,
	\begin{equation}\label{pintconst}
	c(p,s):=\frac s2\ \left(\int_0^\pi\frac{\sin\theta\,d\theta}{2}\ (\sin{\theta})^{2p}\right)^{\textstyle\frac1p}.
	\end{equation}

\end{proof}

% % % Robertson % % %
\subsection{Robertson's technique: a generalization}\label{sec:Robson}
We have criticized the Robertson inequality \eqref{Robson} for not giving a state independent bound. However, with only little effort it can be used to derive such a bound. Indeed, abbreviating $v_j=\var\rho{L_j}$, and $\lambda_j=\tr\rho L_j$ we can add the three inequalities of the form $v_1v_2\geq\lambda_3^2/4$ and use that $\sum_j(v_j+\lambda_j^2)=s(s+1)$ to obtain
\begin{equation}\label{Robson3}
	v_1v_2+v_2v_3+v_3v_1 \geq\frac14\bigl(s(s+1)-(v_1+v_2+v_3)\bigr).
\end{equation}
Clearly, this no longer allows $v_1=v_2=0$, since $v_3\leq s^2$. The set of variance triples satisfying this is shown in FIG.~\ref{fig:roblem}. Comparison with FIG.\ref{fig:spin1mon} readily shows that this bound is not optimal. However, we can generalize Robertson's technique from two to three components rather than extend his two component result in this trivial way. The basis of the technique is to utilize the observation that for any finite collection of operators $X_j$ (not necessarily hermitian or normal) the matrix $m_{jk}=\tr\rho X_j^*X_k$ is positive definite, which is the same as saying that for any complex linear combination $X=\sum_ja_jX_j$ the expectation of $X^*X$ must be positive. In order to get Robertson's inequality for $L_1,L_2$ this idea is applied to the three operators $\idty,L_1,L_2$. In fact, this leads to Schr\"odinger's improvement of the inequality \cite{SchroRob} which also contains the square of the covariance matrix element $\Lambda_{12}(\rho)^2$ on the right hand side.

We will apply the method to the four operators $\idty,L_1,L_2,L_3$. In order to simplify the expressions, however, we will not look for variances and the off-diagonal elements of $\Lambda(\rho)$, but for inequalities involving the eigenvalues $\mu_j$. As discussed at the beginning of this section, this will contain all the information needed. In other words, we will take the matrix $\Lambda(\rho)$ as the diagonal matrix with entries $\mu_1\geq\mu_2\geq\mu_2$. The matrix which then needs to be positive is
\begin{equation}\label{M4}
M=\left(\begin{array}{cccc}
1&\lambda_1&\lambda_2&\lambda_3\\[3pt]
\lambda_1&\mu_1+\lambda_1^2 &\lambda_1\lambda_2+i\lambda_3/2&\lambda_1\lambda_3-i\lambda_2/2\\[3pt]
\lambda_2&\lambda_1\lambda_2-i\lambda_3/2&\mu_2+\lambda_2^2&\lambda_2\lambda_3+i\lambda_1/2\\[3pt]
\lambda_3&\lambda_1\lambda_3+i\lambda_2/2&\lambda_2\lambda_3-i\lambda_1/2&\mu_3+\lambda_3^2
\end{array}\right).
\end{equation}
The positivity of this matrix is equivalent (see e.g. \cite[Thm.~7.2.5]{HorneJ}) to the positivity of the principal minors, i.e., the determinants of the submatrices of the first $k$ rows and columns for $k=1,2,3,4$. The first three of these are $1$, $\mu_1$, and $\mu_1\mu_2-\lambda_3^2/4$. The positivity of the third one is Robertson's inequality \eqref{Robson}. The only remaining condition for the positivity of $M$ is $\det M\geq0$, which evaluates to
\begin{equation}\label{Robson4}
	\mu_1\mu_2\mu_3-\frac14\bigl(\lambda_1^2\mu_1+\lambda_2^2\mu_2+\lambda_3^2\mu_3\bigr)\geq0.
\end{equation}
This will be combined with the normalization condition
\begin{equation}\label{RobsonNorm}
	\lambda_1^2+\lambda_2^2+\lambda_3^2=s(s+1)-(\mu_1+\mu_2+\mu_3).
\end{equation}
The condition on the triples $(\mu_1,\mu_2,\mu_3)$ we have to evaluate is the existence of $\lambda_i$ satisfying both these relations. Since only the squares enter, let us set $x_j=\lambda_j^2$. Then \eqref{RobsonNorm} describes a triangle in the positive quadrant with equal intercept $s(s+1)-(\mu_1+\mu_2+\mu_3)$ with the axes. The inequality \eqref{Robson4} describes a tetrahedron spanned by the origin and the axis intercepts $x_1^0=4\mu_2\mu_3$, and cyclic. Note that Robertson's inequality is automatically satisfied on this tetrahedron. Obviously the tetrahedron and the triangle intersect if an only if one of the axis intercepts of the tetrahedron reaches or lies above the triangle. Since we can take the eigenvalues ordered: $\mu_1\geq\mu_2\geq\mu_3$ this means
\begin{equation}\label{RobsoNew}
	4\mu_1\mu_2\geq s(s+1)-(\mu_1+\mu_2+\mu_3).
\end{equation}
This is a bound to the eigenvalue of the $\Lambda$-matrix. By Birkhoffs theorem, the variances arising from such $\Lambda$ also includes all convex combinations of permutations of the $\mu_i$ (see beginning of Sect.~\ref{sec:prep}).
In order to characterize the set of variance triples generated in this way we need the following Lemma. In its formulation the variables $\sigma \in S_3$ run over the permutation group on three elements, and are applied to the components of a 3-vector (see also FIG.~\ref{fig:hexagon}).

\begin{figure}[t]
	\begin{center}
		\includegraphics[width=0.4\textwidth]{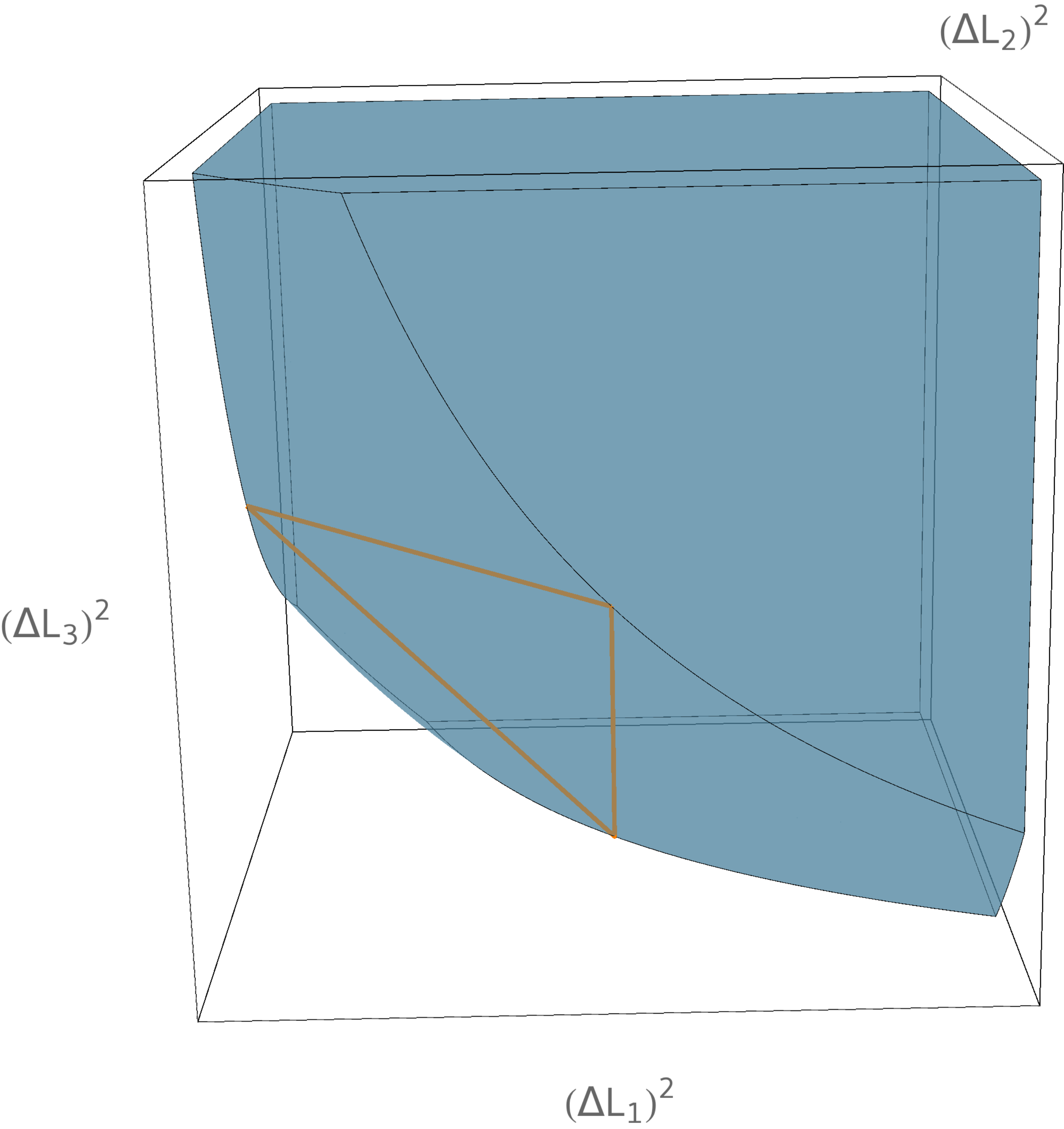}
		\includegraphics[width=0.4\textwidth]{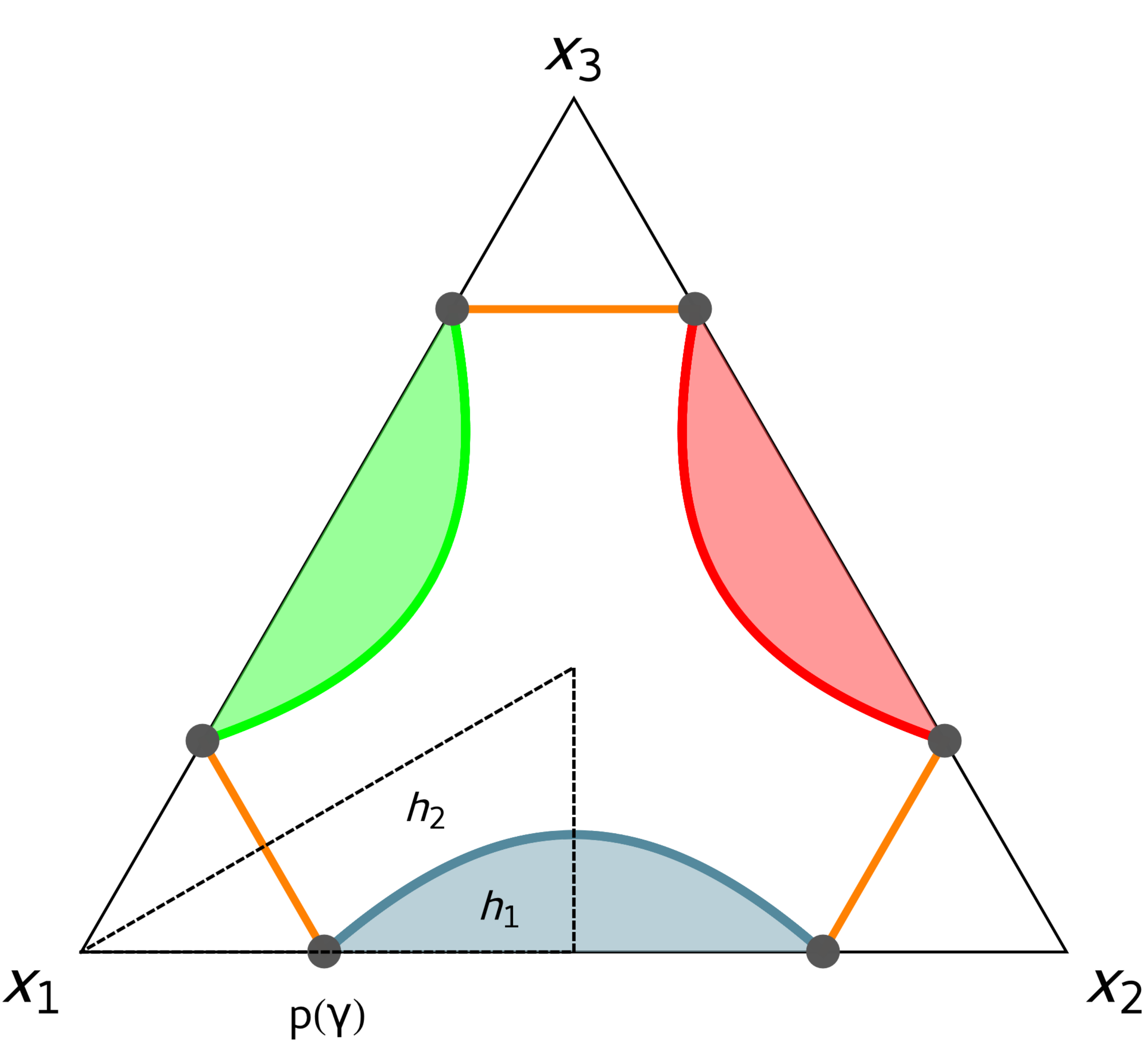}
	\end{center}
	\caption{Left: Region bounded by the inequality \eqref{Robson3}.
     Right: Plane of constant $\gamma$ orthogonal to the $(1,1,1)$ direction.}
	\label{fig:roblem}
\end{figure}

\begin{lem}\label{roblem}
	With the notation from above, the following sets $K_1$ and $K_2$ are  equal.
	\begin{eqnarray}
	& K_1 = \bigcup\limits_{\gamma \geq s} H_1(\gamma), \text{  where} \\
	& H_1(\gamma)= \conv\left(\bigcup\limits_{\sigma \in S_3} \sigma( h_1(\gamma) ) \right),  \textbf{  with } \nonumber \\
	& h_1(\gamma)= \{ (\mu_1,\mu_2, \mu_3) \ \vert \  \sum_i \mu_i = \gamma, \ \mu_1 \geq \mu_2 \geq \mu_3 \geq 0, \ 4 \mu_1 \mu_2 \geq s(s+1) - \gamma \}  \nonumber \\
	& \textbf{ and } \nonumber  \\
	& K_2 = \bigcup\limits_{\gamma \geq s} H_2(\gamma), \text{  where} \\
	& H_2(\gamma)=\bigcup\limits_{\sigma \in S_3} \sigma( h_2(\gamma) ),  \textbf{  with } \nonumber \\
	& h_2(\gamma) = \{ (v_1, v_2, v_3) \  \vert \  \sum_i v_i = \gamma, \ v_1 \geq  v_2 \geq v_3 \geq 0, \ 4 v_1 ( v_2 + v_3 ) \geq  s(s+1) - \gamma \} \nonumber.
	\end{eqnarray}
\end{lem}
\begin{proof}
	For the equality of $K_1$ and $K_2$ it is sufficient to show that $H_1$ and $H_2$ coincide for every $\gamma$. The restriction $\sum_i v_i = \sum_i \mu_i = \gamma$, together with the 3-fold symmetry of the problem, tells us that $H_1(\gamma)$ and $H_2(\gamma)$ are subsets of the triangle, whose corners lie on the axes at a distance $\gamma$ from the origin.  In this triangle the ordering of the $v_i$ and $\mu_i$ reduces $h_1$ and $h_2$ to the dashed subset marked in FIG.~\ref{fig:roblem}. \\
	Now the first and last condition in the definition of $h_2$ can be combined to obtain
	\begin{eqnarray}
		4 v_1 (\gamma - v_1) &\geq& s(s+1) - \gamma \\
		0 &\geq& v_1^2 - \gamma v_1 - \frac{s(s+1) - \gamma}{4}
	\end{eqnarray}
	so we get
	\begin{equation}
		v_1 \leq - \frac \gamma 2 \pm \sqrt{\frac{\gamma^2} 4 + \frac{s(s+1) - \gamma}{4}  } =: c(\gamma).
	\end{equation}
	Because $v_1 \geq 0$ we have to choose the positive sign, which means that $H_2(\gamma)$ is the intersection of the triangle with the three halfspaces $v_i\leq c(\gamma)$, whose boundaries are marked as a orange lines in FIG.~\ref{fig:roblem}, i.e.,
	\begin{equation}
		H_2(\gamma)=\{ (v_1, v_2, v_3) \  \vert \  \sum_i v_i = \gamma, v_i\leq c(\gamma) \}.
	\end{equation}
	$H_2(\gamma)$ is clearly a convex polytope. The extremal points of $H_2(\gamma)$ have to saturate at two of the defining inequalities. In the ordered triangle ($v_1\geq v_2\geq v_3$) the only extreme point is given by  $p(\gamma):=(c(\gamma),\gamma - c(\gamma) ,0)$, and all others can be obtained by permutations.
	Hence $H_2$ can be described as the hexagon $ H_2(\gamma) = \conv\left(\bigcup_\sigma \sigma(p(\gamma))\right)$.
	On the one hand, by comparing the defining inequalities for $h_1$ and $h_2$, we can see that every triple $\mu_i \in h_1$ is also part of $h_2$. So including the permutations and by the fact that $H_2(\gamma)$ is convex, we get $H_1(\gamma) \subseteq H_2(\gamma)$.\\
	On the other hand, the 3-component of $p$ is zero, so it is also part of $h_1 \in H_1(\gamma)$. While the point $p$ and its permutations are the extremal points of $H_2(\gamma)$ and $H_1(\gamma)$ is convex, we have $H_2(\gamma) \subseteq H_1(\gamma)$.
\end{proof}
Therefore we get the following statement:
\begin{prop} \
	\\ Let $v_1 \geq v_2 \geq v_3 \geq 0$ be variances of the angular momentum components, then the following holds:
	\begin{equation}
	4 v_1 ( v_2 + v_3 ) \geq s(s+1) - (v_1 + v_2 + v_3 ).
	\end{equation}
\end{prop}
\begin{figure}[bt]
	\begin{center}
		\includegraphics[width=0.45\textwidth]{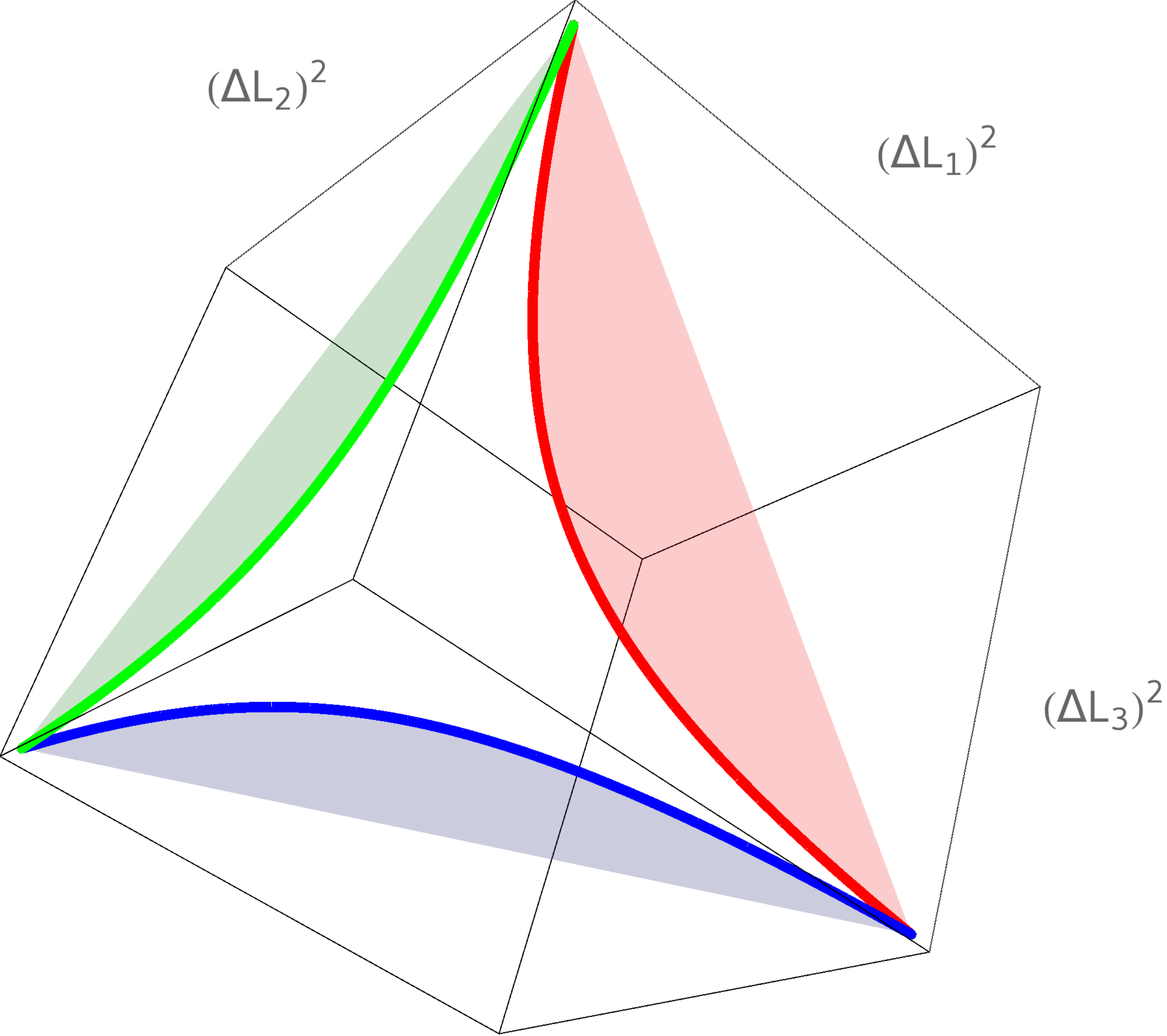}
		\includegraphics[width=0.45\textwidth]{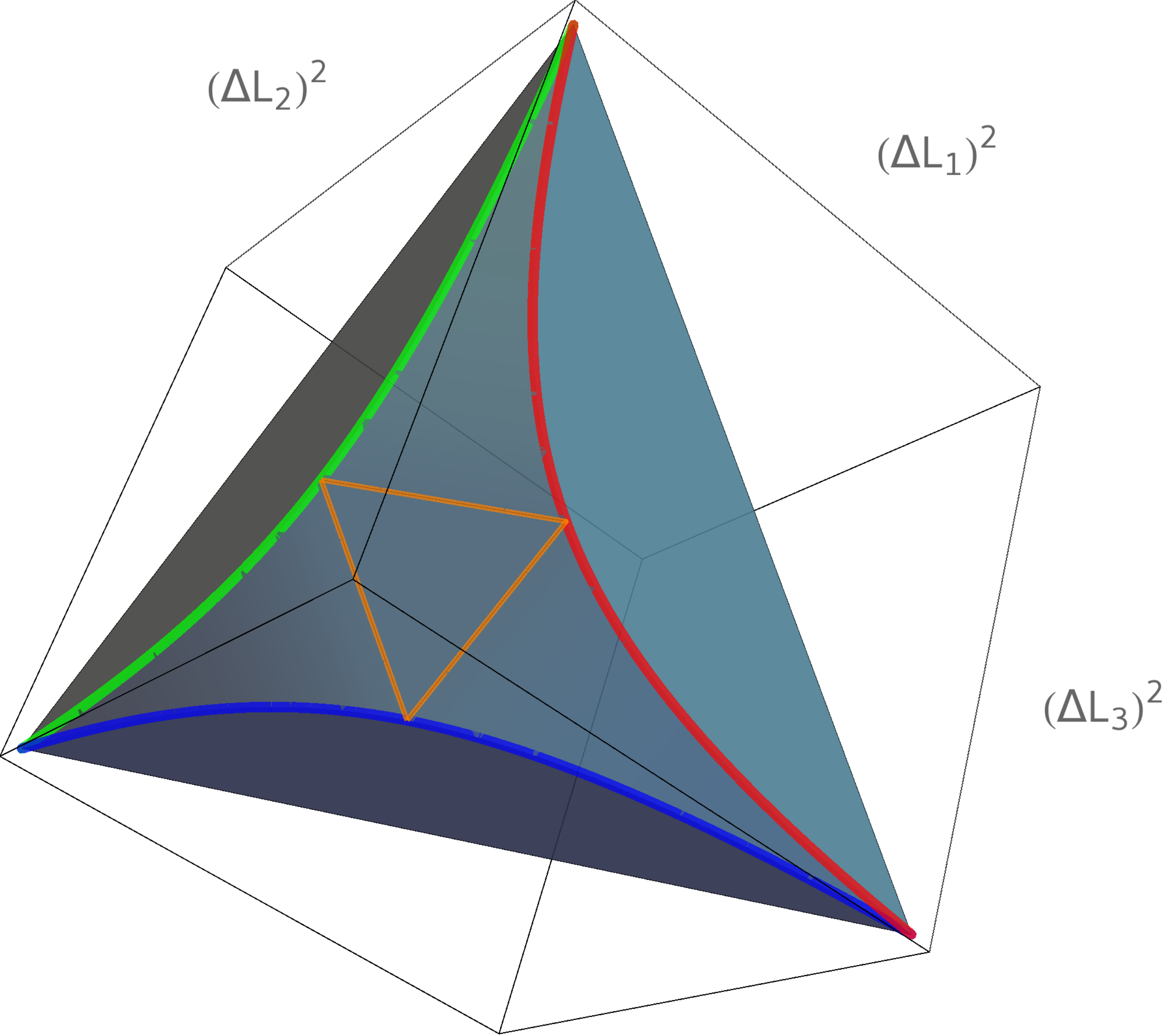}
	\end{center}
	\caption{Region bounded by the generalized Robertson inequality. Left: Hyperbolic curves on the faces for the eigenvalues. Right: Uncertainty region given by the variances and the base triangle formed by the spin coherent states.}
	\label{fig:roblemmapic}
\end{figure}
As one can see in FIG.~\ref{fig:roblem}, the boundaries of the corresponding uncertainty region on the coordinate planes are given by permutations of the hyperbolic curve $4v_1 v_2 = s(s+1)-v_1-v_2.$ This uncertainty region is monotonously closed and given by the convex hull of the above hyperbolic curves. This is shown in FIG.~\ref{fig:roblemmapic}.

% % % Asymptotic % % %
\subsection{Asymptotic Case}\label{sec:asymptotic}
Now we take a look at the behavior of the asymptotic uncertainty region for $s\to\infty$. We already know that $\var{\rho}{L_1}+\var{\rho}{L_2}+\var{\rho}{L_3} \geq s$ and hence it is appropriate to scale the problem by 1/$s$, which will fix the sum of the variance in the lower base triangle to $1$.
We start with the asymptotic behavior of the generalized Robertson inequality derived in the previous section.
On the scale of $1/s$, i.e. ${v_i}/{s} = \nu_i$, and the ordering $\nu_1\geq \nu_2 \geq \nu_3$ this inequality reads
\begin{equation}\label{rob1s}
	4\nu_1(\nu_2+\nu_3)\geq 1+\frac{1}{s}(1-(\nu_1+\nu_2+\nu_3))
\end{equation}
and as $s$ goes to infinity the set of possible variances shrink to
\begin{equation}	\label{asymrob}
	4\nu_1(\nu_2+\nu_3)\geq 1,
\end{equation}
because $\sum_i \nu_i \geq 1$. Hence the inequality \eqref{rob1s} gets stronger for increasing $s$.

In this section we will show that this bound is attained by states, which will be constructed in the following way.
Using the technique described in part \ref{sec:method}, we look for the states $ \psi$, which minimize the expectation of the operator
\begin{equation}\label{exact}
	H(s,\mathbf{w})=\frac 1 s\Big( w_1 (L_1-\lambda_1)^2 + w_2 ( L_2 - \lambda_2 ) ^2 + w_3 (L_3-\lambda_3)^2 \Big),
\end{equation}
for a normal vector $\mathbf{w}$.
We do this using the Holstein Primakoff transformation \cite{HP}:
\begin{equation}
	L_+ = \sqrt{2s}\sqrt{1-\frac{a^*a}{2s}}a \quad L_- = \sqrt{2s} a^* \sqrt{1-\frac{a^*a}{2s}} \quad L_3 = s-a^*a.
\end{equation}
Here $a$ and $a^*$ are the creation and annihilation operators, so we have a representation of the angular momentum algebra in the oscillator basis. For large s and appropriate states, this transformation can be reduced to
\begin{equation}\label{approx}
	L_+ = \sqrt{2s}a + \Order(s^{-\frac 1 2})\quad L_- = \sqrt{2s} a^* + \Order(s^{-\frac 1 2}) \quad L_3 = s-n.
\end{equation}
Notice that in the Holstein Primakoff basis, the spin coherent state $\ket s$ is transformed to the ground state $\ket 0 _{HP}$, hence the state $\ket n_{HP}$ corresponds to $\ket{s-n}$ in the standard angular momentum $L_3$ eigenbasis. Now we rewrite $H$ using the above transformation and the relation for position $L_1 = \frac 1 2 ( L_+ +i L_-) = \frac{\sqrt{2 s}}{2} ( a + a^* ) = \sqrt s X$ and momentum $L_2 = \frac 1 2 ( L_+ - iL_-) = \frac{\sqrt{2 s}}{2i} ( a - a^* ) = \sqrt s P$.  We arrive at
\begin{equation}\label{approx2}
	H(s,\mathbf{w})_{HP}=\Big( w_1 (X-\xi)^2 + w_2 (P - \eta ) ^2 + \frac{w_3}{s} (s -a^*a -\zeta)^2 \Big) \; + \Order(s^{-\frac 1 2}).
\end{equation}
Here $\xi, \eta$ and $\zeta$ denote the transformed expectation values.
From \ref{sec:Pbasics} we know that $\ket s$ has minimal uncertainty for $w\sim (1,1,1)$ and arbitrary $s$. Based on this observation we make the assumption that we are close to the $L_3$ spin coherent state. We thus have $s\gg \langle a^*a\rangle$ and $\lambda_3 \approx s$, hence $\zeta$ is linear in $s$. Furthermore we can order the weights, such that $w_1 \leq w_2 \leq w_3$ to minimize the expectation value.  Now we take the limit and let $s$ becomes large, the operator converges to harmonic oscillator
\begin{equation}
	H(\mathbf{w})_{HP}= w_1 X^2 + w_2 P^2.
\end{equation}
Here we use that the expectation value of the harmonic oscillator is translation-invariant in phase space, so that we can choose $\xi$ and $\eta$ to be zero.
The state which minimizes the expectation of this operator is simply the harmonic oscillator ground state $\psi(m,\omega)$, with $m=\frac{1}{2 w_2}$ and $\omega=\sqrt{4 w_1 w_2}$. In the following these will be combined in the parameter $\alpha:=m \omega= \sqrt{\frac{w_1}{ w_2}}$.
For the comparison of this result with numerical calculations using the above described algorithm, we must express these ground states in a common basis $\ket n_{HP}$, i.e. decomposing $\psi(\alpha)$ in the basis of a harmonic oscillator with $\alpha=1$.  This transformation is given by
\begin{equation}
	\psi_n := \brAket{n}{\psi(\alpha)}= \frac{(\alpha )^{\frac 1 4}}{\sqrt{2^n \pi n!}} \int H_n(x) \exp\left(-\frac{(1+\alpha)}{2}x^2\right) dx,
\end{equation}
which is zero for odd $n$ and can be solved for even $n$ through
\begin{equation}
	\int H_n(x) \exp(-cx^2) dx = \frac{i^{3n} \sqrt{\pi} n! (c-1)^{\frac{n}{2}}}{\frac n 2 ! c^{\frac {n+1}{2}}}.
\end{equation}
The corresponding probability distribution is given by
\begin{equation}
	p_n:=|\psi_n|^2=\frac{\sqrt{\alpha}}{1+\alpha}\frac{(1-\alpha)^n}{(2+2\alpha)^n} \frac{n!}{(\frac n 2 !)^2} (1+(-1)^n).
\end{equation}
Because this is zero for odd $n$, we can set $n=2k$ and get
\begin{equation}\label{hpdis}
	p_{2k}=\frac{2\sqrt{\alpha}}{1+\alpha}\frac{(1-\alpha)^{2k}}{(2+2\alpha)^{2k}}\binom{2k}{k}.
\end{equation}
The above approximation does not necessarily yield the optimal states and it is not rigorously justified so far. As a first step, we compare the distribution $p_n$ with numerically determined ones for finite $s$. These tend to converge as shown in FIG.~\ref{pic:hpdingens}.
\begin{figure}[t] \centering
	\includegraphics[width=0.65\textwidth]{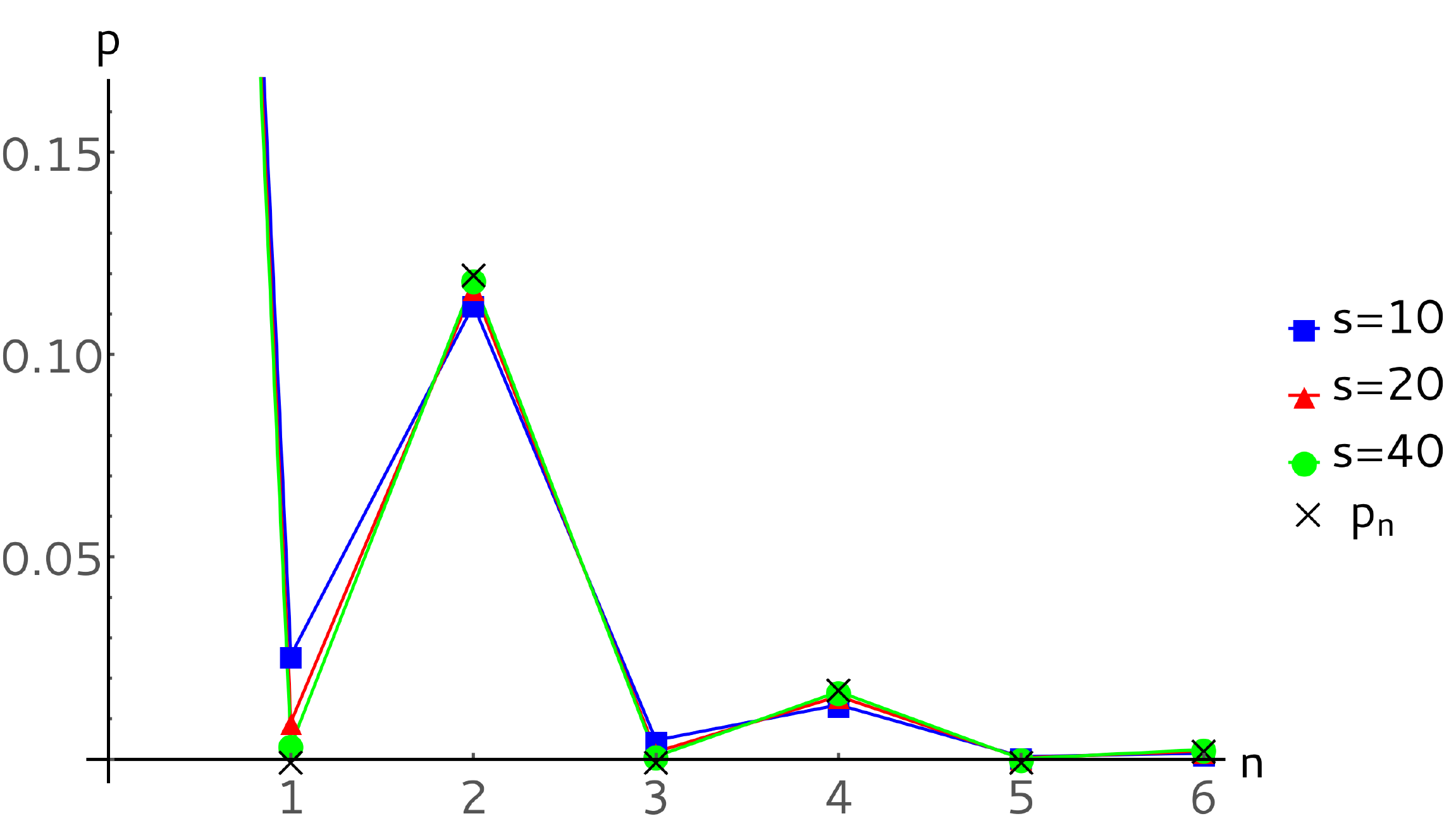}
	\caption{Comparison of the occupation number from $\psi(\alpha)$ with the numerical calculations.}
	\label{pic:hpdingens}
\end{figure}
\begin{thm}\label{thm:largeS}
	The lower bound of the asymptotic uncertainty region on a scale of 1/$s$ is fully described by the generalized Robertson inequality (\ref{asymrob}) and is saturated by the states $\psi(\alpha)$.
\end{thm}
\begin{proof}
	First we will show that the approximation (\ref{approx}) is justified for $\psi(\alpha)$ and evaluate the corresponding asymptotic variances.
	While the generalized Robertson inequality gets stronger for increasing $s$, every extremal point of the corresponding boundary is attained by $\psi(\alpha)$, which will prove the above statement.
	Moreover, by truncating the sequence $\psi_n(\alpha)$ at $n=2s+1$ and renormalizing, we get a sequence of spin-$s$ states well approximating $\psi(\alpha)$ as $s$ goes to infinity. \\
	With this in mind, we will prove the above statement in two steps:\\
	(i)
	On the one hand we have to verify that $\lim\limits_{s\rightarrow\infty} \frac {1}{ \sqrt{s}}L_{+}\ket{\psi(\alpha)}=\sqrt{2}a\ket{\psi(\alpha)}$ and $\lim\limits_{s\rightarrow\infty} \frac {1}{ \sqrt{s}}L_{-}\ket{\psi(\alpha)}=\sqrt{2}a^*\ket{\psi(\alpha)}$
	which is true if $\psi(\alpha)$ is in the domain of $a^*a$. On the other hand we have to show that the term $\frac{w_3}{s} (s -a^*a -\zeta)^2 $ from (\ref{approx2}) will vanish for $\psi(\alpha)$ with $\zeta=\brAAket{\psi(\alpha)}{s-a^*a}{\psi(\alpha)}$ as $s$ goes to infinity.
	Both requirements are fulfilled if the moments $\brAAket{\psi(\alpha)}{a^*a}{\psi(\alpha)}$ and $\brAAket{\psi(\alpha)}{(a^*a)^2}{\psi(\alpha)}$ are finite.\\
	In the Holstein-Primakoff occupation basis $\ket{n}_{HP}$, these moments are given by series of the form
	\begin{equation}
		\sum\limits_{n=0}^\infty n^c p_n=\sum\limits_{k=0}^\infty (2k)^c p_{2k},
	\end{equation}
	and can be computed as derivatives using the generating function
	\begin{equation}
		\frac{c_1}{\sqrt{1-4x}} = c_1 \sum\limits_{n=0}^\infty x^k \binom{2k}{k}
	\end{equation}
	of the probability distribution $p_{2k}$ \eqref{hpdis}. By straightforward calculations we get
	\begin{align}
		\brAAket{\psi(\alpha)}{a^*a}{\psi(\alpha)}&=
		\frac{(1-\alpha)^2}{4\alpha}\\
		\brAAket{\psi(\alpha)}{(a^*a)^2}{\psi(\alpha)}&=
		\frac{3(1-\alpha)^4}{16\alpha^2}+\frac{(1-\alpha)^2}{2\alpha},
	\end{align}
	which is finite for $\alpha > 0$. \\
	(ii) Now the asymptotic variances for $\psi(a)$ can be determined. For $\psi(\alpha)$ the operators $ \frac{1}{\sqrt s}L_1 $, $\frac{1}{\sqrt s}L_2$ and $\frac{1}{\sqrt s} L_3$ converge to $P$, $Q$ and a multiple of the identity, we obtain
	\begin{align}
		\lim\limits_{s\rightarrow\infty}\frac{1}{s}
		\var{\psi(\alpha)}{L_1}&=\var{\psi(\alpha)}{Q} =\frac{1}{2\alpha}\\
		\lim\limits_{s\rightarrow\infty}\frac{1}{s}
		\var{\psi(\alpha)}{L_2}&=\var{\psi(\alpha)}{P}=\frac{\alpha}{2}\\
		\lim\limits_{s\rightarrow\infty}\frac{1}{s}
		\var{\psi(\alpha)}{L_3}&= \lim\limits_{s\rightarrow\infty}\frac{1}{s}    \var{\psi(\alpha)}{s-a^*a}=0.
	\end{align}
	This set of variance triples $(\frac{1}{2\alpha},\frac{\alpha}{2},0)$ saturates the asymptotic generalized Robertson bound (\ref{asymrob}). Moreover they describe the extremal boundary curves, see the proof of theorem \ref{roblem}, of the associated uncertainty region.
\end{proof}

% % % Preparation - Special topics % % %
\section{Preparation uncertainty: Special topics}

% % % Vectormodel % % %
\subsection{The vector model and moment problems}\label{sec:vecMod}
This may be a good place to comment on the so-called vector model of angular momentum, as it was suggested by Old Quantum Theory. It still seems to be quite popular in teaching, although theoreticians tend to deride it as ridiculously classical and obviously inconsistent. Indeed, its two-particle version gives manifestly false predictions even for spin-1/2, as witnessed by Bell's (CHSH) inequality. Since {\it any} local classical model fails this test, not much can be learned about angular momentum from this observation. Therefore we consider here only the one-particle version, and try to sort out how far it can be trusted.

The basic rationale of the vector model is shown in FIG.~\ref{fig:vmodel}: Angular momentum is thought of as a classical random variable taking values on a sphere of radius $r_s=\sqrt{s(s+1)}$. For an eigenstate $\ket m$ the corresponding classical distribution is supposed to be concentrated at latitude $m$, and uniform with respect to rotations around the 3-axis. The expectation value of this distribution is $(0,0,m)$. Moreover, its matrix of second moments is also diagonal, since the coordinate axes are clearly the inertial axes of a mass uniformly distributed on a circle of fixed latitude. One readily checks that all second moments are the same as for the corresponding quantum state. This can be generalized:
\begin{figure}[ht]
	\includegraphics[width=0.35\textwidth]{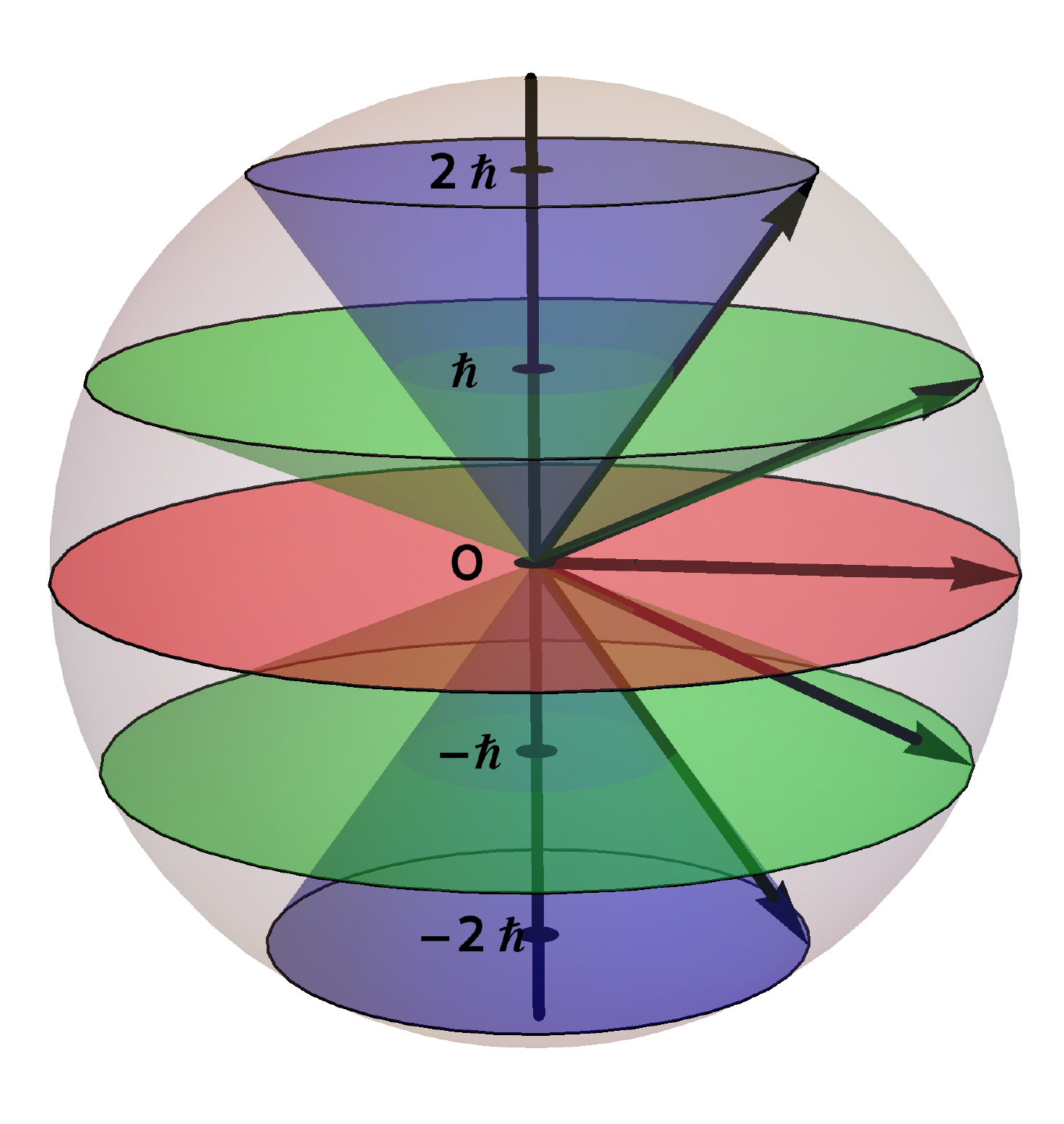}
	\caption{The well known vector model of anguluar momentum for $s=2$}
	\label{fig:vmodel}
\end{figure}

\begin{prop}\label{prop.vecmodel}
	For any quantum state $\rho$ on $\HH$ there is a classical probability distribution $\mu$ on the sphere of radius $\sqrt{s(s+1)}$ which has the same first and second moments as the angular momentum in $\rho$, i.e.,
	\begin{equation}\label{vmmoments}
	m_j=\int\mu(dx)\, x_j=\tr\rho L_j   \qquad\mbox{and}\quad
	M_{jk}=\int\mu(dx)\, x_jx_k= \re\tr\rho L_jL_k .
	\end{equation}
\end{prop}

For the proof we only need a characterization of the moments $(m_j,M_{jk})$ of probability measures on a sphere of radius $r$, which turns out to be quite simple. This in turn provides an immediate proof of Proposition~\ref{prop.vecmodel}, since the quantum moments $M_{jk}=\re\tr\rho L_jL_k$, $m_j=\tr\rho L_j$ obviously have the required properties with radius $r^2=s(s+1)$.

\begin{lem}
	Let $(m_{j})\in\Rl$ for $j=1,2,3$, and $(M_{jk})_{j,k=1}^3$ a real symmetric $3\times3$-matrix. Then these numbers are the first and second moments of a probability distribution on the sphere of radius $r$ if and only if $\tr M=r^2$ and the covariance matrix $M_{jk}-m_jm_k$ is positive semi-definite.
\end{lem}
\begin{proof}
	Necessity is obvious because covariance matrices are always positive and the function $v\mapsto\sum_jv_j^2=r^2$ is constant on the sphere.
	For the converse consider the set $K$ of moments $(m,M)$ of probability measures on the sphere. This is a compact convex set, which we can think of
	as embedded into $3+6-1$ dimensional real space, because the real symmetric matrix $M$ is specified by $6$ parameters, and we have an additional linear constraint $\tr M=r^2$. By the separation theorems for compact convex sets the set $K$ is therefore completely characterized by a collection of affine inequalities
	\begin{equation}\label{affuncVM}
		f(m,M)=\tr AM-b\cdot m+\gamma\geq0,
	\end{equation}
	where $A$ is real symmetric, $b\in\Rl^3$ with the dot indicating scalar product, and $\gamma\in\Rl$. The functionals for which these inequalities have to be satisfied are precisely those for which the above inequality holds for all {\it pure} probability measures, i.e., for $M_{jk}=v_jv_k$ $m_j=v_j$ for some $v\in\Rl^3$ on the sphere. In this case we slightly abuse notation and write $f(m,M)=f(v)$.

	Not all inequalities are needed to characterize $K$, but only the extremal ones, which furnish a minimal subset from which all the others follow as linear combinations with positive scalar factors. In particular, we can assume that $f$ is not strictly positive, so has a zero $f(u)=0$, which then also has to be a minimum. The extremality condition gives $2Au-b=2\lambda u$, where $\lambda\in\Rl$ is a Lagrange multiplier. This determines $b$, and from $f(u)=0$ we get $\gamma$, so that we can rewrite
	\begin{equation}\label{avvuncVM}
		f(v)=(v-u)\cdot A(v-u)+2\lambda u\cdot(v-u).
	\end{equation}
	Now since $u,v$ lie on a sphere of radius $r$, we can write $2u\cdot(v-u)=-(v-u)^2+(v^2-u^2)=-(v-u)^2$, so we can combine the two terms, and obtain again the form \eqref{avvuncVM}, with $A$ modified by a multiple of the identity, and $\lambda=0$.

	It remains to determine all real symmetric matrices $A$ such that $(v-u)\cdot A(v-u)\geq0$, whenever $u,v$ lie on a sphere of radius $r$. Equivalently, $\xi\cdot A\xi\geq0$ for all multiples of vectors of the form $v-u$. But this set is dense in $\Rl^3$. Hence the desired condition is just the positive semi-definiteness of $A$. The resulting inequality for $(m,M)$ can be rewritten in terms of the covariance matrix $V_{jk}=M_{jk}-m_jm_k$ as
	\begin{equation}\label{affuMcVM}
	f(m,M)=\tr VA+ (m-u)\cdot A(m-u)\geq0.
	\end{equation}
	Since the second term is anyhow positive, the positive semidefiniteness of $V$ is sufficient for all these inequalities. This shows the sufficiency of the conditions stated in the Lemma.
\end{proof}

Let us make some remarks, which all fit into a fruitful analogy here with the phase space case, i.e., the case of two canonical operators $P,Q$, and {\it moment problems} posed in the respective  contexts.
\begin{enumerate}
\item The phase space analogue of Prop.~\ref{prop.vecmodel} is the statement that for any quantum state the first and second moments can also be realized by a classical probability distribution on phase space.
Of course, not all classically allowed first and second moments can arise in this way: This is just the theme of preparation uncertainty relations.
\item The classical probability measure $\mu$ is not uniquely defined by $\rho$. For example, the density operator $\rho=(2s+1)\inv\idty$ can either be represented by the uniform distribution on the sphere, or by an equal-weight mixture of the distributions with constant latitude $m$ (in any direction). In the phase space case it is well-known that with the given, quantum-realizable moments one can always find a {\it Gaussian state}, which is defined as the distribution with the maximal entropy given those moments. The same idea also works for angular momentum, and it gives probability densities which are the exponential of a quadratic form in the variables. In contrast to the phase space case, when approaching eigenstates (any direction, any $m$) this entropy will go to $-\infty$, since for eigenstates only the singular measures depicted in FIG.~\ref{fig:vmodel} can be used.
\item Prop.~\ref{prop.vecmodel} is certainly false if we include higher than second moments. For example, consider a pure qubit state with $m=+1/2$. Without loss of generality we can choose the measure $\mu$ invariant under rotations around the 3-axis. Since $\mu$ must be concentrated on $m=1/2$, this uniquely fixes the measure $\mu$, and hence the moments to all orders. Now consider a direction $\ve$ which is at an angle strictly between $0$ and $\pi/2$ to $\ve_3$. Then the quantum expectation of $(\eL)^3=\eL/4$ is $(\ve{\cdot}\ve_3)/8$, but the classical expectation of $(x{\cdot}\ve)^3$ is larger, reflecting the non-linearity of the cube function.
\item The quantum analogue of the classical Hamburger moment would be to reconstruct a quantum state from the set of moments, i.e., the expectations of the monomials in the basic operators ($P,Q$, or $L_1,L_2,L_3$). Commutation relations impose some constraints on these moments, so that in the end only monomials like $L_1^{n_1}L_2^{n_2}L_3^{n_3}$ need to be considered. Of course, the expectation values of such operators will generally be complex numbers. Can we do the reconstruction for arbitrary states in the spin-$s$ representation? Indeed, we can, and it is actually much easier than in the phase space case since only finitely many moments suffice. The basic observation is that the moments fix all expectations on the von Neumann algebra $\mathcal A$ generated by the $L_i$. Because the representation is irreducible, the commutant of this algebra consists of the multiples of the identity. Hence $\mathcal A$ must be the full matrix algebra, and the state is uniquely determined. That finitely moments suffice is clear because $\dim\mathcal A<\infty$.
\item Noncommutative moment problems are plagued by ``operator ordering'' issues. But in some sense we have already adopted a standard ``symmetrized'' solution for operator ordering, namely to form moments only of the operators $\eL$ for all fixed $\ve$. This is analogous in the phase space case to considering the moments of linear combinations of $P$ and $Q$. Now, famously, the full distributions of all such combinations are correctly rendered by the Wigner distribution function, which is itself hardly ever positive \cite{Hudson}. The analogy to the angular momentum case is immediate. So what do we get if we accept ``quasi-probability distributions''? Can every state be represented like that? This is answered by the following proposition.
\end{enumerate}

\begin{prop}
  Let $\rho$ be a quantum state in the spin-$s$ representation. Then there is a unique tempered distribution $\wigner_\rho$ on $\Rl^3$, which is formally real, has support in a ball of radius $s$ and satisfies, for all $\ve$ and $n\in\Nl$
\begin{equation}\label{wignerMoments}
  \int d\veta\ \wigner_\rho(\veta)\, (\ve\cdot\veta)^n =\tr\bigl(\rho (\eL)^n\bigr)\ .
\end{equation}
\end{prop}

\begin{proof}
We can compute the Fourier transform of $\wigner_\rho$ directly from \eqref{wignerMoments}, by multiplying with $(i k)^n/n!$ and summing over $n$. This turns the left side into the Fourier integral over $\wigner_\rho$ allowing the sum to be evaluated also on the right hand side:
\begin{equation}\label{wigner}
  \widehat\wigner_\rho(\vk)=\int d\veta\ \wigner_\rho(\veta)\, e^{i \vk\cdot\veta}
       =\tr\Bigl(\rho e^{i\vk\cdot\vL}\Bigr),
\end{equation}
where $\vk=k\ve$. Strictly speaking this computation should be regularized by multiplying with an arbitrary test function before summation, but this would lead to the same explicit representation of the Fourier transform $\widehat\wigner_\rho$ as a bounded $\CC^\infty$-function. This shows that the desired tempered distribution $\wigner_\rho$ is essentially unique, and can be defined for every $\rho$. It is formally real, because
$\widehat\wigner_\rho(-\vk)=\overline{\widehat\wigner_\rho(\vk)}$. For the claim about the support we invoke the distributional version of the Paley-Wiener-Schwartz Theorem \cite[Thm.~7.3.1]{Hoerman}. According to that theorem we need to show only that for real vectors $\vk,\vkap$ the estimate
\begin{equation}\label{paley}
  \left|\tr\bigl(\rho e^{i(\vk+i\vkap)\cdot\vL}\bigr)\right| \leq C e^{s|\vkap|}
\end{equation}
holds for some constant $C$. Treating the sum in the exponential by the Trotter formula, which we may, because these are finite dimensional matrices, we get
\begin{equation}\label{paleytrotter}
  \left\Vert e^{i(\vk+i\vkap)\cdot\vL}\right\Vert
    =\lim_{N\to\infty}\left\Vert\left(e^{i(\vk/n)\cdot\vL}\ e^{(-\vkap/n)\cdot\vL}\right)^n\right\Vert
    \leq\lim_{N\to\infty} \left(e^{(\vkap/n)\Vert\vL\Vert}\right)^n=e^{s|\vkap|}.
\end{equation}
Clearly this implies the desired estimate.
\end{proof}

This Proposition is remarkable in comparison to Prop.~\ref{prop.vecmodel}: If we insist on positivity but require only the first two moments to be correct, the vector model requires a sphere of radius $\sqrt{s(s+1)}$. On the other hand, if we waive instead the positivity of the classical distribution, we are forced to use a ball of radius $s$.

One might ask at this point, in continuation of the above analogies to the phase space case, whether there are ``classical'' states, for which the Wigner function $\wigner_\rho$ is a positive function. However, this is easily seen to be impossible. Indeed, when the marginal of a proper probability distribution has support on $\{-s,\ldots,s\}$, the measure itself has to have support on a union of hyperplanes $\{\veta\,|\,\veta\cdot\ve=m\}$. But these families of hyperplanes, drawn for various $\ve$ have empty intersection, contradicting the normalization of the measure.

The main use of Wigner functions on phase space is the visualization of quantum states. Unfortunately, the much more singular nature of $\wigner_\rho$ for angular momentum will prevent these Wigner functions from becoming similarly popular. This irregularity can be tamed by replacing on the right hand side of \eqref{wigner}
\begin{equation}\label{scully}
  e^{\textstyle i\vk\cdot\vL} \mapsto e^{\textstyle ik_1L_1}\ e^{\textstyle ik_2L_2}\ e^{\textstyle ik_3L_3}.
\end{equation}
The corresponding distribution is then a sum of point measures sitting on a finite cubical grid \cite{scully}. This may actually be useful in quantum information, where it relates to a discrete phase space structure over the cyclic group of $d=2s+1$ elements. However, for angular momentum proper we find this breaking of rotational symmetry abhorrent.

% % % Entropy % % %
\subsection{Entropic uncertainty}\label{sec:entropic}
In this section we will have a look at the entropic uncertainty relations.
Given a measurement of a hermitian operator $A=\sum_ia_iP_i$, with eigenprojectors $P_i$, the probability of obtaining the $i^{\rm th}$ measurement outcome will be denoted by $\pi_i(\rho)=\tr(\rho P_i)$ and the associated probability distribution as $\pi(\rho)= \{\pi_1(\rho),\cdots,\pi_d(\rho)\}$. Then the output entropy of $A$ in the state $\rho$ is defined as the Shannon entropy of $\pi(\rho)$
\begin{align}
H(A,\rho):=H(\pi(\rho))=-\sum\limits_{i=1}^{d} \pi_i(\rho) \log_d (\pi_i(\rho)),
\end{align}
which serves as an uncertainty measure. Note that we normalize the Shannon entropy by its maximal value $\log d$ so that all occurring entropies are bounded by $1$.
In contrast to the variance, the entropy of a probability distribution does not change by permuting or rescaling the measurement outcomes and so only depends on the choice of the $P_i$ (up to permutations) and not on the eigenvalues $a_i$. This implies that an entropic uncertainty relation, which constrains the output entropies of two (for simplicity non-degenerate) observables $A,B$, only depends on the unitary operator $U$ connecting the respective eigenbases. A well-known bound in this setting is the general Maassen-Uffink bound \cite{muff}
\begin{eqnarray}\label{eq:muff}
H(A,\rho)+H(B,\rho)&\geq& -2\log_d(c)\\
c=c(U)&=&\max\limits_{ij}\abs{U_{ij}} .
\end{eqnarray}
\begin{figure}[h]
	\includegraphics[width=0.4\textwidth]{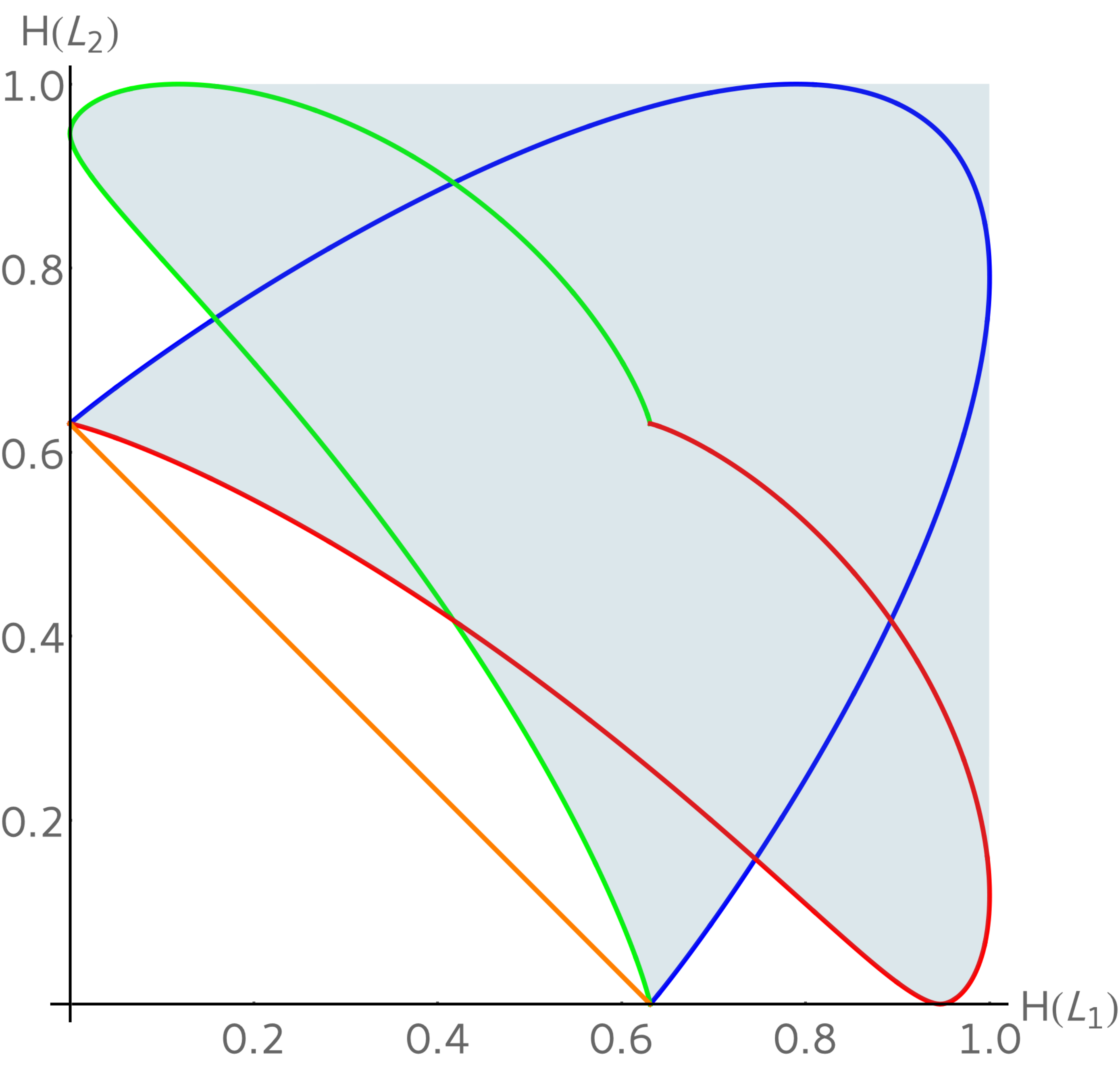} \quad
	\includegraphics[width=0.4\textwidth]{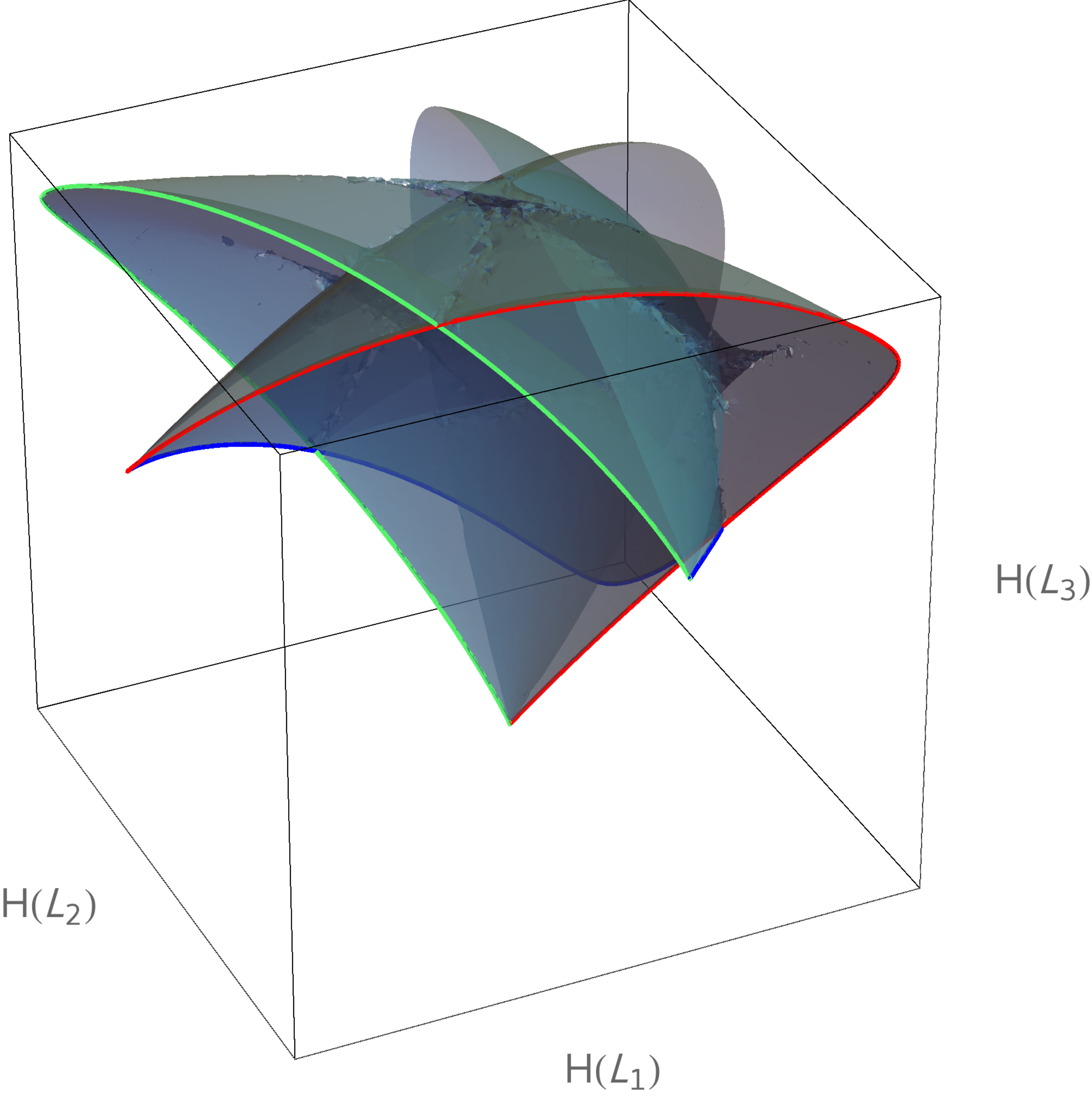}
	\caption{Entropic uncertainty regions for measurements of two and three orthogonal spin components ($s=1$).
		The orange line in the left panel is the Maassen-Uffink bound \eqref{eq:muff}, which is tight only in the case $s=1$.}
	\label{fig:entrobody}
\end{figure}
For angular momentum measurements in two orthogonal directions the connecting unitary operators are rotations, i.e., given as a rotation by $\pi/2$ around the third coordinate axis according to the spin-$s$ representation of $SO(3)$ on $\Cx^{2s+1}$. For arbitrary angles these representations are called Wigner-D matrices \cite{bieden} and will also be used in Sect.~\ref{minmeas}. It turns out that the Maassen-Uffink bound is in general not optimal, but describes precisely the uncertainty region for $s\to\infty$.

For spin $s=1$, the uncertainty region can still be reliably investigated by parameterizing the set of pure states in the $L_3$ eigenbasis. Numerics suggests that real valued states and their permutations characterize the lower bound of the uncertainty region. The resulting uncertainty regions for two and three components are shown in FIG.~\ref{fig:entrobody}, which should be compared directly to FIG.~\ref{fig:spin1mon}(left) and FIG.~\ref{fig:parabolics}(right).

The marked lines in FIG.~\ref{fig:entrobody} correspond to states of a form $\ket{\phi(t)}:=\left(\frac{\cos(t)}{\sqrt{2}},\sin(t),\frac{\cos(t)}{\sqrt{2}}\right)$, written in $L_3$-eigenbasis, and their permutations.  They correspond exactly to the extremal curves found for variance uncertainty as shown in  FIG.~\ref{fig:parabolics}.
Remarkably, however, the ordering of ``uncertainties'' turns out to be different in the two cases. Consider the $L_2$-eigenstates, which are shown in the left panels of both 2D Figures \ref{fig:spin1mon} and \ref{fig:entrobody} as the points on the horizontal axis. The respective $L_1$-probability distribution are $(1/4,1/2,1/4)$ for $m=\pm1$ and
$(1/2,0,1/2)$ for $m=0$. The former distribution has the larger entropy ($(3/2)\log_32 >\log_32$) but the smaller variance ($1/2<1$).

For larger $s$ this inversion no longer holds. FIG.~\ref{fig:muffcompare} shows this effect.
Because we can exchange the roles of $L_1$ and $L_2$ by a unitary rotation, the uncertainty diagrams are symmetric with respect to the diagonal. Therefore the optimal linear bound must be of the form \eqref{eq:muff}, with a suitable $c$. The entropy sums for the eigenstates $\ket m$ with minimal and maximal $\abs m$ are shown in FIG.~\ref{fig:muffcompare}. For all half-integer $s$ and for integer $s>7$ the coherent state $\ket s$ produces not only the lowest variance, but also the lowest entropy. FIG.~\ref{fig:muffcompare} also shows the Maassen-Uffink bound, which has been computed by S{\'a}nchez-Ruiz \cite{ruiz}. It is attained for the overlap of two spin coherent states and is given by
\begin{align}\label{eq:muffconst}
c^2=2^{-2s}\binom{2s}{[s+1/2]}.
\end{align}
$s=1$ seems to be the only case in which the bound is tight.

However, for large $s$, the bound is again optimal, as the following result shows.

\begin{prop}\label{lem:muff}
	In the limit $s\rightarrow\infty$ the optimal lower bound on the entropic uncertainty region of $L_1$ and $L_2$ is given by
	the Maassen-Uffink inequality, which converges to
	\begin{align}
		H(L_1,\rho)+H(L_2,\rho)\geq \frac{1}{2}\label{eq:muffasy}.
	\end{align}
\end{prop}

\begin{proof}
	As a first step we will compute the asymptotic behavior of the bound $-\log_{2s+1}[c^2]$, with $c$ given in $\eqref{eq:muffconst}$. Expanding the central binomial coefficient in factorials and using the Stirling approximation up to order $\log_{2s+1}(s)$ gives
	\begin{align}
		\lim\limits_{s \to \infty}-\log_{2s+1}[c^2]=
		\lim\limits_{s \to \infty}-\log_{2s+1}\left[2^{-2s}\binom{2s}{[s+1/2]}\right]=1/2,
	\end{align}
	which proves the convergence of the Maassen-Uffink bound to the right hand side of \eqref{eq:muffasy}.\\

	In order to show that this bound describes the asymptotic uncertainty region, we have to exhibit sequences of states saturating for every point on the boundary curve. We first show that the endpoint $(0,1/2)$ is asymptotically attained by the $L_1$-eigenstates $\ket{s}$:
	The output entropy of $\ket{s}$ in the $L_1$ basis is always zero, whereas the output entropy in the $L_2$ basis can be evaluated as $H(L_2,\ket{s})=H(L_1,U_3\ket{s})$, with $U_3=e^{-\frac{i\pi}{2}L_3}$.
	Because $\ket s$ has maximal quantum number $m$, the probability amplitudes of $U_3\ket s$ are given by the last column of the Wigner-D Matrix $U_3=D(0,\pi / 2,0)$. By expanding the Wigner-D matrix in terms of Jacobi polynomials \cite{bieden}, one can verify that $\pi_m(U_3\ket s)=\abs{\bra{m}{U_3}\ket{s}}^2$ is  a binomial distribution in $m$, symmetric on the domain $\{-s,s\}$. The entropy of this distribution is
	$-1/2\log_{2s+1}(2\pi e \frac{2s+1}{4})+\mathcal{O}(\frac{1}{n})$, which converges to $\frac{1}{2}$ as $s$ goes to infinity, hence \eqref{eq:muffasy} can be saturated:
	\begin{align}\label{eq:optientro}
		\lim\limits_{s \to \infty}\Bigl(H(L_1,\ket{s}),H(L_2,\ket{s})\Bigr)=\bigl(0,{\textstyle\frac12}\bigr).
	\end{align}

	Finally we have to construct states saturating every point of this bound.
	For $\ket{s}$, still in the eigenbasis of $L_1$, and arbitrary $s$, we define a family of unit vectors
	$\ket{\psi_\alpha}$ as
	\begin{align}
		\ket{\psi_\alpha}=c_\alpha \left(\cos(\alpha)\ket{s}+\sin(\alpha)U_3\ket{s}\right).
	\end{align}
	The $L_1$-outcome probability distribution associated with this vector is
	\begin{align}
		\pi_m(\psi_\alpha)\approx \cos(\alpha)^2
		\delta_{ms}+\sin(\alpha)^2\abs{\bra{m}{U_3}\ket{s}}^2,
	\end{align}
	because the two probability distributions have practically no overlap for large $s$, and hence also $c_\alpha\approx1$. For the same reason we can evaluate the entropy of $\pi(\Psi_\alpha)$ as a sum,
	obtaining
	\begin{align}
		H(L_1,\ket{\psi_\alpha})&\approx
		\sin^2(\alpha)H(L_1,U_3\ket{s})+\frac{H_2(\cos^2(\alpha),\sin^2(\alpha))}{\log_2(2s+1)}
		\approx \frac12 \sin^2(\alpha),
	\end{align}
	where $H_2$ is the binary entropy function.
	In the $L_2$-basis the roles of the two terms in $\psi_\alpha$ are exchanged, and we get
	\begin{align}
		H(L_2,\ket{\psi_\alpha})\approx \frac12 \cos^2(\alpha).
	\end{align}
	Hence the sequence $\psi_\alpha$ realizes the point $\frac12(\sin^2(\alpha),\cos^2(\alpha))$ on the boundary.
\end{proof}
\begin{figure}[t]
	\includegraphics[width=0.4\textwidth]{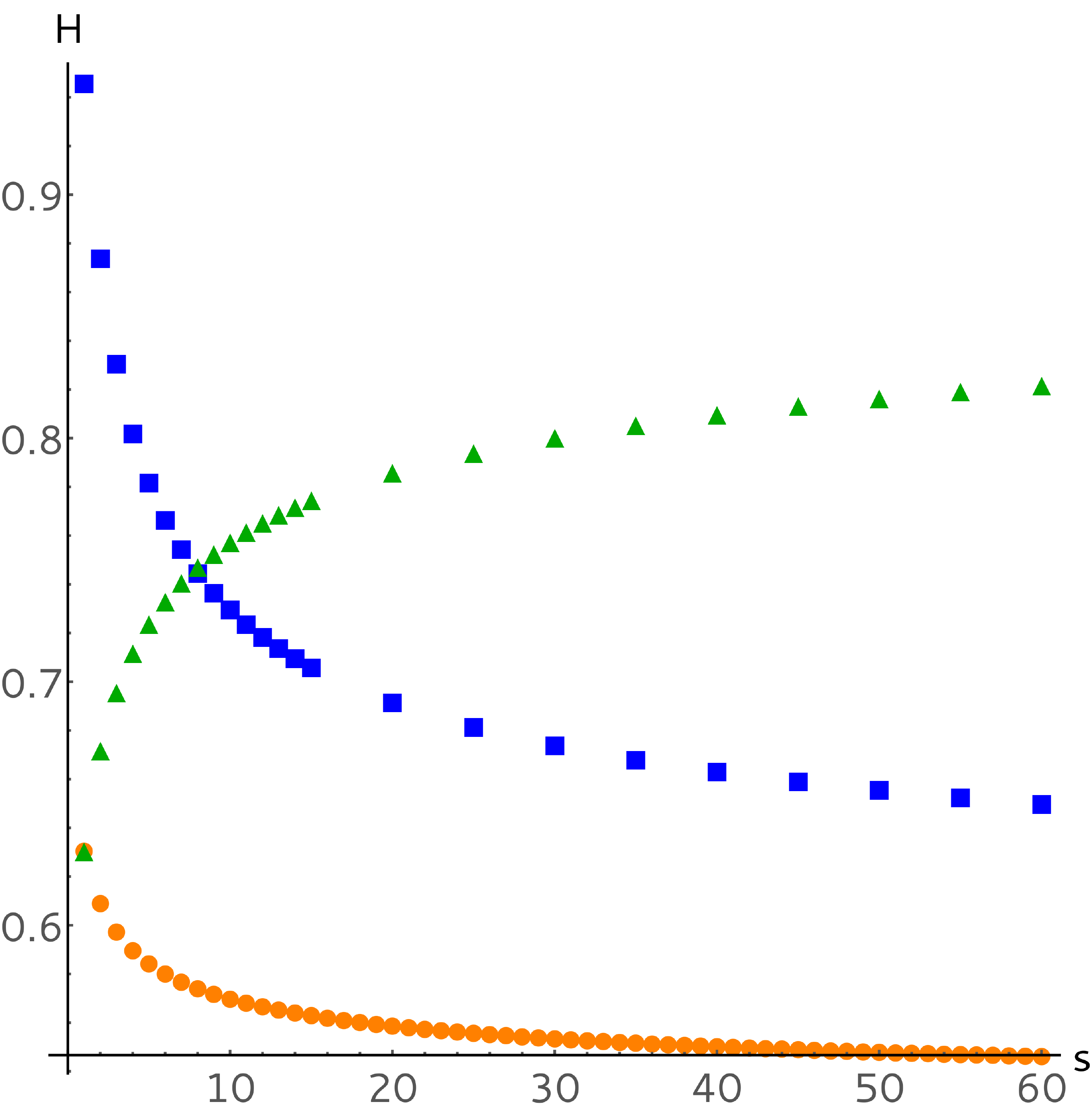}
	\includegraphics[width=0.4\textwidth]{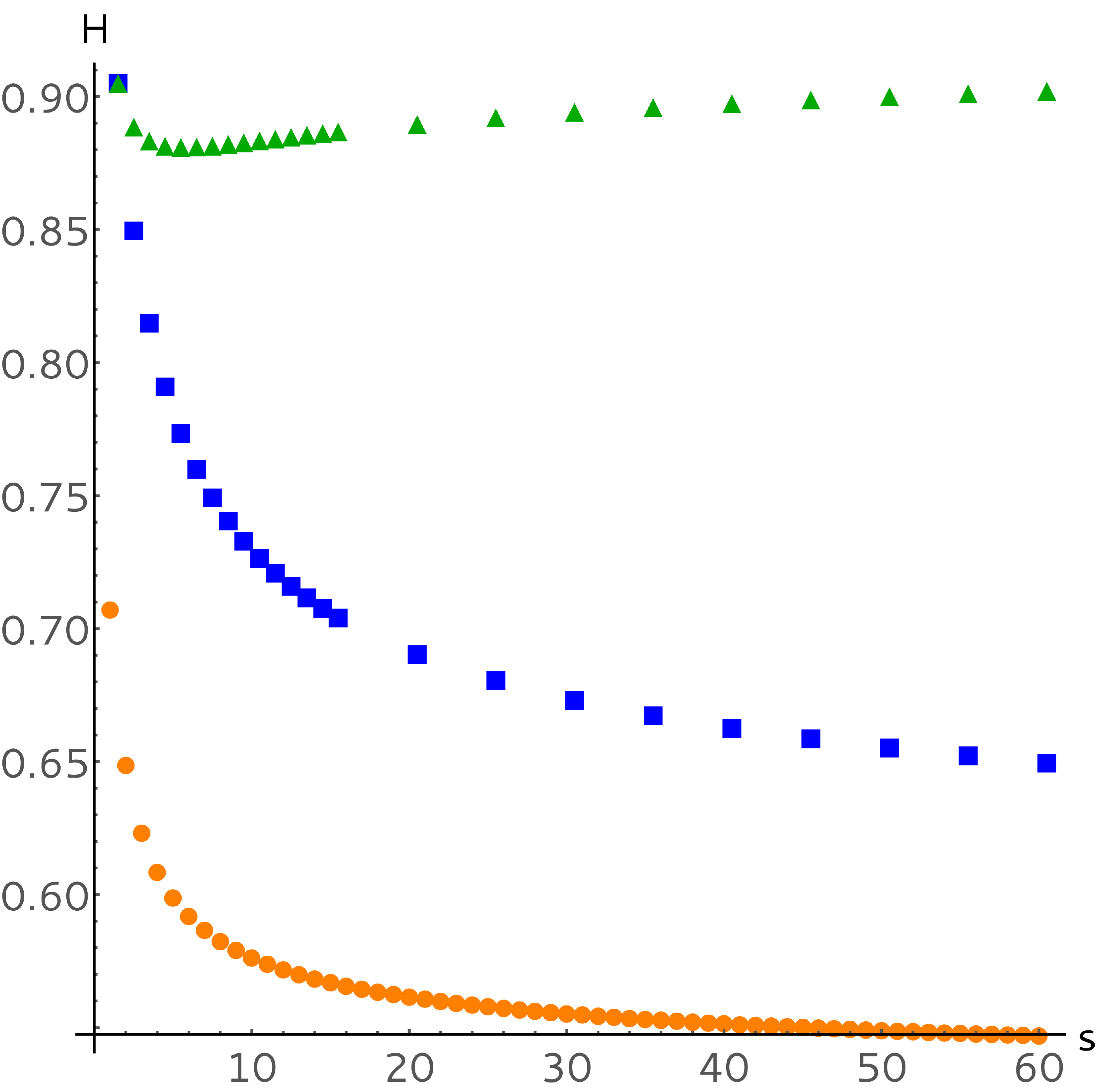}
	\caption{
		Entropic uncertainty sum of the coherent state $\ket{s}$ (blue squares) and the state $\ket{0}$ (green triangles) in comparison to the Maassen-Uffink bound (orange circles), for integer spin (left) and half integer spin (right)
	}
	\label{fig:muffcompare}
\end{figure}

% % % Measurement Uncertainty % % %

\section{Measurement uncertainty}\label{sec:measure}

% % % Introduction % % %
\subsection{Introduction}
As mentioned in the introduction, a measurement uncertainty relation is a quantitative bound on the accuracy with which two observables can be measured approximately on the same device. Already in Kennard's 1927 paper \cite{Kennard} it is clearly stated that in quantum mechanics the notion of a ``true value'' loses its meaning, so that we should not think of ``measurement error'' as the deviation of the observed value from a true value. What we can always do, however, is to compare the performance of two measuring devices, one of which is an (perhaps hypothetical) ``ideal'' measurement and the other an approximate one. The only requirement is that these two measurements give outputs which lie in the same space $X$ and whose distance is somehow defined. A good approximate measurement is then one which will give, on every input state, almost the same output distribution as the ideal one.  This operational focus on the output distributions is also in keeping with the way one would detect a disturbance of the system. Consider how we discover that trying to detect through which of two slits the particles pass disturbs them: The interference pattern, i.e., the output distribution of the interferometer is changed and fringes washed out.

Two related ways to build up a quantitative comparison of distributions, and thereby a quantitative approximation measure between observables, were introduced in the papers \cite{BLW1,BLW2} and applied to the standard situation of a position and a momentum operator. These two notions, called calibration error and metric error will be described in the following subsections. Either way we get a natural figure of merit for an observable $F$ jointly measuring two or more components of angular momentum. In fact, we will only treat the case where $F$ jointly measures {\it all} components. By this we simply mean an observable whose output is not a single number but a vector $\veta$. From this, one derives a ``marginal measurement'' $F^\ve$ of the  $\ve$-component by post-processing, i.e., by taking the $\ve$-component $\ve{\cdot}\veta$ of the output vector as the output of $F^\ve$. These marginals can then be compared with the standard projection valued measurement of the angular momentum component $\eL$. When $D(G,E)$ is the quantity chosen to characterize the error of an observable $G$ with respect to the ideal reference $E$ we get, in our special case,
\begin{equation}\label{Dmax2}
	\Dmax(F)=\max_\ve D(F^\ve,\eL)
\end{equation}
is the desired figure of merit. This is the quantity which we will minimize. But first we have to be more explicit about the two choices for the error quantity $D(G,E)$. This will be done in the next two subsections.

% % % Calibration % % %
\subsection{Calibration Error}
The simpler one assumes that the ``ideal'' observable is projection valued, so that we can produce states which have a very narrow distribution around one of its eigenvalues (or points in the continuous spectrum). In other words, we have some states available which come close to having a ``true value'' in the sense that the ideal distribution is sharp around a known value. A good approximate measurement should then have an output distribution, which is also well peaked around this value. Thus we only have to compare probability distributions to $\delta$-function like distributions, i.e., point measures $\delta_x$ with $x\in X$. This is straightforward, and we set, for any probability measure $\mu$ on the space $X$, $1\leq\alpha<\infty$
\begin{equation}\label{Ddelta}
	D_\alpha(\mu,\delta_x)=\left(\int_X\mu(dy)\ D(y,x)^\alpha\right)^{1/\alpha},
\end{equation}
where $D$ under the integral is the given metric on $X$. This could be called the power-$\alpha$ deviation of $\mu$ from the point $x$. We are mostly interested in {\it quadratic deviations}, i.e., $\alpha=2$. However, in this section we keep $\alpha$ general, which causes no extra difficulty, but makes clear which numbers ``$2$'' arise directly from the role of the averaging power $\alpha$ in \eqref{Ddelta} and similar equations.

We apply this now to $F_\rho$ the output distribution $F_\rho$ obtained by measuring the observable $F$ on the input state $\rho$, and its ideal counterpart $E_\rho$. The $\veps$-deviation or $\veps$-calibration error of the observable $F$ with respect to the ideal observable $E$ is
\begin{equation}\label{eq:caleps}
	\Delta_\alpha^\veps(F,E)=\sup \left\{ D_\alpha(F_\rho,\delta_x) | D_\alpha(E_{\rho},\delta_x) \leq \veps  \right\},
\end{equation}
where the supremum is over all $x\in X$ and ``calibration states'' $\rho$, which are sharply concentrated on $x$ up to quality $\veps$. Note that as a function of $\veps$ this expression is decreasing as $\veps\to0$, because the supremum is taken over smaller and smaller sets. Therefore the limit exists, and we define the {\it calibration error} of $F$ with respect to $E$ by
\begin{equation}\label{eq:cal}
	\Delta_\alpha^c(F,E)=\lim\limits_{\veps \rightarrow 0} \Delta_\alpha^\veps(F,E).
\end{equation}
For observables $E$ with discrete spectrum (like angular momentum components) we can also take $\veps=0$, in \eqref{eq:caleps}, and directly get $\Delta_\alpha^c(F,E)=\Delta_\alpha^0(F,E)$.

% % % Metric % % %
\subsection{Metric error}
A possible issue with the calibration error is that it describes the performance of $F$ only on a very special subclass of states. On the one hand this makes it easier to determine it experimentally, but on the other hand we get no guarantee about the performance of the device on general inputs. Classically this problem does not arise, because broad distributions can be represented as mixtures of sharply peaked ones, and this allows us to give an estimate also on the similarity of output distributions for general inputs. The form of this estimate gives a good hint towards how to define the distance of probability distributions both of which are diffuse. Indeed suppose $\rho$ is an input state such that $\rho=\int\mu(dx)\,\rho_x$, where $\rho_x$ is an $\veps$-calibration state at point $x$, and $\mu$ is an arbitrary probability measure. Then we can define measure $\gamma$ on $X\times X$ by $\gamma(dx\,dy)=\mu(dx)F_{\rho_x}(dy)$, which gives the probability of the joint event of having a ``true value'' $x\in dx$ and finding $y\in dy$. If one integrates out the $x$ variable one gets the output distribution for $\rho$, because $F_\rho$ is linear in $\rho$, and if one integrates out $y$ one gets $\mu$, because each $F_{\rho_x}$ is normalized. To within $\veps$ this is the output distribution $E_\rho$, and with known calibration error we get the bound
\begin{equation}\label{getWasser}
	\int\gamma(dx\,dy)\ D(x,y)^\alpha=\int\mu(dx)\int F_{\rho_x}(dy)\ D(x,y)^\alpha= \int\mu(dx) D_\alpha(F_{\rho_x},\delta_x)^\alpha\leq \Delta_\alpha^\veps(F,E)^\alpha.
\end{equation}
This suggests the following definitions. For two probability distributions $\mu$ and $\nu$ on $X$ we define  a {\it coupling} to be a measure $\gamma$ on $X\times X$ whose first marginal is $\mu$ and whose second marginal is $\nu$. The set of couplings will be denoted by $\Gamma(\mu,\nu)$, and is always non-empty because it contains the product measure. We then define the {\it Wasserstein $\alpha$-distance} of $\mu$ and $\nu$ as
\begin{equation}\label{Wasser}
	D_\alpha(\mu,\nu)= \inf_{\gamma \in \Gamma(\mu,\nu)} \left (\int \gamma(dx\,dy)\ D(x,y)^\alpha \right )^{1/ \alpha}.
\end{equation}
This is also called a transport distance, because of the following interpretation, first seen by Gaspar Monge in the 18th century who considered the building of fortifications. We consider $\mu$ and $\nu$ as some distribution of earth, and the task of a builder who wants to transform distribution $\mu$ into distribution $\nu$. The workers are paid by the bucket and the power $\alpha$ of the distance travelled with each bucket (giving a bonus pay on long distances). The builder's plan is precisely the coupling $\gamma$ saying how many units are to be taken from $x$ to $y$, and the integral is the total cost. The infimum is just the price of the optimal transport plan. The theory of such metrics is well developed, and we recommend the book of Villani \cite{transport} on the subject, but in the present context we only need some simple observations.

With a metric between probability distributions we define the distance of two observables as the worst case distance of their output distributions:
\begin{equation}
	D_\alpha(F,E):= \sup_\rho D_\alpha(F_\rho, E_\rho).
\end{equation}
For the connection between this {\it metric error} and the calibration error introduced above, note first that when $\nu$ is the point measure $\delta_x$, and $\mu$ is arbitrary the product is the only coupling, and the two definitions $D_\alpha(\mu,\delta_x)$ from equations \eqref{Ddelta} and \eqref{Wasser} coincide. Therefore, if $D_\alpha(E_\rho,\delta_x)\leq\veps$, we have
\begin{equation}\label{caleps<met}
	D_\alpha(F_\rho,\delta_x)\leq D_\alpha(F_\rho,E_\rho)+D_\alpha(E_\rho,\delta_x)\leq D_\alpha(F,E)+\veps.
\end{equation}
By taking the supremum \eqref{eq:caleps} and letting $\veps\to0$, we hence have
\begin{equation}\label{cal<met}
	\Delta_\alpha^c(F,E)\leq D_\alpha(F,E).
\end{equation}
Intuitively, this merely indicates that for calibration we test deviations only in the small subset of highly concentrated states. Then \eqref{getWasser} is a partial converse: If $\rho$ has a convex decomposition into $\veps$-concentrated states, $D_\alpha(F_\rho,\mu)\leq\Delta_\alpha^\veps(F,E)$, and since $D_\alpha(\mu,E_\rho)\leq\veps$ we get $D_\alpha(F_\rho,E_\rho)\leq\Delta_\alpha^\veps(F,E)+\veps$. In the classical case such a decomposition always exists, so we have equality $D_\alpha(F_\rho,E_\rho)=\Delta_\alpha^c(F,E)$. In the quantum case, however, we not only have convex mixtures of sharply concentrated states but also coherent superpositions. Using these it is easy to build examples in which \eqref{cal<met} is strict.

There is a second ``quasi-classical'' setting, in which calibration and metric error coincide, and this will actually be used below. This is the case when $F$ and $E$ differ only by classical noise generated in the measuring apparatus. More formally this is described by a transition probability kernel $P(x,dy)$, which is for every $x$ the probability measure in $y$ describing the output of $F$, given that $E$ has been given the value $x$. We can think of this as classical probabilistic post-processing or noise. It is, of course, not necessary that $F$ actually operates in two steps, but only that it could be simulated in this way, i.e., the relation $F(dy)=\int E(dx)\ P(x,dy)$ holds. This is enough to conclude $\Delta_\alpha^c(F,E)= D_\alpha(F,E)$, and to give a formula for both in terms of the size of the noise kernel $P$. In the following Lemma the $E$-essential supremum of a measurable function $f$ with respect to a measure $E$ (denoted $\Eessup_{x\in X}f(x)$) is the supremum of all $\lambda$ such that the upper level set $\{x|f(x)>\lambda\}$ has non-zero $E$-measure. In our application
$E$ is the spectral measure of a component $\eL$, so it is concentrated on the finite set $\{-s,\ldots,s\}$. The essential supremum is then simply the maximum of $f$ over this set.

\begin{lem}\label{lem:postproc}%
	Let $E$ be a projection valued observable on a separable metric space $(X,D)$. Let $F$ be an observable arising from $E$ by post-processing with a transition probability kernel $P$. Then, for all $\alpha$,
	\begin{equation}\label{cal=met}
		\Delta_\alpha^c(F,E)= D_\alpha(F,E) = \left(\Eessup_{x\in X}\int P(x,dy)\ D(x,y)^\alpha \right)^{1/\alpha}.
	\end{equation}
\end{lem}

\begin{proof}Let I, II, III be the three terms in this equation. Then I$\leq$II  is given by \eqref{cal<met}. To show II$\leq$III, note that for any state $\rho$ we get a coupling $\gamma$ between $F_\rho$ and $E_\rho$ by
	$\gamma(dx\,dy)=E_\rho(dx)P(x,dy)$. Hence
	\begin{equation}\label{essup}
		D_\alpha(F_\rho,E_\rho)^\alpha
		\leq \int\gamma(dx\,dy)\ D(x,y)^\alpha=\int E_\rho(dx)\int P(x,dy)\ D(x,y)^\alpha \nonumber.
	\end{equation}
	We introduce the function $f(x)=\bigl(\int P(x,dy)\ D(x,y)^\alpha\bigr)^{1/\alpha}$ and split the integral with respect to $E_\rho$ into an integral over  $X_{>}=\{x|f(x)>t\}$ and an integral over its complement $X_{\leq}$,  where and $t>\Eessup_x f(x)$. Then, by definition of the essential supremum,  $E_\rho$ vanishes on $X_{>}$ and on $X_{\leq}$ the integrand is bounded by $t^\alpha$. Hence $ D_\alpha(F_\rho,E_\rho)^\alpha<t^\alpha$.
	Taking the supremum over $\rho$ and the $\alpha^{\rm th}$ root we get $D_\alpha(F,E)<t$ for every $t>$III, proving II$\leq$III.

	It remains to show that III$\leq$I. This time we pick a $t<\Eessup_x f(x)$. By definition, this means that $E\bigl(\{x|f(x)>t\}\bigr)\neq0$. Now let $\veps>0$. Then because we have assumed $X$ to be separable, it is covered by a countable collection of $\veps$-balls $B_\veps(x_i)=\{y\in X|D(x_i,y)\leq\veps\}$, $ i=1,2,\ldots$. Hence due to countable additivity of $E$
	\begin{equation}\label{essup2}
		0\lneqq E\bigl(\{x|f(x)>t\}\bigr) \leq \sum_i E\bigl(B_\veps(x_i)\cap\{x|f(x)>t\}\bigr).
	\end{equation}
	Hence for some $z=x_i$ we have  $E\bigl(B_\veps(z)\cap\{x|f(x)>t\}\bigr)\neq0$. Since $E$ is projection valued, we find a state $\rho$ such that the probability measure $E_\rho(dx)$ is concentrated on this set. In particular, $D_\alpha(E_{\rho},\delta_z) \leq \veps$. Moreover, $m^\alpha=\int E_\rho(dx)f(x)^\alpha\geq t^\alpha$ and
	\begin{equation}\label{essup3}
		D_\alpha(F_{\rho},\delta_z)^\alpha=\int E_\rho(dx)\int P(x,dy)\ D(z,y)^\alpha=\int E_\rho(dx)g(x)^\alpha,
	\end{equation}
	where the second function defines a function $g$. We interpret these quantities as $L^\alpha$-norms with respect to $E_\rho$, i.e.,
	$m=\norm f_\alpha$ and $D_\alpha(F_{\rho},\delta_z)=\norm g_\alpha$. Then by the triangle inequality
	\begin{equation}\label{essup4}
		D_\alpha(F_{\rho},\delta_z)\geq \norm f_\alpha-\norm{f-g}_\alpha \geq t-\norm{f-g}_\alpha.
	\end{equation}
	To get an upper bound on $\norm{f-g}_\alpha$ note that the expression for $\norm{f-g}_\alpha^\alpha$ is the integral over $E_\rho(dx)$ of
	\begin{equation}\label{essup5}
		\left(\Bigl(\int P(x,dy) D(x,y)^\alpha \Bigr)^{1/\alpha}-\Bigl(\int P(x,dy) D(z,y)^\alpha \Bigr)^{1/\alpha}\right)^\alpha.
	\end{equation}
	Again we can read the outer parenthesis as a difference of norms, namely the $L^\alpha$-norms of the functions $h_x(y)=D(x,y)$ and $h_z(y)=D(z,y)$ with respect to integration by $P(x,dy)$ where $x$ is considered a fixed parameter. But by the triangle inequality for the metric $D$ we have $|h_x-h_z|\leq D(x,z)$ independent of $y$. Since the transition kernel $P$ is a probability measure with respect to $dy$, we find that \eqref{essup5} is bounded above by
	$(\norm{h_x}_\alpha-\norm{h_z}_\alpha)^\alpha\leq\norm{h_x-h_z}_\alpha^\alpha\leq D(x,z)^\alpha$. Hence in \eqref{essup4} we have
	\begin{equation}\label{essup6}
		\norm{f-g}_\alpha^\alpha\leq \int E_\rho(dx) D(x,z)^\alpha\leq\veps^\alpha,
	\end{equation}
	because by construction $E_\rho$ has support in $B_\veps(z)$. Combining the estimates we get $D_\alpha(F_{\rho},\delta_z)\geq t-\veps$. The supremum over all calibrating states can only increase the left hand side, and on the right we use that the only condition on $t$ was that  $t<\Eessup_x f(x)$, so that
	\begin{equation}\label{essup7}
		\Delta_\alpha^\veps(F,E)\geq \Eessup_x f(x) -\veps.
	\end{equation}
	Now III$\leq$I follows in the limit $\veps\to0$.
\end{proof}

In \cite{BLW2} a special case of this Lemma was used to show $\Delta^c=D$ for the position and momentum marginals of a covariant phase space measurement. In that case the noise kernel $P$ is even translation invariant, i.e., the output of the marginal observable can be simulated by just adding some state-independent noise to the output of the ideal position or momentum observable. Such translation invariance makes no sense in the case of angular momentum, since the range of the outputs $m\in\{-s,\ldots,s\}$ of the ideal observable is bounded. This is why the above generalization was needed, in which the noise can depend on the ideal output value. The reason for the existence of a post-processing kernel, however, will be the same as in the phase space case: the covariance of the joint measurement. Roughly speaking this makes the marginal corresponding to $\eL$ invariant under rotations around the $\ve$-axis, which in an irreducible representation means that it must be a function of $\eL$. It is therefore crucial to argue that the optimal joint measurement is covariant, which will be done in the next section.

% % % Covariant % % %
\subsection{Covariant Observables}
Consider a general observable $F$ with outcome space $X$. Suppose some group $G$ acts on $X$, with the action written $(g,x)\mapsto gx$ as usual. Suppose that the group also acts as a symmetry group of the quantum system.  That is, there is a representation $g\mapsto U_g$ of $G$ by operators $U_g$, which are unitary or antiunitary, and satisfy the group law (possibly up to a phase factor). The observable $F$ is then called covariant if
$U_g^*F(S)U_g=F(g\inv S)$ for all $g\in G$ and every measurable set $S$. In other words, shifting the input state by $U_g$ will result in the entire output distribution shifted by $g$. For our purposes it will be convenient to express this in terms of an action $F\mapsto T_gF$ of $G$ on the the set of observables:
\begin{equation}\label{GonFs}
	(T_gF)(S)=U_g F(g\inv S)U_g^*.
\end{equation}
Then the covariant observables are precisely those for which $T_gF=F$ for all $g\in G$.

For angular momentum the group will be the rotation group with its action on the 3-vectors ($X=\Rl^3$). The representation $U$ is then up to a factor $\pm1$. Alternatively, we can take $G$ as the covering group $\SU2$. Since the covariant observables are exactly the same this choice is completely equivalent.

Covariance is certainly a reasonable condition to impose on a ``good'' observable, so it would make sense to study uncertainty relations just for these. However, there is no need for such an ad hoc restriction, because the minimum of uncertainty over all observables is anyway attained on a covariant one. The basic reason for this is that our figure of merit \eqref{Dmax2} does not single out a direction in space, so that it is invariant under the action $T_g$. We therefore only have to show that there is no symmetry breaking, i.e., the symmetric variational problem has a symmetric solution. This will be done in the following Lemma.

\begin{lem}\label{lem:CovSuffices}%
	For any observable $F$ with an outcome set $X=\Rl^3$, and $1\leq\alpha<\infty$ let
	\begin{eqnarray}
		\Dmax(F)&=&\max_\ve D_\alpha(F^\ve,\eL)\label{Dmax}\\
		\Delmax(F)&=&\max_\ve \Delta^c_\alpha(F^\ve,\eL).\label{Delmax}
	\end{eqnarray}
	Then
	\begin{enumerate}
		\item  both these functionals are invariant under the action $T$, and $\Dmax(F)^\alpha$ and $\Delmax(F)^\alpha$ are convex.
		\item
		each of the infima $\inf_F\Dmax(F)$ and $\inf_F\Delmax(F)$ is independent of whether it is taken over all observables or just the covariant ones.
	\end{enumerate}
\end{lem}

\begin{proof}
	By definition of $F^\ve$ we have $(T_RF)^\ve=U_RF^{R\inv\ve}U_R^*$. When $E^\ve$ denotes the spectral measure of $\eL$, the relation $U_R\vL U_R^*=R\inv\vL$ similarly implies $U_RE^{R\inv\ve}U_R^*=E^\ve$. Moreover, due to the supremum over all states, $D_\alpha(F,E)$ does not change, if both observables are rotated with the same unitary. Hence
	\begin{equation}\label{hhh}
		D_\alpha((T_RF)^\ve),E^\ve)=D_\alpha(U_RF^{R\inv\ve}U_R^*,U_RE^{R\inv\ve}U_R^*)=D_\alpha(F^{R\inv\ve},E^{R\inv\ve}).
	\end{equation}
	Hence the supremum over $\ve$ is unchanged and $\Dmax(T_RF)=\Dmax(F)$. For $\Delta$ note that we can carry out the limit $\veps\to0$ directly, because $\eL$ has finite spectrum and states with $D(E^\ve,\delta_x)$ small are norm-close to eigenstates with $x=m$. Hence
	\begin{equation}\label{calibm}
		\Delta_\alpha^c(F^\ve,\eL)^\alpha=\max \Bigl\{D(F^\ve_{\kettbra \psi},\delta_m)^\alpha\Bigm| m,\psi:\ \eL\psi=m\psi\Bigr\}
		=\max_{m,\psi}\int\brAAket\psi{F(dx)}\psi\ (\ve{\cdot}{\vx}-m)^\alpha.
	\end{equation}
	If we now insert the definition of $T_RF$, rewrite the maximum over $\psi$ in terms of $\psi'=U_R^*\psi$, and substitute $\vx'=R\inv\vx$ in the integral, we find $\Delta_\alpha^c((T_RF)^\ve,\eL)=\Delta_\alpha^c(F^{R\inv\ve},E^{R\inv\ve})$, and again $\Delmax(T_RF)=\Delmax(F)$.

	Convexity for $D_\alpha$ follows from the corresponding property of transport distances. Indeed let $\mu=\sum_k\lambda_k\mu_k$ and $\nu=\sum_k\lambda_k\nu_k$ be convex combinations of measures with the same weights $\lambda_k$, and let $\gamma_k$ be a coupling between $\mu_k$ and $\nu_k$. Then $\gamma=\sum_k\lambda_k\gamma_k$ is a coupling of the convex combinations. Moreover,
	\begin{equation}\label{transportConvex}
		D_\alpha(\mu,\nu)^\alpha\leq\int\gamma(dx\,dy)\ D(x,y)^\alpha=\sum_k\lambda_k \int\gamma_k(dx\,dy)\ D(x,y)^\alpha.
	\end{equation}
	Here the first inequality holds because $D_\alpha$ is defined as the infimum over couplings. Then if we take the infimum over each of the $\gamma_k$, we get
	$D_\alpha(\mu,\nu)^\alpha\leq\sum_k\lambda_kD_\alpha(\mu_k,\nu_k)^\alpha$. In particular,  $F\mapsto D_\alpha(F^\ve_\rho,E^\ve_\rho)$ is convex. Since the pointwise supremum of convex functions is convex, $\Dmax$, as the supremum with respect to $\rho$ and $\ve$, is convex. The same observation, applied to \eqref{calibm} with an additional supremum over $\ve$, shows that $\Delmax$ is likewise convex.

	Given any observable $F$ we now form its average $\overline F=\int dR\ T_RF$ with respect to the normalized Haar measure $dR$. Then $\overline F$ is covariant and $\Dmax(\overline F)\leq\int dR\ \Dmax(T_RF)=\Dmax(F)$. Taking the infimum here, and using that $\overline F$ is covariant, and all covariant observables are such averages ($\overline F=F$) gives the second inequality in
	\begin{equation}\label{covmin2}
		\inf_F \Dmax(F)\leq   \inf_{F\ {\rm covariant}} \Dmax(F)\leq \inf_F \Dmax(F),
	\end{equation}
	while the first follows trivially because the covariant observables are a subset. Hence the two infima are the same, and the same argument also applies to $\Delmax$, proving the second claim.
\end{proof}

We could have included also a statement that the infima in this Lemma are all attained. The argument for that is the compactness of the set of observables (in a suitable topology) and the lower semi-continuity of $\Dmax$ and $\Delmax$ which follows, like convexity, from the representation of these functionals as the pointwise supremum of continuous functionals. However, since we will later anyhow explicitly exhibit minimizers, we will skip the abstract arguments. We also remark that one of the main difficulties in the position/momentum case \cite{BLW2} does not arise here: In contrast to the group of phase space translations, the rotation group is compact, so the average is an integral and not an ``invariant mean'', which has the potential of producing singular measures with some support on infinitely far away points.

The main importance of this Lemma is to make the variational problem much more tractable. For covariant observables we have a fairly explicit parameterization,  which allows us to explicitly compute the minimizers. In contrast, for the seemingly easier case of joint measurement of just two components covariance gives only a very weak constraint, and we were not able to complete the minimization.

To develop the form of covariant observables, let us first consider the case when the output vectors have a fixed length $r$. A plausible value would be $r=s$, but we will leave this open. In this case $X$ reduces to a sphere of radius $r$, which is a homogeneous space for the rotation group. We could thus apply the covariant version\cite{Scutaru,screen} of the Naimark \cite{Naimark} dilation theorem to obtain a complete classification. But we do not need this machinery in this elementary case. Note first, that $\tr F(dx)$ is an invariant measure on the sphere, and all probability densities $\tr\rho F(dx)$ are bounded by this measure (since $\rho$ is bounded). Hence, for every $\rho$ there is a bounded probability density with respective to the uniform measure. Since this depends linearly on $\rho$, it is given by a bounded operator depending on $\vx$. The $\vx$-dependence is then completely resolved by covariance, and it is sufficient to know this density at one point, say the north pole $r\north$. Moreover, by covariance this density must commute with the rotations around $\north$, and is hence a linear combination of the eigenstates $\kettbra n$ with $n\in\{-s,\ldots,s\}$. The only choices to be made are hence the coefficients $F_n$ of this liner combination. We write the resulting observable in terms of its integrals with an arbitrary function $h$ on the sphere:
\begin{eqnarray}\label{Fcovr}
	\int F(d\vx) h(\vx)&=& (2s+1)\int dR\ U_R \Bigl( \sum_n F_n\kettbra n\Bigr) U_R^*\ h(r\,R\north) \\
	&=& (2s+1)\int\frac{\sin\theta\,d\theta\,d\phi}{4\pi}\sum_nF_n U_{\theta,\phi}\kettbra nU_{\theta,\phi}^*
	\ h(r\,\ve_{\theta,\phi})
	\nonumber
\end{eqnarray}
Here the first integral is over Haar measure on the rotation group (or $\SU2$), whereas the second is expressed in polar coordinates $(\theta,\phi)$ on the sphere, with $\ve_{\theta,\phi}$ the corresponding unit vector, and $U_{\theta,\phi}$ some rotation rotating the north pole $\north$ to $\ve_{\theta,\phi}$. It does not matter which rotation we choose, because $\kettbra n$ is invariant with respect to rotations around the 3-axis. The two expressions are related by introducing Euler angles on the rotation group and integrating out the initial rotation around the 3-axis. The normalization factor $(2s+1)$ is chosen so that
the constraints on $F_n$ are exactly $F_n\geq0$ and $\sum_nF_n=1$, i.e., the observable is represented as a convex combination of observables using only one fixed $\kettbra n$ as the density. What changes when $r$ is not fixed is simply that we get an additional integration over $r$, where the $F_n$ may also depend on $r$. Effectively, we get a probability measure $F_n(dr)$ on $\{-s,\ldots,s\}\times \Rl_+$ and the second version of \eqref{Fcovr} just becomes
\begin{equation}\label{Fcov}
\int F(d\vx) h(\vx)= \sum_n\int F_n(dr)\
(2s+1)\int\frac{\sin\theta\,d\theta\,d\phi}{4\pi} U_{\theta,\phi}\kettbra nU_{\theta,\phi}^*\ h(r\,\ve_{\theta,\phi}).
\end{equation}

The criterion for joint measurability does not depend on the full observable, but only on the marginals along the various directions $\ve$. It is one of the direct consequences of covariance, evident from the proof of Lemma~\ref{lem:CovSuffices}, that $D_\alpha(F^\ve,E^\ve)$ and $\Delta^c_\alpha(F^\ve,E^\ve)$ do not depend on $\ve$. We will therefore only consider the case $\ve=\north$ in the following. In \eqref{Fcov} this just means that we specialize to functions of the form $h(\vx)=h_1(\vx{\cdot}\north)$ with $h_1:\Rl\to\Rl$. Thus in the integrand we get $h_1(r \vx_{\theta,\phi}{\cdot}\north)=h_1(r\cos\theta)$, which no longer depends on $\phi$. We can therefore carry out the $\phi$-integration. The resulting operator will commute with rotations around the 3-axis, so we can express it as a linear combination of operators $\kettbra m$:
\begin{eqnarray}\label{Fecov}
	\int F^\north(dx)h_1(x)&=&\sum_m\kettbra m\ \int P(m,dx)h_1(x)\quad\mbox{with}\nonumber\\
	\int P(m,dx)h_1(x)&=&\sum_n\int F_n(dr)\ (2s+1)\int\frac{\sin\theta\,d\theta}2\ \abs{\brAAket m{U_\theta}n}^2\ h_1(r\cos\theta).
\end{eqnarray}
The first line establishes the connection with the premise of Lemma~\ref{lem:postproc}: For covariant observables the marginals can be simulated by an exact measurement of $m$, with a post-processing kernel $P$. Therefore, for covariant $F$ we have $\Dmax(F)=\Delmax(F)$. Since by Lemma~\ref{lem:CovSuffices} the infimum of this quantity, say $\Delta_{\rm min}(s)$, is the same as the minimization over all observables. Therefore we can state the measurement uncertainty in the forms
\begin{equation}\label{measuncert}
	\Dmax(F)\geq \Delta_{\rm min}(s)\quad\text{and }\quad \Delmax(F)\geq \Delta_{\rm min}(s)
\end{equation}
for {\it all} observables $F$, whether covariant or not. We will now compute $\Delta_{\rm min}(s)$, and show that both minima are attained for a unique covariant observable.

% % % Minimal Uncertainty % % %
\subsection{Minimal Uncertainty}\label{minmeas}
While the above holds for arbitrary exponent $\alpha$, we will now restrict to the standard variance case, i.e. $\alpha=2$. \\
So far, we have derived that the optimal observable $\mathbf{F}$ is covariant, leading to the parametrization \eqref{Fecov}. In particular, $F^\ve$ arises from $\eL$ by a transition probability kernel, so that metric and calibration error coincide. In the sequel, we will therefore only consider the calibration error, which is easier to evaluate. By covariance the calibration error $\Delta_2^c(F^\ve,\eL)$ is independent of $\ve$, so we can choose $\ve=\ve_3$. Observing that for discrete valued observables we can take $\veps=0$ in \eqref{eq:caleps}, and we get from \eqref{Fecov} the basic figure of merit
\begin{eqnarray}
	\Delta_2^c(F^\ve,\eL)^2 &=& \Delta_2^c(F^3,L_3)^2=
	\max_m \int \tr \left(\kettbra{m} F(d\veta)\right)  (\ve_3 \cdot \veta - m )^2\nonumber\\
	&=&\max_m  \;  \sum_{n=-s}^s \int\limits_0^\infty F_n(dr) \; (2s+1) \int_0^\pi\frac{\sin\theta\,d\theta}2 \; |\brAAket n {U_{\theta}} m |^2 (r \cos \theta -m)^2.
	\label{eq:mesp}
\end{eqnarray}
Before calculating the optimal case, we introduce the following Lemma, which provides a more manageable expression $I(s,r,n,m)$ for the integral over $\theta$, such that \eqref{eq:mesp} reads
\begin{equation}
	\Delta_2^c(F^3,L_3)^2=\max_m  \sum_{n=-s}^s \int\limits_0^\infty F_n(dr) \; I\left(s,r,n,m\right).\label{eq:mesp2}
\end{equation}

\begin{lem}\label{lem:MeasI}
	For $s\geq 1$ the integral $I(s,r,n,m)$ can be written as\\
	\begin{equation}
		I(s,r,n,m)=A_s(r,n)+m^2 B_s(r,n),
	\end{equation}
	where
	\begin{equation}
		A_s(r,n)=\frac{r^2 s (s+1) \left(-2 n^2+2 s (s+1)-1\right)}{s (s+1) (2 s-1)  (2 s+3)}
	\end{equation}
	and
	\begin{equation}
		B_s(r,n)=\frac{6 n^2 r^2-2 n r (4 s (s+1)-3)+s (s+1) \left(-2 r^2+4 s (s+1)-3\right)}{s (s+1) (2 s-1)  (2 s+3)}.
	\end{equation}
	For $s=\frac 1 2$ we have
	\begin{equation}
		I \left (\frac 1 2,r,\pm\frac 1 2,m \right )=\frac{r^2}{3}\mp\frac{r}{3}+\frac{1}{4}.
	\end{equation}
\end{lem}

\begin{proof}
	Here we have to solve the integral
	\begin{equation} I\left(s,r,n,m\right)=(2s+1)
		\int\limits_{0}^{\pi}
		\frac{\sin{\theta} d\theta }{2} \; \big|d_{nm}^{(s)}(\theta) \big|^2 (r \cos \theta - m)^2 ,
	\end{equation}
	where $d_{nm}^{(s)}(\theta)=\brAAket n {U_{\theta}} m$ is the small Wigner $d$-matrix \cite{bieden}.
	First we expand $d_{nm}^{(s)}(\theta)$ in terms of the Jacobi polynomials:
	\begin{equation}
		d_{nm}^{(s)}(\theta)	=  \left [ \frac{(s+m)!(s-m)!}{(s+n)!(s-n)!} \right]^{\frac 1 2} \left (\sin \frac{\theta}{2} \right )^{m-n} \left (\cos \frac{\theta}{2} \right )^{n+m} P^{(m-n,m+n)}_{s-m}(\cos \theta).
	\end{equation}
	In the following we use a recurrence relation for the Jacobi polynomials. This three-term relation does not hold for $s=1/2$, so that we have to treat this case separately: The integrals $I\left(\frac 1 2 , r, n , \pm \frac 1 2\right)$ can indeed be  calculated directly from the above expression, and the results are given in the statement of the Lemma.  From now on we assume $s\geq 1$.

	We can simplify some case distinctions by introducing $k=\min (s+m,s-m,s+n,s-n)$ and substitute the arising positive integers $s-m,m-n,m+n$ according to Tab.~\ref{tab:wignerd}, which is possible due to the symmetries of the Wigner $d$-matrix \cite{bieden}.  Our expression then depends implicitly on $(s, m,n)$ through $(\mu,\nu,k)$:
	\begin{table}[b]
		\begin{tabular}{cccc}
			& $k$ &\; $\mu$ \;&\; $\nu$ \\
			\hline I) & \;$s+m$ \;&\;$ n-m$ &\;$ -n-m$ \\
			\hline II) & \;$s-m $\;&\;$ m-n$ &\; $n+m$ \\
			\hline III) &\;$ s+n$ \;&\;$ m-n$ &\;$ -n-m$ \\
			\hline IV) & \;$s-n $\;&\;$ n-m $&\;$ n+m$ \\
			\hline
		\end{tabular}
		\caption{Wigner-d matrix substituion. }
		\label{tab:wignerd}
	\end{table}
	\begin{equation}
	d_{nm}^{(s)}(\theta)	=\left [ \frac{k!(2s-k)!}{(k+\mu)!(k+\nu)!} \right]^{\frac 1 2} \left (\sin \frac{\theta}{2} \right )^\mu \left (\cos \frac{\theta}{2} \right )^\nu P^{(\mu,\nu)}_k(\cos \theta).
	\end{equation}
	Substituting $x=\cos \theta$ yields
	\begin{equation} I\left(s,r,n,m\right)=
		\kappa(j,k,\mu,\nu) \int\limits_{-1}^{1}dx
	\left (1-x \right )^\mu \left (1+x \right )^\nu P^{(\mu,\nu)}_k(x) P^{(\mu,\nu)}_k(x) \big (rx-m \big)^2
	\end{equation}
	where
	\begin{equation}
		\kappa(s,k,\mu,\nu)=\left [\frac{(2s+1) k!(2s-k)!}{2^{\mu+\nu+1}  (k+\mu)!(k+\nu)! }\right].
	\end{equation}
	This integral can be solved by expanding the factor $(rx-m)^2$, using the Jacobi polynomial orthogonality relation \cite{szego}
	\begin{equation}
		\int_{-1}^1 (1-z)^a (1+z)^b P_m^{(a,b)}(z) P_n^{(a,b)}(z) \, dz =\underbrace{\frac{ \left(2^{a+b+1} \Gamma (a+n+1) \Gamma (b+n+1)\right)}{n! (a+b+2 n+1) \Gamma (a+b+n+1)}}_{\omega(n,a,b)}\delta _{m,n}
	\end{equation}
	and the recurrence relation \cite{szego}
	\begin{equation}
		\begin{aligned}
			z P_n^{(a,b)}(z)=\underbrace{\frac{\left(b^2-a^2\right) }{(a+b+2 n) (a+b+2 n+2)}}_{\alpha(n,a,b)} P_n^{(a,b)}(z)
			+\underbrace{\frac{(2 (a+n) (b+n)) }{(a+b+2 n) (a+b+2 n+1)}}_{\beta(n,a,b)}P_{n-1}^{(a,b)}(z) \\
			+\underbrace{\frac{(2 (n+1) (a+b+n+1)) }{(a+b+2 n+1) (a+b+2 n+2)}}_{\gamma(n,a,b)} P_{n+1}^{(a,b)}(z).
		\end{aligned}
	\end{equation}
	Then the expressions in the Lemma arise by simplifying the corresponding polynomials:
	\begin{equation}
		\begin{aligned}
			\label{eq:integral}
			I\left(s,r,n,m\right)=&   \; \kappa(s,k,\mu,\nu)\bigg [r^2 \big (\alpha(k,\mu,\nu)^2\;\omega(k,\mu,\nu)+ \beta(k,\mu,\nu)^2 \omega(k-1,\mu,\nu)+\gamma(k,\mu,\nu)^2\;\omega(k+1,\mu,\nu) \big )  \\
			& -2rm \; \alpha(k,\mu,\nu)\;\omega(k,\mu,\nu)+    m^2 \;\omega(k,\mu,\nu) \bigg ]  =   A_s(r,n)+m^2 B_s(r,n).
		\end{aligned}
	\end{equation}
\end{proof}

We will use this Lemma to simplify the minimization over $m$. Moreover, the integral over $r$ and the sum over $n$ can be seen as taking a convex combination over two-dimensional vectors $\bigl(A_s(r,n), B_s(r,n)\bigr)$. Hence optimizing $F$ can be analyzed geometrically in terms of the set of such pairs (see Figs.~\ref{fig:messcurves}).
This is solved in the next theorem, whose results are visualized in FIG.~\ref{fig:meas}.
\begin{figure}[h]
	\includegraphics[width=0.49\textwidth]{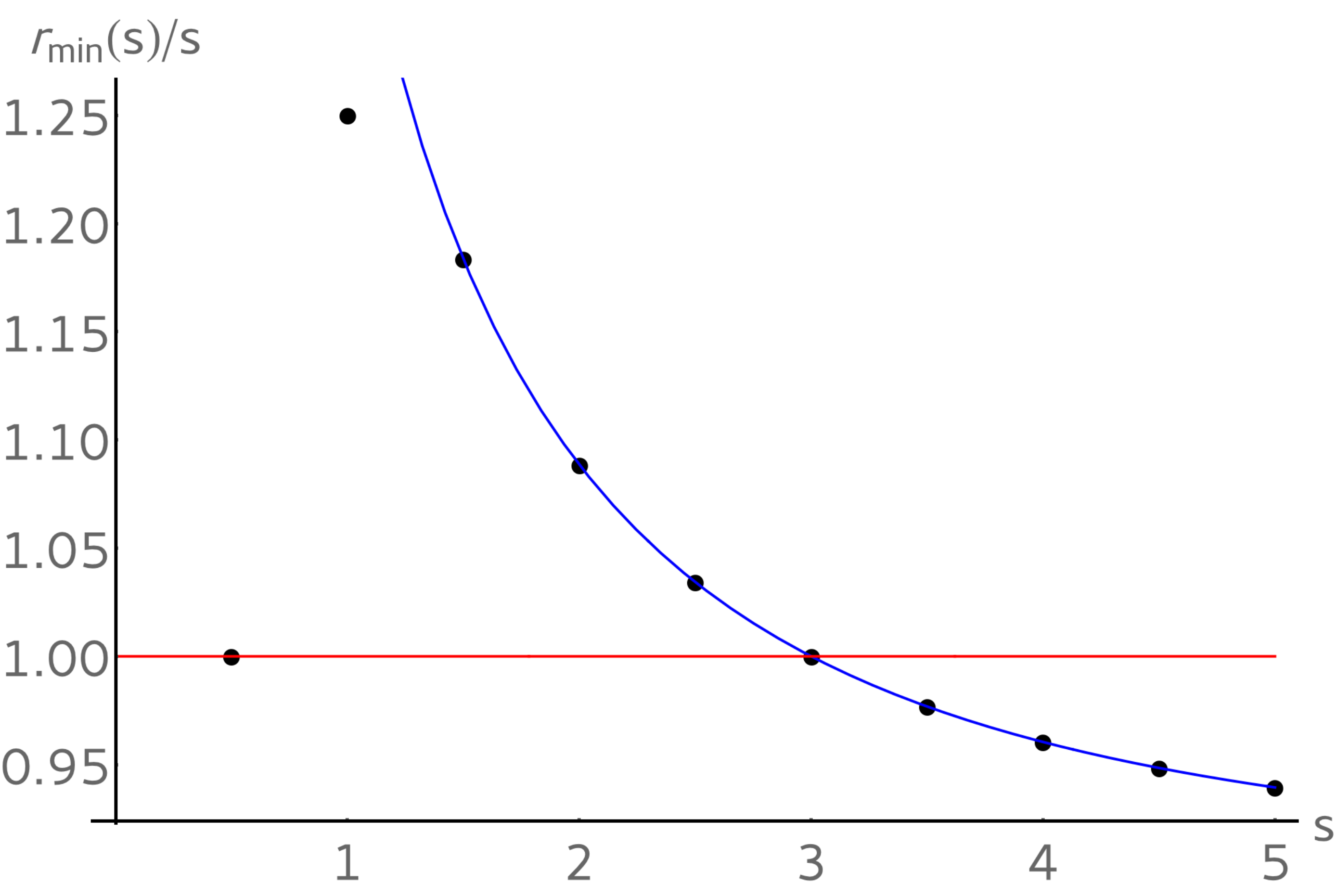} \
	\includegraphics[width=0.49\textwidth]{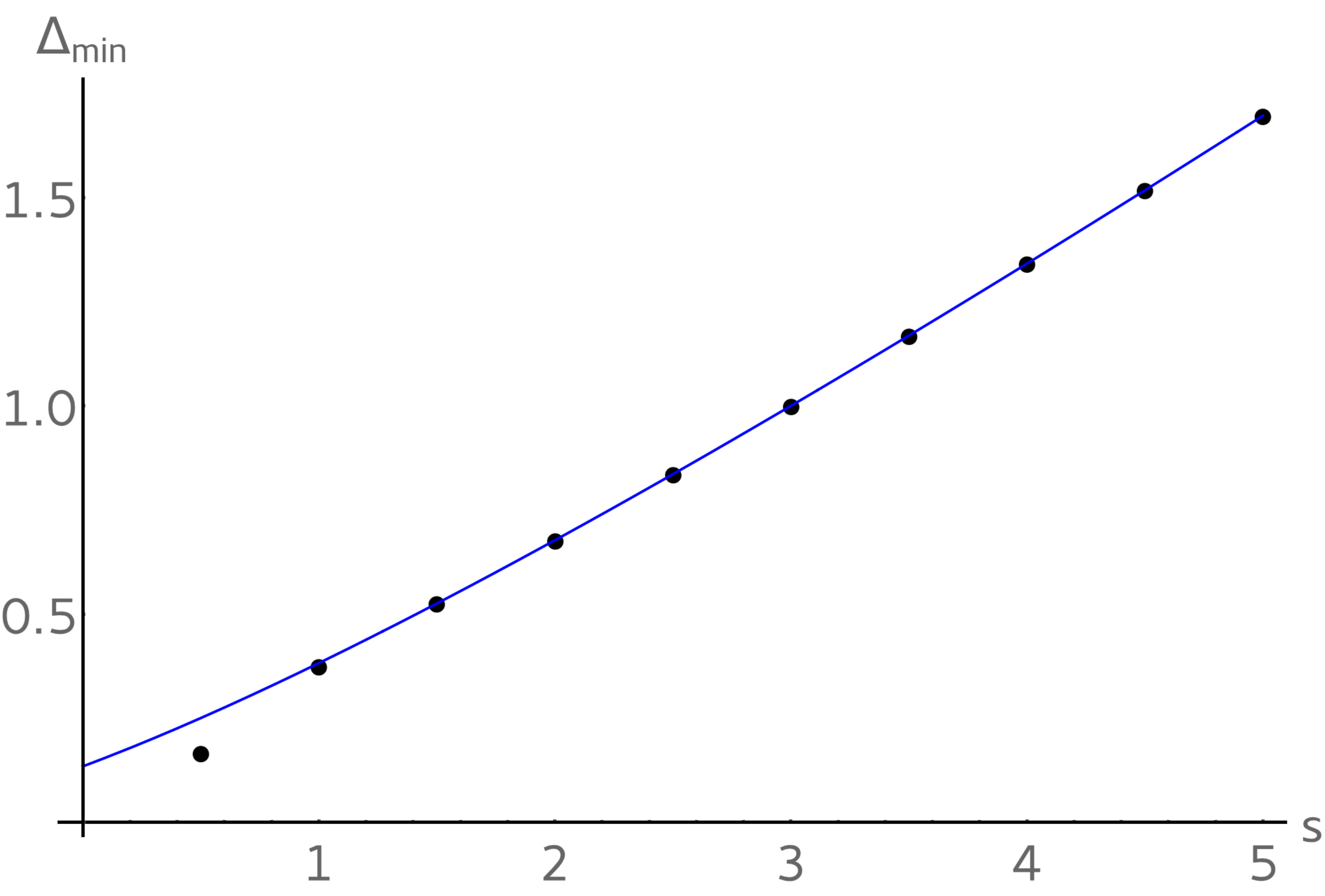}
	\caption{Optimal radii $r_{\rm min}(s)$ and measurement uncertainties $\Delta_{\rm min}$ according to Eq.~\eqref{eq:mur}. The radii are scaled by~$s$, showing that for $1\leq s\leq 5/2$ the outputs of the optimal observable are vectors of modulus $>s$. For larger $s$, the output vectors are shorter than $s$, and, after a minimum around $s=27/2$ we have $r_{\rm min}(s)/s\to1$.  In both panels, the functional expression valid for $s>1$ is plotted in blue. }
	\label{fig:meas}
\end{figure}

\begin{thm}\label{thm:mur}
	The minimal measurement uncertainty $\Delta_{\rm min}(s)$ in the sense of \eqref{measuncert} is attained at a unique covariant observable, for which $F_n(dr)$ is a point measure at $n=s$ and $r=r_{\rm min}$,
	with
	\begin{equation}\label{eq:mur}
		r_{\rm min}(s)=\Cases{1/2&\\ 5/4&\\ \left(2 s-\sqrt{2 s+3}+3\right)/2&} \quad\text{and}\quad
		\Delta_{\rm min}(s)=\Cases{1/6&\strut\qquad s=1/2\\
		3/8&\qquad s=1\\ \left(s-\sqrt{2 s+3}+2\right)/2&\strut\qquad s>1.}
	\end{equation}
	Except for $s=1$, the maximum over $m$ in \eqref{eq:mesp2} is trivial for the optimal observable, i.e., the calibration error is the same for all calibration inputs $m$.
\end{thm}

\begin{proof}
	We consider first the case $s\geq1$. For arbitrary $s\geq1$ we reformulate the problem using that $\Delta(F^3,L_3)^2$ is a convex combination of the functions $A_s(r,n)$ and $B_s(r,n)$. Here we must find the best $n$ as well as the probability distribution $F_n(dr)$ for the \textit{worst} $m$. We denote the convex set of all possible combinations by $\Omega \in \Rl^2$, i.e.
		\begin{equation}
			\begin{aligned}
			\Omega:&= \left\{ (a,b) \;\bigg \vert \;    a= \int F_n(dr) \; A_s(r,n) \text{ and } b= \int F_n(dr) \; B_s(r,n) \right\} \\
			&= \conv\Bigl\{\bigl( A_s(r,n),B_s(r,n)\bigr)\Bigm| r>0,n=-s,\ldots,s\Bigr\}.
			\end{aligned}
		\end{equation}
	All information about a possible observable is now contained in $(a,b)\in \Omega$.
	Furthermore a maximum over $m$ is part of the definition of $\Delta^2_c(F^3,L_3)^2$, so we can rewrite it as a functional on $\Omega$:
	\begin{equation}
		K: \Omega \to \Rl \quad K(a,b)=
		\begin{cases}
			a + s^2 b & \text{if }  b > 0 \\
			a \;         & \text{if } b \leq 0 \text{ and } s \in \Nl \\
			a \;  + \frac 1 4 b       & \text{if } b \leq 0 \text{ and } s + \frac 1 2 \in \Nl
		\end{cases}.
	\end{equation}
	The problem is now to minimize the functional $K(a,b)$. Since, for general $n$ and $s$, $\Omega$ is hard to describe, we choose the following strategy, which is illustrated in the left panel of FIG.~\ref{fig:messcurves}.

	We will show that, for $s>1$,  $K$ takes its minimum at
	\begin{equation}\label{vmin}
		\mathbf{v}= (A_s(r_{\rm min}(s),s),0)
		= \left( \frac{1}{2} \left(s-\sqrt{2 s+3}+2\right),0\right )
	\end{equation}
	by constructing a line $\phi$ which separates the set $\Omega$ from the convex level set
	$	K_\mathbf{v}:=\lbrace (a,b) \in \mathbb{R}^2 \; | \; K(a,b) \leq  K(\mathbf{v})\rbrace$.\\
	For this we will take the line $\phi$ to be the tangent of the curve $r\mapsto(A(r,s),B(r,s))$ at the point $\mathbf{v}$.
	The normal $\mathbf{u}$ of $\phi$ is
	\begin{equation}
		\mathbf{u} =\left (\frac{2 \sqrt{2 s+3}}{2 s^2+5 s+3},\frac{2 s-\sqrt{2 s+3}+3}{2 s+3}\right )
	\end{equation}
	and
	$\Phi:=\{\mathbf{x} \in\Rl^2 \; |\; \mathbf{x}\cdot\mathbf{u}\geq \mathbf{v} \cdot \mathbf{u} \}$
	is the half plane above $\phi$.

	Now we show that $\Omega \subset \Phi$, i.e.,
	\begin{equation}
		g_s(n,r):=
		(A_s(r,n),B_s(r,n)) \cdot \mathbf{u} - (A_s(r_{\rm min}(s),s)),0) \cdot \mathbf{u} \geq 0,
		\label{subset}
	\end{equation}
	with equality iff $n=s$ and $r=r_{\rm min}(s)$.
	Note that the function $g_s(n,r)$ is quadratic in $r$, so for verifying \eqref{subset} it is sufficient to show that $\partial_r^2 \;g_s(n,r)\geq0$ and $g_s(s,r_n)\geq0$, where $r_n$ is the stationary point determined by $\partial_r\;g_s(n,r_n)=0$.\\
	Indeed we have
	\begin{align}
		\partial_r^2 \;  g_s(n,r)
		&= 4 s^2 \sqrt{2 s+3} -4 s^2 + 12 n^2 - 4 n^2 \sqrt{2 s+3} - 4 s
		\nonumber\\
		&=4\bigl( (s^2-n^2) (\sqrt{2 s+3} - 3) + 2 s^2 - s \bigr)
		\label{posi}.
	\end{align}
	Now for $s\geq3$ the factor with the square root is positive, and since $\abs n<s$ we can estimate the expression \eqref{posi} as $\geq2s^2-s\geq0$. For $1\leq s\geq3$, the minimum with respect to $n$ is assumed at $n=0$, for which \eqref{posi} can be evaluated explicitly, and shown to be positive.

	The stationary points $r_n$ can be straightforwardly computed as
	\begin{align}
		r_n=\frac{n (2 s-1) \left(-2 s+\sqrt{2 s+3}-3\right)}{2 \left(n^2 \left(\sqrt{2 s+3}-3\right)+s \left(-\sqrt{2 s+3} s+s+1\right)\right)}
	\end{align}
	with
	\begin{align}
		g_s(n,r_n) =\frac{s^2-n^2} {s \left(s\sqrt{2 s+3}-s-1\right) - n^2 \left(\sqrt{2 s+3}-3\right) }.\label{eq:label1}
	\end{align}
	The denominator of $g_s(n,r_n)$ in \eqref{eq:label1} is again the quadratic coefficient of the expression \eqref{subset}, i.e.,  \eqref{posi} which was already shown to be positive. Hence $\Omega\subset\Phi$.

	Finally we have to certify that
	$K_\mathbf{v}\cap \Phi=\lbrace\mathbf{v}\rbrace$.
	Using the gradient of $\phi$ and comparing it to the linear boundaries of $K_{\mathbf{v}}$ we get the conditions
	\begin{align}
		-4 < - \frac{1+\sqrt{2s+3}}{(1+s)^2} < - \frac 1 {s^2}\quad & \text{for }s+\frac 1 2 \in \Nl
		\nonumber\\
		- \infty < - \frac{1+\sqrt{2s+3}}{(1+s)^2} < - \frac 1 {s^2} \quad &\text{for } s \in \Nl
	\end{align}
	which are true for $s>1$. This concludes the proof for $s>1$.
	\begin{figure}[t] \centering
		\includegraphics[width=0.35\textwidth]{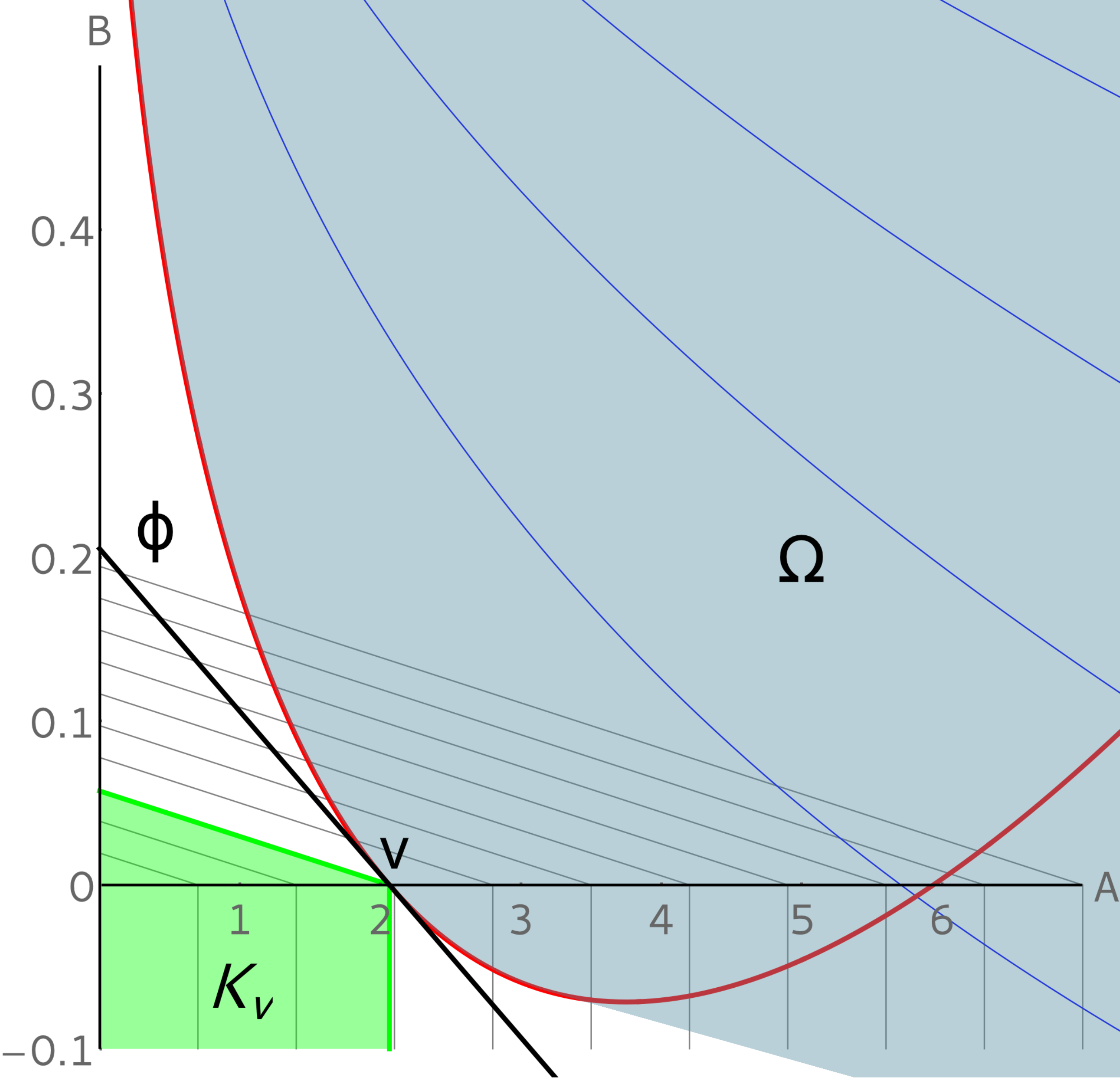} \; \;
		\includegraphics[width=0.336\textwidth]{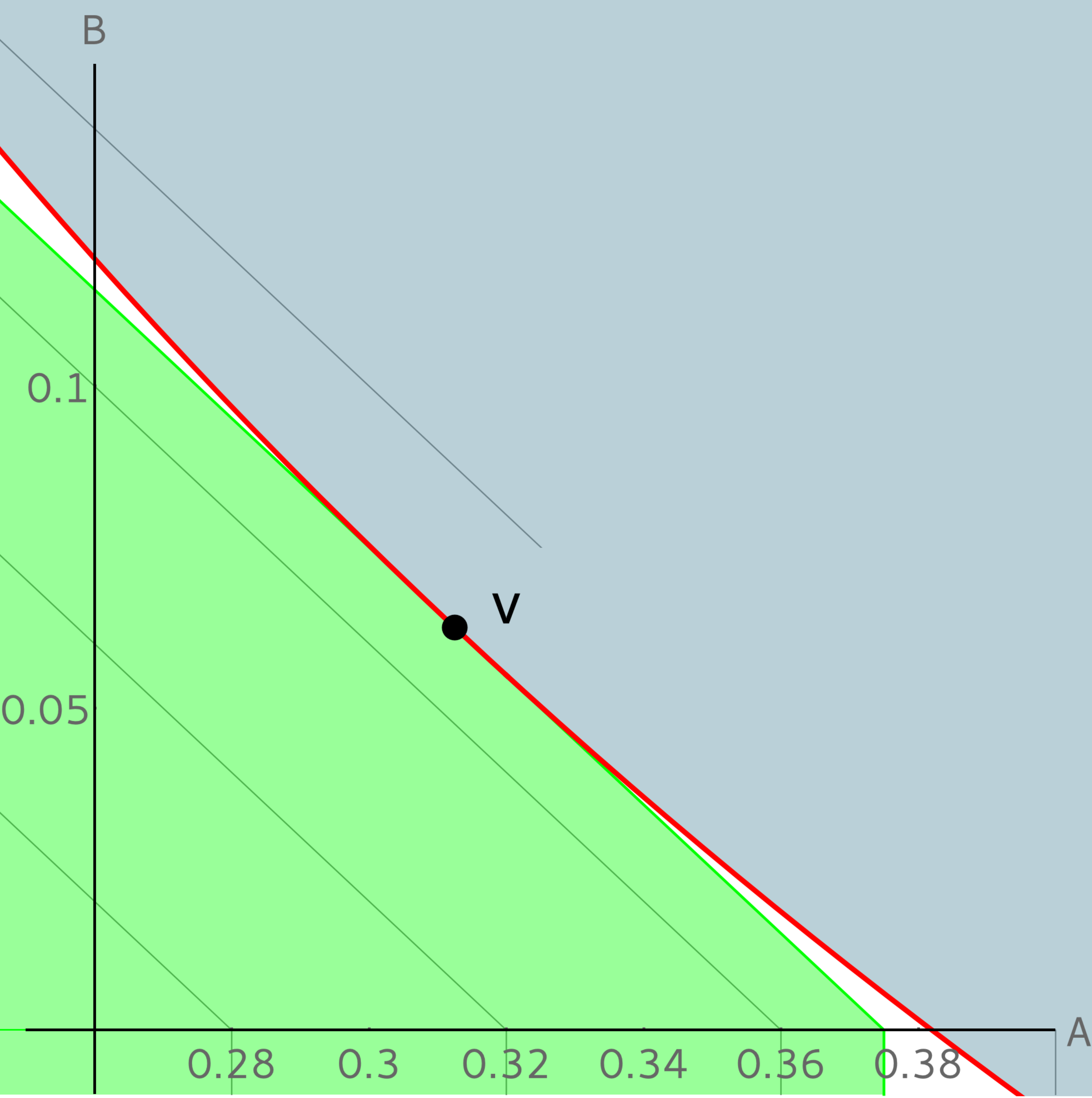}
		\caption{Left: Example for $s=6$ and the $n=s$ curve marked in red. Right: Intersection of $K_{\mathbf{v}}$ with $\phi$ in the $s=1$ case. }
		\label{fig:messcurves}
	\end{figure}

	For $s=1$, this last step of the above proof fails. Indeed, the level set $K_\mathbf{v}$ which was determined by taking that point $v$ on the horizontal axis which is also on the boundary of $\Omega$ does intersect $\Omega$, as can be seen in FIG.~\ref{fig:messcurves}. We therefore have to take a level set of $K$ for a slightly smaller value. Since the tangents of the level sets are all the same for $b>0$, we can readily find the level set which is tangent to $\Omega$. This gives the optimal radius $r_{\rm min}(s=1)=5/4$, and the
	\begin{equation}\label{vmin1}
		\mathbf{v}= (A_s(r_{\rm min}(s),s),B_s(r_{\rm min}(s),s))=(5/16,1/16).
	\end{equation}
	Analogous to the above arguments, one verifies easily that $\{\mathbf{v}\}=K(\mathbf{v})\cap\Omega$.

	Finally, for $s=1/2$ one can draw the conclusion directly from the form of $I(1/2,r,n,m)$ given in Lemma~\ref{lem:MeasI}: This expression does not depend on $m$, and has a unique global minimum at
	$r_{\rm min}(1/2)=1/2$ and $n=+1/2$, where the optimal probability measure $F$ must therefore be concentrated.

	In all cases, the optimal value $\Delta_{\rm min}(s)$ is computed by substituting the obtained optimal $r_{\rm min}(s)$ and $n=s$ in \eqref{eq:mesp2}.
\end{proof}

% % % Summary % % %

\section{Conclusions and Outlook}
Uncertainty relations can be built for any collection of observables. In this paper we provided some methods, which work in a general setting, but chiefly looked at angular momentum as one of the paradigmatic cases of non-commutativity in quantum mechanics.

The basic mathematical methods are well-developed for the case of preparation uncertainty, so that even in a general case the optimal tradeoff curves can be generated efficiently. We resorted to numerics quite often, since it turns out that the salient optimization problems can rarely be solved analytically for general $s$. One of the features one might hope to settle analytically in the future is the asymptotic estimate $c_2(s)\propto s^{2/3}$ which comes out with a precision that suggests an exact result.

Much is left to be done for entropic uncertainty. Here we gave only some basic comparisons to the variance case. It would be interesting to see whether the entropic relations can be refined to the point that they can be used to derive sharp variance inequalities as Hirschman did in the phase space case \cite{Hirschman}.

For measurement uncertainty the general situation is not so favourable, perhaps due to the much more recent introduction of the subject. At this point we know of no efficient way to derive sharp bounds for generic pairs of observables. Nevertheless, we were able to treat the case of a joint measurement of all components in arbitrary directions, because in this case rotational symmetry is not broken and leads to considerable simplification. One of these simplifications is the observation that the two basic error criteria, namely metric uncertainty and calibration error lead to the same results. This was already familiar from the phase space case. However, a further simplification one might have expected from this analogy definitely does not hold: There seems to be no quantitative link between preparation and measurement uncertainty for angular momentum. Further research will show whether useful general connections between the two faces of the uncertainty coin can be established.

The limit large $s\to\infty$ can be understood as a mean field limit \cite{RW89}, when the spin-$s$ representation is  considered as $2s$ copies of a spin-$1/2$ system in a symmetric state. We can also see this as a classical limit $\hbar\to0$ \cite{WWolff}, in the sense that the angular momentum in physical units, i.e., $\hbar$ is fixed, and hence the dimensionless half-integral representation parameter $s$ has to diverge. This offers a way to treat not just the uncertainty aspects of this limit, but also the limit of the whole theory of angular momentum.

% % % Acknowledgements % % %

\section*{Acknowledgements}
The authors thank Ashley Milsted, Ciara Morgan and David Reeb for critically reading our manuscript.
We also acknowledge the financial support from the ERC grant DQSIM, RTG1991 funded by the DFG and the collaborative research project Q.com-Q funded by the BMBF.

We acknowledge support by Deutsche Forschungsgemeinschaft and the Open Access Publishing Fund of Leibniz Universität Hannover.

%%% Bibliography %%%
\bibliography{Uram}

%%% Cut-out model %%%
\newpage
\quad
\thispagestyle{empty}
\newpage
\thispagestyle{empty}
\includegraphics[width=\textwidth]{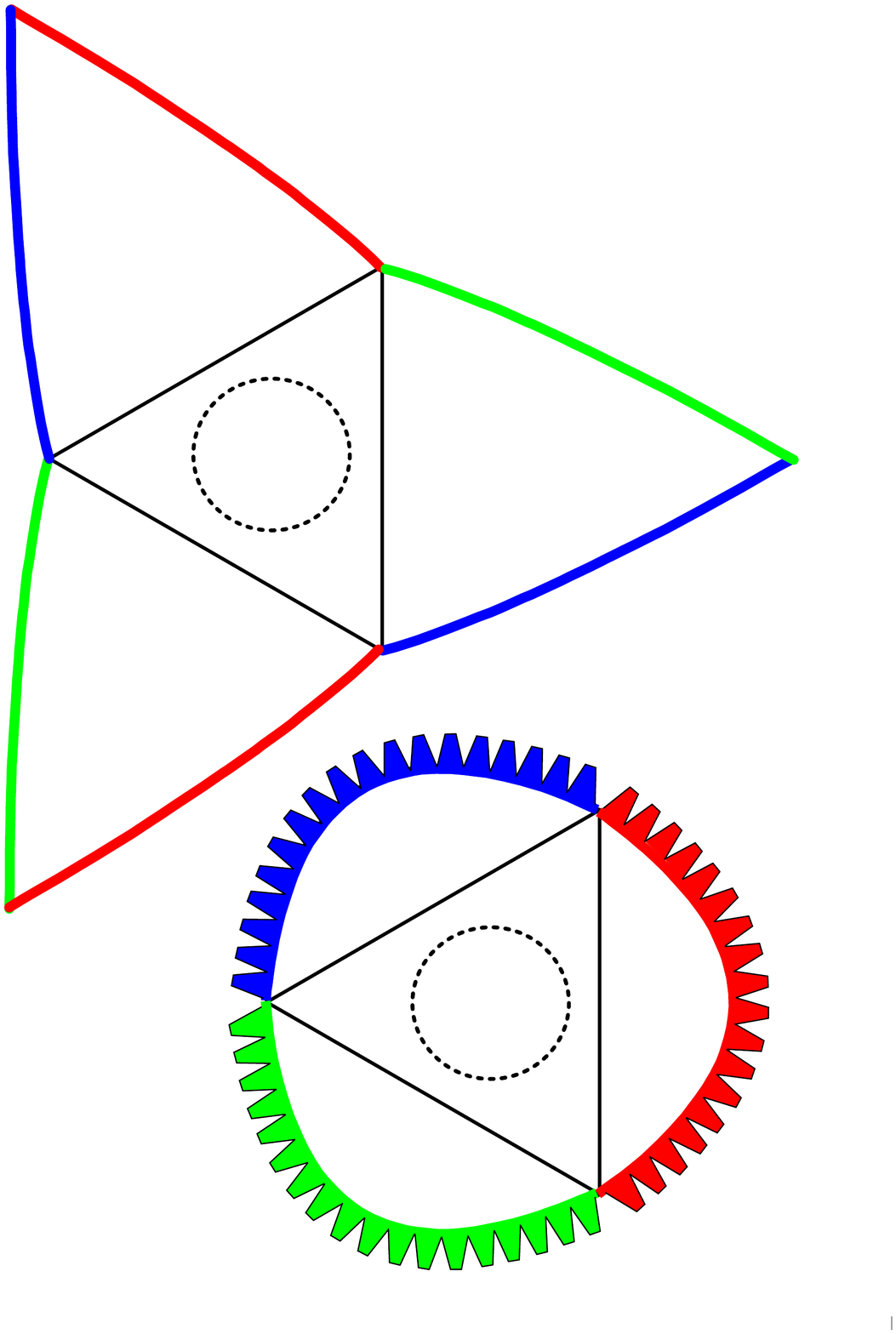}
\newpage
\thispagestyle{empty}
\quad
\newpage
\thispagestyle{empty}
\includegraphics[width=\textwidth]{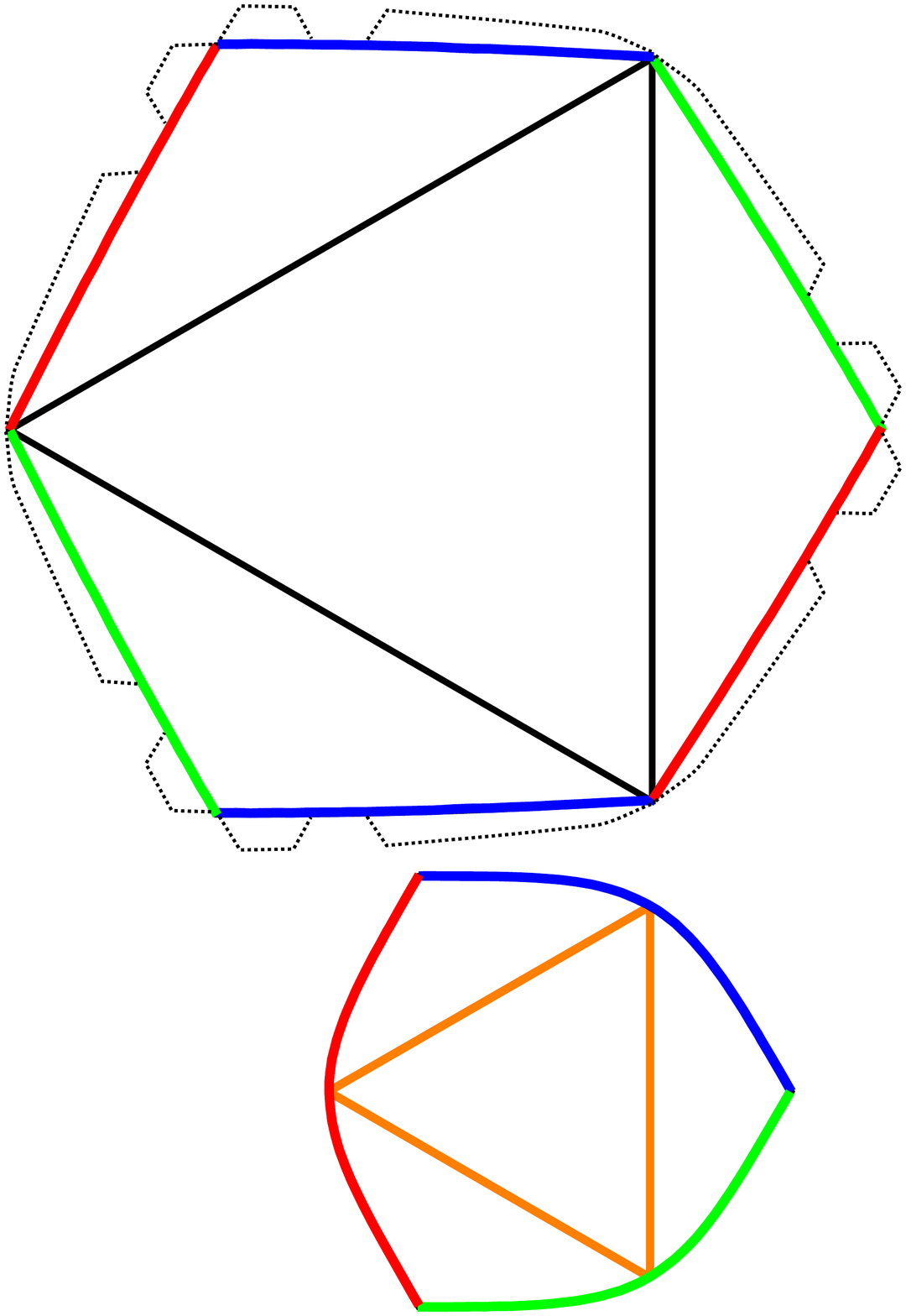}
\thispagestyle{empty}

\end{document}